\documentclass[10pt,reqno]{amsart}
\usepackage{latexsym,amsmath}
\usepackage{amsfonts}
\usepackage{amssymb}
\usepackage{fixmath}

\usepackage{fullpage}
\usepackage{graphicx}
\usepackage[colorlinks=true,citecolor=blue]{hyperref}
\usepackage[dvipsnames]{xcolor}
\usepackage{enumerate}

\usepackage[T1]{fontenc}
\usepackage{fourier}
\usepackage{bbm}

\usepackage[boxed, linesnumbered, noend, noline]{algorithm2e}
\SetKwInput{KwData}{Input}
\SetKwInput{KwResult}{Output}

\numberwithin{equation}{section}

\newcommand\mypsi{\mathfrak{t}}
\newcommand\myphi{\mathfrak{l}}
\newcommand\UCP{\texttt{PUC}}
\newcommand\PM{\{\pm1\}}
\newcommand\ism{\cong}
\newcommand\uni{\mathtt{u}}
\newcommand\luni{\mathtt{n}}
\newcommand\sign{\mathrm{sign}}

\newcommand\disteq{\stackrel{\mathrm{dist}}{=}}

\newcommand\MU{\vec\mu}
\newcommand{\BP}{\mathrm{BP}}
\newcommand{\vDelta}{\vec\Delta}
\newcommand{\va}{\vec a}

\newcommand{\vC}{\vec C}

\newcommand{\vd}{\vec d}
\newcommand{\vD}{\vec D}

\newcommand{\vT}{\vec T}
\newcommand{\vN}{\vec N}
\newcommand{\vR}{\vec N}

\newcommand{\vs}{\vec s}

\newcommand{\vX}{\vec X}
\newcommand{\vl}{\vec l}

\renewcommand{\epsilon}{\eps}
\renewcommand{\subset}{\subseteq}

\newcommand\vr{\vec r}

\newcommand\PHI{\vec\Phi}
\newcommand\hPHI{\hat{\PHI}}

\newcommand\nix{\,\cdot\,}
\newcommand\vV{\vec V}

\newcommand\vS{\vec S}

\newcommand\dd{{\mathrm d}}

\renewcommand{\vec}[1]{\boldsymbol{#1}}

\newcommand\SIGMA{\vec\sigma}

\newtheorem{definition}{Definition}[section]
\newtheorem{claim}[definition]{Claim}

\newtheorem{theorem}[definition]{Theorem}
\newtheorem{lemma}[definition]{Lemma}
\newtheorem{proposition}[definition]{Proposition}
\newtheorem{corollary}[definition]{Corollary}

\newtheorem{fact}[definition]{Fact}

\newcommand\fQ{\mathfrak{Q}}

\newcommand\fC{\mathfrak{C}}

\newcommand\fA{\mathfrak{A}}
\newcommand\fE{\mathfrak{E}}
\newcommand\fH{\mathfrak{H}}
\newcommand\fR{\mathfrak{R}}
\newcommand\fF{\mathfrak{F}}

\newcommand\cB{\mathcal{B}}
\newcommand\cC{\mathcal{C}}

\newcommand\cU{\mathcal{U}}
\newcommand\cN{\mathcal{N}}

\newcommand\cL{\mathcal{L}}

\newcommand\cP{\mathcal{P}}

\newcommand\cV{\mathcal{V}}
\newcommand\cW{\mathcal{W}}

\def\cR{{\mathcal R}}

\newcommand\vv{\vec v}

\newcommand\vZ{\vec Z}

\newcommand\eul{\mathrm{e}}
\newcommand\eps{\varepsilon}

\newcommand\Var{\mathrm{Var}}
\newcommand\Erw{\mathbb{E}}
\newcommand{\vecone}{\mathbb{1}}

\newcommand{\Po}{{\rm Po}}
\newcommand{\Bin}{{\rm Bin}}

\newcommand\bc[1]{\left({#1}\right)}
\newcommand\cbc[1]{\left\{{#1}\right\}}
\newcommand\bcfr[2]{\bc{\frac{#1}{#2}}}

\newcommand\brk[1]{\left\lbrack{#1}\right\rbrack}

\newcommand\norm[1]{\left\|{#1}\right\|}
\newcommand\abs[1]{\left|{#1}\right|}

\newcommand\RR{\mathbb{R}}

\def\?#1{}
\makeatletter
\def\whp{w.h.p\@ifnextchar-{.}{\@ifnextchar.{.\?}{\@ifnextchar,{.}{\@ifnextchar){.}{\@ifnextchar:{.:\?}{.\ }}}}}}
\makeatother

\makeatletter
\def\Whp{W.h.p\@ifnextchar-{.}{\@ifnextchar.{.\?}{\@ifnextchar,{.}{\@ifnextchar){.}{\@ifnextchar:{.:\?}{.\ }}}}}}
\makeatother

\newcommand{\tensor}{\otimes}

\newcommand{\Bollobas}{Bollob\'as}
\newcommand{\Chvatal}{Chv\'{a}tal}

\newcommand\pr{\mathbb{P}}

\newcommand\Lem{Lemma}
\newcommand\Prop{Proposition}
\newcommand\Thm{Theorem}

\newcommand\Cor{Corollary}
\newcommand\Sec{Section}

\newcommand*{\dif}{\mathop{}\!\mathrm{d}}
\newcommand{\LLN}{\mathtt{logBP}^\tensor_{d,t}} % LL new

\def\ex{{\mathbb E}}
\def\pr{{\mathbb P}}

\newcommand\RSA{Random Structures and Algorithms}

\newcommand\CPC{Combinatorics, Probability and Computing}

\begin{document}

	\title{The number of random 2-SAT solutions is asymptotically log-normal}

	%\thanks{Amin Coja-Oghlan's research is supported by DFG CO 646/3, DFG CO 646/5 and DFG CO 646/6.  Pavel Zakharov's research is supported by DFG CO 646/6.  Haodong Zhu's research is supported by the European Union's Horizon 2020 research and innovation programme under the Marie Sk\l odowska-Curie grant agreement no.~945045, and by the NWO Gravitation project NETWORKS under grant no.~024.002.003.  Noela M\"uller's research is supported by the NWO Gravitation project NETWORKS under grant no.~024.002.003.}

	\author{Arnab Chatterjee, Amin Coja-Oghlan, No\"ela M\"uller, Connor Riddlesden, Maurice Rolvien, Pavel~Zakharov, Haodong Zhu}
	\address{Arnab Chatterjee, {\tt arnab.chatterjee@tu-dortmund.de}, TU Dortmund, Faculty of Computer Science, 12 Otto-Hahn-St, Dortmund 44227, Germany.}
	\address{Amin Coja-Oghlan, {\tt amin.coja-oghlan@tu-dortmund.de}, TU Dortmund, Faculty of Computer Science and Faculty of Mathematics, 12 Otto-Hahn-St, Dortmund 44227, Germany.}
	\address{Noela M\"uller, {\tt n.s.muller@tue.nl}, Eindhoven University of Technology, Department of Mathematics and Computer Science, MetaForum MF 4.084, 5600 MB Eindhoven, the Netherlands.}
	\address{Connor Riddlesden, {\tt c.d.riddlesden@tue.nl}, Eindhoven University of Technology, Department of Mathematics and Computer Science, MetaForum MF 4.084, 5600 MB Eindhoven, the Netherlands.}
	\address{Maurice Rolvien, {\tt maurice.rolvien@tu-dortmund.de}, TU Dortmund, Faculty of Computer Science, 12 Otto-Hahn-St, Dortmund 44227, Germany.}
	\address{Pavel Zakharov, {\tt pavel.zakharov@tu-dortmund.de}, TU Dortmund, Faculty of Computer Science and Faculty of Mathematics, 12 Otto-Hahn-St, Dortmund 44227, Germany.}
	\address{Haodong Zhu, {\tt h.zhu1@tue.nl}, Eindhoven University of Technology, Department of Mathematics and Computer Science, 5600 MB Eindhoven, the Netherlands.}

	\begin{abstract}%
		We prove that throughout the satisfiable phase, the logarithm of the number of satisfying assignments of a random 2-SAT formula satisfies a central limit theorem.
		This implies that the log of the number of satisfying assignments exhibits fluctuations of order $\sqrt n$, with $n$ the number of variables.
		The formula for the variance can be evaluated effectively.
		By contrast, for numerous other random constraint satisfaction problems the typical fluctuations of the logarithm of the number of solutions are {\em bounded} throughout all or most of the satisfiable regime.
		\hfill {\em MSc:~05C80, 60C05, 68Q87}
	\end{abstract}

	\maketitle

	\section{Introduction}\label{sec_intro}

	\subsection{Background and motivation}\label{sec_motiv}
	The quest for satisfiability thresholds has been a guiding theme of research into random constraint satisfaction problems~\cite{ANP,Cheeseman,DSS3}.
	But once the satisfiability threshold has been pinpointed a question of no less consequence is to determine the distribution of the number of satisfying assignments within the satisfiable phase~\cite{pnas}.
	Indeed, the number of solutions is intimately tied to phase transitions that affect the geometry of the solution space, which in turn impacts the computational nature of finding or sampling solutions~\cite{Barriers,BreslerHuang,Charis}.
	However, few tools are currently available to count solutions of random problems.
	Where precise rigorous results exist (such as in random NAESAT or XORSAT), the proofs typically rely on the method of moments (e.g.,~\cite{nae,DuboisMandler,PittelSorkin,Feli2}).
	Yet a necessary condition for the success of this approach is that the problem in question exhibits certain symmetries, which are absent in many interesting cases~\cite{ANP,CKM}.

	The aim of the present paper is to shed a closer light on the number of satisfying assignments in random 2-SAT, the simplest random CSP that lacks said symmetry properties.
	While the random 2-SAT satisfiability threshold has been known since the 1990s~\cite{CR,Goerdt}, a first-order approximation to the number of satisfying assignments has been obtained only recently~\cite{2sat}.
	This timeline reflects the computational complexity of the respective questions.
	As is well known, deciding the satisfiability of a 2-CNF reduces to directed reachability, solvable in polynomial time~\cite{APT79}.
	%Accordingly, independent derivations of the satisfiability threshold by Goerdt~\cite{Goerdt} and Chavatal and Reed~\cite{CR} boil down to percolation arguments.

	By contrast, calculating the number of satisfying assignmets $Z(\Phi)$ of a 2-CNF $\Phi$ is a $\#$P-hard task~\cite{Valiant}.
	Nonetheless, Monasson and Zecchina~\cite{MZ} put forward a delicate physics-inspired conjecture as to the exponential order of the number of satisfying assignments of random 2-CNFs.
	Achlioptas et al.~\cite{2sat} recently proved this conjecture.
	Their theorem provides a first-order, law-of-large-numbers approximation of the logarithm of the number of satisfying assignments.
	The present paper contributes a much more precise result, namely a central limit theorem.
	We show that throughout the satisfiable phase the logarithm of the number of satisfying assignments, suitably shifted and scaled, converges to a Gaussian.
	%The variance of the limiting distribution can be computed effectively within any accuracy.
	This is the first central limit theorem of this type for any random CSP.

	Let $\PHI=\PHI_{n,m}$ be a random 2-CNF on $n$ Boolean variables $x_1,\ldots,x_n$ with $m$ clauses, drawn independently and uniformly from all $4\binom n2$ possible $2$-clauses.
	Suppose that $m\sim dn/2$ for a fixed real $d>0$.
	Thus, $d$ gauges the average number of clauses in which a variable $x_i$ appears.
	The value $d=2$ marks the satisfiability threshold; hence, $\PHI$ is satisfiable with high probability (`w.h.p.') if $d<2$, and unsatisfiable \whp\ if $d>2$~\cite{CR,Goerdt}.
	Achlioptas et al.~\cite{2sat} determined a function $\phi(d)>0$ such that for all $d<2$, i.e., throughout the entire satisfiable phase we have
	\begin{align}\label{eqlaregenum}
		Z(\PHI)&=\exp(n\phi(d)+o(n))&&\mbox{\whp},
		%	\lim_{n\to\infty}\frac1n\log Z(\PHI)&=\phi(d)>0&&\mbox{ in probability}.
	\end{align}
	%The precise formula for the value $\phi(d)$ is mildly complicated and involves a stochastic fixed point problem, but it matches the physics prediction from~\cite{MZ}.
	thereby determining the leading exponential order of $Z(\PHI)$.

	However, \eqref{eqlaregenum} fails to identify the limiting distribution of $Z(\PHI)$.
	%For all \eqref{eqlaregenum} reveals, the (multiplicative) fluctuations could either be bounded, or perhaps as large as $exp(o(n))$.
	To be precise, since \eqref{eqlaregenum} shows that $Z(\PHI)$ scales exponentially, we expect this random variable to exhibit {\em multiplicative} fluctuations.
	Therefore, the appropriate goal is to find the limiting distribution of the logarithm of this random variable, i.e., of $\log Z(\PHI)$.
	Indeed, physics intuition suggests that $\log Z(\PHI)$ should be asymptotically Gaussian~\cite{MM}.
	The main result of the present paper confirms this hunch.
	Specifically, letting $\vec\Gamma_{\eta(d)}$ be a Gaussian with mean $0$ and standard deviation $\eta(d)>0$, we prove that for all $0<d<2$, $\log Z(\PHI)$ satisfies
	\begin{align}\label{eqfluct}
		%	Z(\PHI)\disteq\exp\bc{n\phi(d)+\sqrt{dn/2}\cdot\vec\Gamma_d+o(\sqrt n)}.
		\pr\brk{\log Z(\PHI)-\ex[\log Z(\PHI)\mid Z(\PHI)>0]<z\sqrt m}\sim\pr\brk{\vec\Gamma_{\eta(d)}<z}&&(z\in\RR).
		%\frac{1}{\sqrt{2 \pi}\,\eta(d)}\int_a^b\exp\bc{-\frac{z^2}{2\eta(d)^2}}\dd z&&(a<b).
	\end{align}
	%The variance $\eta(d)$ of the limiting distribution can be approximated numerically within any accuracy.
	%We will state the precise main result, including the formula for $\eta(d)$, momentarily.

	The order $\Theta(\sqrt n)$ of fluctuations confirmed by \eqref{eqfluct} sets random 2-SAT apart from a large family of other random constraint satisfaction problems.
	For example, for random graph $q$-colouring with $q\geq3$ colours the log of the number of $q$-colourings {\em superconcentrates}, i.e., merely has \emph{bounded} fluctuations throughout most of the regime where the random graph is $q$-colourable~\cite{silent}.%
	\footnote{Formally, up to the so-called condensation threshold, which precedes the $q$-colourabiliy threshold by a small additive constant, the logarithm of the number of $q$-colurings minus its expectation converges in distribution to a random variable with bounded moments~\cite{silent,cond,CKM}.}
	The same is true of random NAESAT, XORSAT and the symmetric perceptron~\cite{ALS,Ayre,CKM,PittelSorkin}.
	In each of these cases, certain fundamental symmetry properties (e.g., that the set of $q$-colourings remains invariant under permutations of the colours) enable the computation of the number of solutions via the method of moments.
	Random 2-SAT lacks the respective symmetry (as the set of satisfying assignments is not generally invariant under swapping `true' and `false'), and accordingly \eqref{eqfluct} establishes that the number of solutions fails to superconcentrate (for more details see \cite{CKM}).

	\subsection{The main result}\label{sec_results}
	%To state the main result precisely we need to detail the formula for the variance $\eta(d)$.
	%Like the expression $\phi(d)$ for the mean of $\log Z(\PHI)$ from \eqref{eqlaregenum}
	The formula for the standard deviation $\eta(d)$ from~\eqref{eqfluct} comes in terms of a fixed point equation on a space of probability measures.
	Thus, let $\cP(\RR^2)$ be the set of all (Borel) probability measures on $\RR^2$.
	For $0<d<2$ and $0\leq t\leq1$ we define an operator
	\begin{align}\label{eqlogBPtensor}
		\LLN:&\cP\bc{\RR^2} \to \cP\bc{\RR^2},& \rho\mapsto\hat \rho = \LLN(\rho),
	\end{align}
	as follows.
	Let
	\begin{align*}
		(\vec \xi_{\rho,i})_{i \geq 1},\,(\vec\xi'_{\rho,i})_{i\geq 1},\,(\vec\xi''_{\rho,i})_{i\geq 1},&&
		\vec\xi_{\rho,i}=\binom{\vec\xi_{\rho,i,1}}{\vec\xi_{\rho,i,2}},\,
		\vec\xi_{\rho,i}'=\binom{\vec\xi_{\rho,i,1}'}{\vec\xi_{\rho,i,2}'},\,
		\vec\xi_{\rho,i}''=\binom{\vec\xi_{\rho,i,1}''}{\vec\xi_{\rho,i,2}''}
	\end{align*}
	be random vectors with distribution $\rho$, let $\vd \disteq \Po(td)$, $\vec d', \vec d''\disteq \Po((1-t)d)$ and let $\vs_i,\vs_i',\vs_i'',\vr_i,\vr_i',\vr_i''$ for $i\geq1$ be uniformly random on $\PM$, all mutually independent.
	Then $\hat \rho$ is the distribution of the vector
	\begin{align*}
		\begin{pmatrix}
			\sum_{i=1}^{\vd} \vec s_i \log\bc{\frac12\bc{1+\vec r_i \tanh(\vec \xi_{\rho,i,1}/2)}} +  \sum_{i=1}^{\vec d'} \vec s_i' \log\bc{\frac12\bc{1+\vec r'_i \tanh(\vec\xi_{\rho,i,1}'/2)}}\\
			\sum_{i=1}^{\vd} \vec s_i \log\bc{\frac12\bc{1+\vec r_i \tanh(\vec \xi_{\rho,i,2}/2)}} +  \sum_{i=1}^{\vec d''} \vec s_i'' \log\bc{\frac12\bc{1+\vec r_i'' \tanh(\vec\xi''_{\rho,i,2}/2)}}
		\end{pmatrix}\in\RR^2\enspace.
	\end{align*}
	In addition, define a function $\cB_{d,t}^\tensor:\cP(\RR^2)\to(0,\infty]$ by letting
	\begin{align}
		\cB^\tensor_{d,t}(\rho)&=
		\ex\brk{\prod_{h=1}^2\log\bc{1-\frac14{(1+\vr_{1}\tanh(\vec\xi_{\rho,1,h}/2))(1+\vr_{2}\tanh(\vec\xi_{\rho,2,h}/2))}}}.
		\label{eqfunctional}
	\end{align}
	%Finally, let $\cN(\chi,\eta^2)$ denote the Gaussian distribution with mean $\chi$ and variance $\eta^2$.

	\begin{figure}
		\includegraphics[height=60mm]{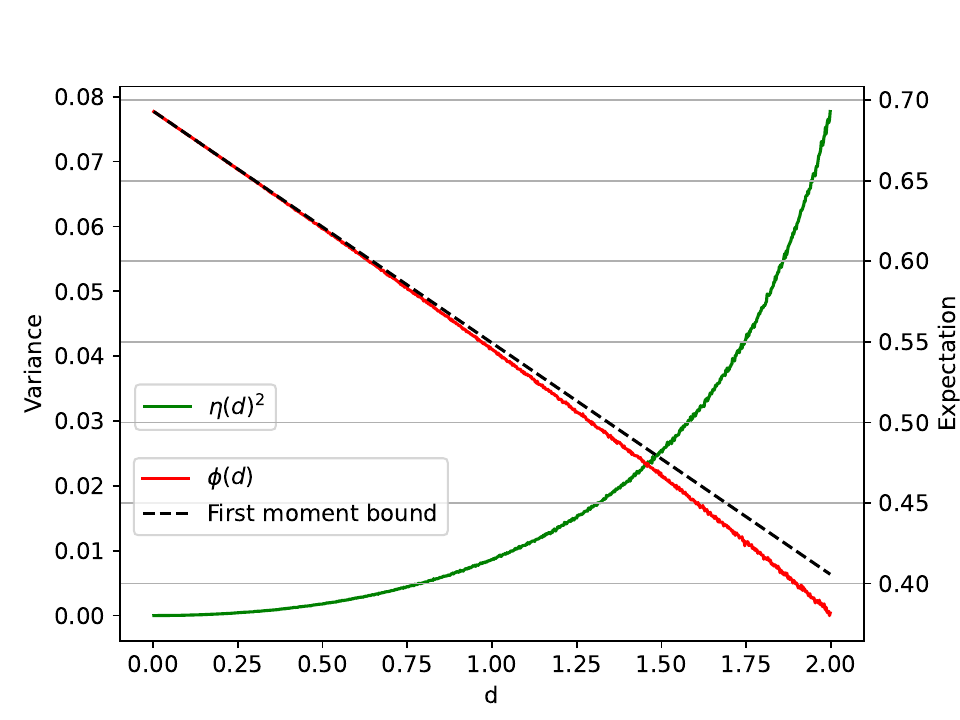}\hfill
		\includegraphics[height=55mm]{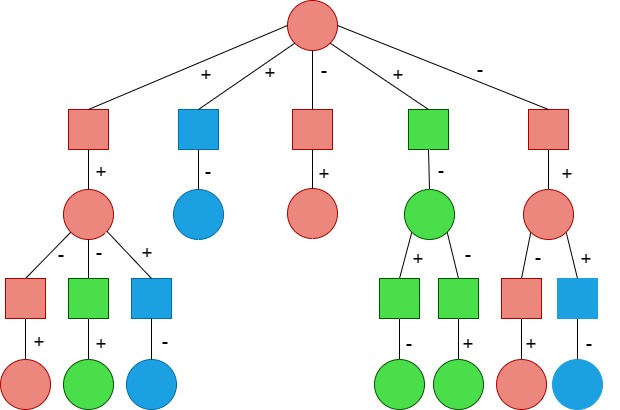}
		\caption{{\em Left:} Numerical approximations to the function $\phi(d)$ from \eqref{eqlaregenum} (red) and the variance $\eta(d)^2$ from \eqref{eqthm_clt} (green). The black dashed line is the first moment bound $d\mapsto\log(2)+\frac d2\log(3/4)$. {\em Right:} An illustration of the tree $\vT^\tensor$ from \Sec~\ref{sec_lwc}.}\label{fig_var}
	\end{figure}

	\begin{theorem}\label{thm_clt}
		For any $0<d<2$, $t\in[0,1]$ there exists a unique probability measure $\rho_{d,t}\in\cP(\RR^2)$ such that
		\begin{align}\label{eqfixedp}
			\rho_{d,t}&=\LLN(\rho_{d,t})&&\mbox{and}&&\int_{\RR^2}\|\xi\|_2^2\dd\rho_{d,t}(\xi)<\infty.
		\end{align}
		%\pav{Why should $\xi$ be bold here? It's not really random in this context.} Should not be bold.
		Furthermore,
		\begin{align}\label{eqthm_clt_lim}
			\lim_{n\to\infty}\frac{\log Z(\PHI)-\ex[\log Z(\PHI)\mid Z(\PHI)>0]}{\sqrt m}&=\vec\Gamma_{\eta(d)}\qquad\mbox{in distribution, where}\\
			\eta(d)^2&=\int_0^1\cB_{d,t}^\tensor(\rho_{d,t})\dd t-\cB_{d,0}^\tensor(\rho_{d,0})\in(0,\infty).\label{eqthm_clt}
		\end{align}
	\end{theorem}

	The conditioning on $\log Z(\PHI)>0$ is necessary in \eqref{eqthm_clt_lim}, because even for $d<2$ the formula $\PHI$ is unsatisfiable with probability $\Omega(n^{-1})$, in which case $\log Z(\PHI)=-\infty$.
	Moreover, the $L^2$-bound from~\eqref{eqfixedp} ensures that the integral \eqref{eqthm_clt} is well-defined.
	Finally, \eqref{eqthm_clt_lim} implies \eqref{eqfluct}.

	How can the formula \eqref{eqthm_clt} be evaluated?
	Because the proof of the uniqueness of the stochastic fixed point $\rho_{d,t}$ from \eqref{eqfixedp} is based on the contraction method, a fixed point iteration will converge rapidly.
	In effect, for any $d,t$ a discrete distribution that approximates $\rho_{d,t}$ arbitrarily well (in Wasserstein distance) can be computed via a randomised algorithm called {\em population dynamics}~\cite[Chapter~ 14]{MM}.
	Since $\cB^\tensor_{d,t}(\rho_{d,t})$ varies continously in $d$ and $t$, $\eta(d)^2$ can thus be approximated within any desired accuracy, see Figure~\ref{fig_var}.

	\section{Proof strategy}\label{sec_outline}

	\noindent
	The main challenge towards the proof of \Thm~\ref{thm_clt} is to get a handle on the variance of $\log Z(\PHI)$ given satisfiability.
	%Hence, we will begin by showing that the formula for the variance is given by \eqref{eqthm_clt}.
	The key idea, inspired by spin glass theory~\cite{ChenDeyPanchenko} but novel to random constraint satisfaction, is to count the {\em joint} number of satisfying assignments of two correlated random formulas.
	Once this is accomplished \Thm~\ref{thm_clt} will follow from the careful application of a general martingale central limit theorem.
	To get acclimatised we first revisit the method of moments, the reasons it fails on random 2-SAT and the combinatorial interpretation of the law of large numbers \eqref{eqlaregenum}.

	%To facilitate this construction we will need to set up a suitable clause exposure martingale and employ a pruning process based on the Unit Clause Propagation algorithm to deal with tail events.
	%Additionally, we will set up a Galton-Watson process to reflect the local structure of the two correlated formulas.
	%The formula for the variance will then ultimately result from an analysis of the Belief Propagation message passing algorithm for estimating the number of satisfying assignments on this Galton-Watson tree.
	%Specifically, we will see that the operator $\BP_{d,t}^\tensor$ from \eqref{eqBPtensor} describes a single iteration of the Belief Propagation operator on the Galton-Watson tree.
	%Finally, in order to obtain the central limit theorem we will combine the techniques from the proof of the variance formula with a martingale central limit theorem.

	\subsection{The method of moments fails}\label{sec_annealed}
	The default approach to estimating the number of solutions to a random CSP is the venerable second moment method~\cite{ANP}.
	Its thrust is to show that the second moment of the number of solutions is of the same order as the square of the expected number of solutions.
	If so then the moment computation together with small subgraph conditioning yields the precise limiting distribution of the number of solutions~\cite{COW,RW}.
	However, this approach works only if the log of the number of solutions superconcentrates around the log of the expected number of solutions.

	%Successful applications include random graph colouring, NAESAT and XORSAT~\cite{Ayre,silent,PittelSorkin}.

	%A necessary condition for the success of this method is that the median number of solutions is of the same order of magnitude as the expected number of solutions \whp~\cite{CKM}.
	This necessary condition is not satisfied in random 2-SAT.
	In fact, a straightforward calculation yields
	\begin{align}\label{eqannealed}
		\frac1n\log\ex[Z(\PHI)]&\sim\log 2+\frac d2\log(3/4).
	\end{align}
	The formula on the r.h.s.\ is displayed as the black dashed line in Figure~\ref{fig_var}.
	As can be verified analytically, this line strictly exceeds the function $\phi(d)$ from \eqref{eqlaregenum} for any $0<d<2$.
	Consequently, \eqref{eqlaregenum} implies that %$Z(\PHI)\leq\exp(-\Omega(n))\ex[Z(\PHI)]$ \whp\
	$\log Z(\PHI)\leq\log\ex[Z(\PHI)]-\Omega(n)$ \whp\
	In other words, the expected number of solutions $\ex[Z(\PHI)]$ overshoots the typical number of solutions by an exponential factor \whp; cf.\ the discussion in~\cite{nae,yuval}.

	\subsection{Belief Propagation}\label{sec_BPop}
	Instead of the method of moments, the prescription of the physics-based work of Monasson and Zecchina~\cite{MZ} is to estimate $\log Z(\PHI)$ by way of the Belief Propagation (BP) message passing algorithm.
	This approach was vindicated rigorously by Achlioptas et al.~\cite{2sat}.

	As we will reuse certain elements of that analysis we dwell on BP briefly.
	For a clause $a$ of a 2-CNF $\Phi$ let $\partial a=\partial_\Phi a$ be the set of variables that $a$ contains.
	Moreover, for $x\in\partial a$ let $\sign_\Phi(x,a)=\sign(x,a)\in\PM$ be the sign with which $x$ appears in $a$.
	Analogously, let $\partial x=\partial_\Phi x$ be the set of clauses in which variable $x$ appears.
	BP introduces `messages' between clauses $a$ and the variables $x\in\partial a$.
	More precisely, each such clause-variable pair $a,x$ comes with two messages $\mu_{x\to a},\mu_{a\to x}$.
	The messages are probability distributions on `true' and `false', which we represent by $\pm1$.
	Thus, $\mu_{x\to a}(\pm1),\mu_{a\to x}(\pm1)\geq0$ and $\mu_{x\to a}(1)+\mu_{x\to a}(-1)=\mu_{a\to x}(1)+\mu_{a\to x}(-1)=1$.

	The messages get updated iteratively by an operator
	\begin{align}\label{eqBPop}
		\BP&:(\mu_{x\to a},\mu_{a\to x})_{a,x\in\partial a}\mapsto(\hat\mu_{x\to a},\hat\mu_{a\to x})_{a,x\in\partial a}=\BP((\mu_{x\to a},\mu_{a\to x})_{a,x\in\partial a}).
	\end{align}
	For a clause $a$ with adjacent variables $\partial a=\{x,y\}$ the updated messages $\hat\mu_{a\to x}(\pm1)$ are defined by
	\begin{align}\label{eqBPop1}
		\hat\mu_{a\to x}(\sign(x,a))&=\frac1{1+\mu_{y\to a}(\sign(y,a))},&
		\hat\mu_{a\to x}(-\sign(x,a))&=\frac{\mu_{y\to a}(\sign(y,a))}{1+\mu_{y\to a}(\sign(y,a))}.
	\end{align}
	Moreover, for a variable $x$ and a clause $a\in\partial x$ we define%
	\footnote{For the sake of tidyness, if the above denominator vanishes we simply let $\hat\mu_{x\to a}(\pm1)=\frac12$.}
	\begin{align}\label{eqBPop2}
		\hat\mu_{x\to a}(s)&=\frac{\prod_{b\in\partial x\setminus\{a\}}\mu_{b\to x}(s)}{\prod_{b\in\partial x\setminus\{a\}}\mu_{b\to x}(1)+\prod_{b\in\partial x\setminus\{a\}}\mu_{b\to x}(-1)}&&(s\in\PM)\enspace.
	\end{align}
	The purpose of BP is to heuristically `approximate' the marginal probabilities that a random satisfying assignment $\SIGMA=\SIGMA_\Phi$ of $\Phi$ will set a certain variable to a specific truth value.
	The `approximation' given by the set $(\mu_{x\to a},\mu_{a\to x})_{a,x\in\partial a}$ of messages reads
	\begin{align}\label{eqBPmargs}
		\mu_x(s)=\frac{\prod_{b\in\partial x}\mu_{b\to x}(s)}{\prod_{b\in\partial x}\mu_{b\to x}(1)+\prod_{b\in\partial x}\mu_{b\to x}(-1)}&&(s\in\PM).
	\end{align}

	The BP `ansatz' now asks that we iterate the $\BP$ operator until an (approximate) fixed point is reached, i.e., ideally until $\hat\mu_{a\to x}=\mu_{a\to x}$ and $\hat\mu_{x\to a}=\mu_{x\to a}$ for all $a,x$.
	Then we evaluate the BP marginals \eqref{eqBPmargs} and plug them into a generic formula called the {\em Bethe free entropy}, which yields the BP `approximation' of $\log Z(\Phi)$; an excellent exposition can be found in~\cite{MM}.
	The BP recipe provably yields the correct result if the bipartite graph induced by the clause-variable incidences of the 2-CNF $\Phi$ is acyclic, but may be totally off otherwise.

	Of course, for $1<d<2$ the bipartite graph associated with the random formula $\PHI$ contains cycles in abundance.
	Nonetheless, \eqref{eqlaregenum} confirms that the BP formula provides a valid approximation to within $o(n)$.
	The proof is based on two observations.
	First, that the local structure of the clause-variable incidence graph can be described by a Galton-Watson tree.
	Second, that the Galton-Watson tree enjoys a spatial mixing property called {\em Gibbs uniqueness}.

	Since the proof of \Thm~\ref{thm_clt} also harnesses Gibbs uniqueness, let us elaborate.
	To mimic the local structure of $\PHI$ consider a multitype Galton-Watson tree $\vT$ whose types are {\em variable nodes} and {\em clause nodes} of four sub-types $(s,s')$ with $s,s'\in\PM$.
	The root $o$ is a variable node.
	The offspring of any variable node is a $\Po(d/4)$ number of clause nodes of each of the four sub-types.
	Finally, the offspring of a clause node is a single variable node.
	The clause type $(s,s')$ indicates that $s$ is the sign with which the parent variable appears in the clause, while $s'$ determines the sign of the child variable.
	Thus, the Galton-Watson tree $\vT$ can be viewed as a (possibly infinite) 2-CNF.
	For an integer $\ell\geq0$ let $\vT^{(2\ell)}$ be the finite tree/2-CNF obtained by deleting all variables and clauses at a distance larger than $2\ell$ from the root.

	The tree $\vT$ approximates $\PHI$ locally in the sense that for any fixed $\ell$ and any given variable $x_i$ the distribution of the depth-$2\ell$ neighbourhood of $x_i$ in $\PHI$ converges to $\vT^{(2\ell)}$ as $n\to\infty$ (in the sense of local weak convergence).
	Moreover, Gibbs uniqueness posits that under random satisfying assignments of the tree-CNF $\vT^{(2\ell)}$ the truth value $\SIGMA_{o}$ of the root under a random satisfying assignment $\SIGMA$ decouples from the values $\SIGMA_{\vT,y}$ of variables $y\in\partial^{2\ell}o$ at distance precisely $2\ell$ from $o$ for large enough $\ell$.
	Formally, with $S(\vT^{(2\ell)})$ the set of satisfying assignments of the 2-CNF $\vT^{(2\ell)}$, the following is true.

	\begin{proposition}[{\cite[\Prop~2.2]{2sat}}]\label{prop_uniqueness_old}
		We have
		\begin{align}\label{eqProp_uniqueness1}
			\lim_{\ell\to\infty}\Erw\brk{\max_{\tau\in S(\vT^{(2\ell)})}\abs{\pr\brk{\SIGMA_o=1\mid\vT^{(2\ell)},\SIGMA_{\partial^{2\ell}o}=\tau_{\partial^{2\ell}o}}-\pr\brk{\SIGMA_o=1\mid\vT^{(2\ell)}}}}&=0.
		\end{align}
	\end{proposition}

	%The expression $\phi(d)$ is ultimately obtained by applying a generic functional called the `Bethe free entropy' to the limiting distribution of the marginals \eqref{eqBPmargs}.

	%to show that
	%the true marginal probabilites $\mu_{\PHI}(\{\SIGMA_x=s\})$ are well approximated by BP.
	%Furthermore, the empirical distributions of these marginals is well approximated by the distribution of the mariginal $\mu_{\vT^{(2\ell)}}(\{\SIGMA_o=s\})$ of the root the tree $\vT^{(2\ell)}$ for large $\ell$.
	%Furthermore, in the limit $\ell\to\infty$ the distribution of the tree marginals $\mu_{\vT^{(2\ell)}}(\{\SIGMA_o=s\})$ converges to the solution of a stochastic fixed point equation.
	%Finally, the formula $\phi(d)$ in~\eqref{eqlaregenum} results by evaluating the Bethe free entropy on this fixed point distribution.
	%The proof that this formula approximates $\log Z(\PHI)$ within $o(n)$ is hinges on a coupling argument from spin glass theory known as the Aizenman-Sims-Starr scheme.

	\subsection{Approaching the variance}\label{sec_comb}
	The proof of the formula~\eqref{eqlaregenum} combines the Gibbs uniqueness property and the local convergence to the Galton-Watson tree with a coupling argument called the `Aizenman-Sims-Starr scheme'~\cite{2sat}.
	Unfortunately, this combination does not seem precise enough to get a handle on the limiting distribution of $\log Z(\PHI)$ by a long shot.
	Actually, it is anything but clear how even the {\em order} of the standard deviation of $\log Z(\PHI)$ could be derived along these lines.
	One specific problem is that the rate of convergence of \eqref{eqProp_uniqueness1} diminishes as $d$ approaches the satisfiability threshold.
	%In fact, given the failure of the method of moments and given that even the derivation of the first order approximation~\eqref{eqlaregenum} of $\log Z(\PHI)$ requires fairly delicate manoeuvres, getting a handle on the second moment $\log^2Z(\PHI)$ of the {\em logarithm} of the number of satisfying assignments seems like a formidable task.

	To tackle this challenge we devise a combinatorial interpretation of $\log^2Z(\PHI)$.
	A key idea, which we borrow from spin glass theory~\cite{ChenDeyPanchenko}, is to set up a family of correlated random formulas.
	%A second key idea is to reduce the study of the correlated random formulas to a sequence of stepwise `local' changes.
	%Towards this second step we will build upon some of the ingredients from~\cite{2sat}, particularly Gibbs uniqueness.
	Specifically, given integers $M,M'\geq0$ we construct a correlated pair $(\PHI_1(M,M'),\PHI_2(M,M'))$ of formulas on the variable set $V_n=\{x_1,\ldots,x_n\}$ as follows.
	Let $(\va_i)_{i\geq1}$, $(\va_i')_{i\geq1}$, $(\va_i'')_{i\geq1}$ be sequences of mutually independent uniformly random clauses on $V_n$.
	Then
	\begin{align}\label{eqPHI12}
		\PHI_1(M,M')&=\va_1\wedge\cdots\wedge\va_M\wedge\va_1'\wedge\cdots\wedge\va_{M'}',&
		\PHI_2(M,M')&=\va_1\wedge\cdots\wedge\va_M\wedge\va_1''\wedge\cdots\wedge\va_{M'}''.
	\end{align}
	Thus, the two formulas share clauses $\va_1,\ldots,\va_M$.
	Additionally, each contains another $M'$ independent clauses.
	In particular, $\PHI_1(m,0)$, $\PHI_2(m,0)$ are identical, while $\PHI_1(0,m)$, $\PHI_2(0,m)$ are independent.

	Interpolating between these extreme cases offers a promising avenue for computing the variance:
	given that $\PHI_1(M,m-M)$ and $\PHI_2(M,m-M)$ are satisfiable for all $M$, we can write a telescoping sum
	\begin{align}\label{eq_lem_comb_var}
		\log Z(\PHI_1(m,0))&\cdot\log Z(\PHI_2(m,0))-\log Z(\PHI_1(0,m))\cdot\log Z(\PHI_2(0,m))\\
		&=\sum_{M=1}^m\log Z(\PHI_1(M,m-M)) \cdot  \log Z(\PHI_2(M,m-M))\nonumber\\&\qquad\qquad\qquad-\log Z(\PHI_1(M-1,m-M+1))\cdot\log Z(\PHI_2(M-1,m-M+1)).
		\nonumber
	\end{align}
	If we {\em could} take the expectation on the l.h.s.\ of \eqref{eq_lem_comb_var}, we would precisely obtain the variance of $\log Z(\PHI)$.
	Moreover, each summand on the r.h.s.\ amounts to a `local' change of swapping a shared clause for a pair of independent clauses.
	Yet we cannot just take the expectation of \eqref{eq_lem_comb_var}, because some $\PHI_{h}(M,m-M)$ may be unsatisfiable.
	To remedy this, we will replace $\log Z(\PHI)$ by a tamer random variable with the same limiting distribution.
	Its construction is based on the Unit Clause Propagation algorithm.

	\subsection{Unit Clause Propagation}\label{sec_forced}
	Employed by all modern SAT solvers as a sub-routine, Unit Clause Propagation is a linear time algorithm that tracks the implications of partial assignments.
	The algorithm receives as input a 2-CNF $\Phi$ along with a set $\cL$ of literals.
	These literals are deemed to be `true'.
	The algorithm then pursues direct logical implications, thereby identifying additional `implied' literals that need to be true so that no clause gets violated.
	This procedure is outlined in Steps 1--2 of Algorithm~\ref{fig_ucp}; the outcome of Steps 1--2 is independent of the order in which literals/clauses are processed.

	\IncMargin{1em}
	\begin{algorithm}[h!]
		\KwData{A 2-CNF $\Phi$ along with a set $\cL$ of literals deemed true.}
		% \KwResult{an estimate of $\SIGMA$}
		% Create a queue $\cD$ and insert the variables $x\in\cX_0$ into $\cD$ in the order of increasing labels $\xi_x$\;
		% Allocate sets $\cU=(V(\Phi)\cup F(\Phi))\setminus\cD_0$, $\cC=\emptyset$ and $\cX=\emptyset$\;
		\While{\upshape there exists a clause $a\equiv l\vee\neg l'$ with $l'\in\cL$ and $l\not\in\cL$}{%
			add literal $l$ to $\cL$\;}
		For variables $x\in V(\Phi)$ such that $x\in\cL$ or $\neg x\in\cL$ let
		\begin{align*}
			\sigma_x=\begin{cases}
				1&\mbox{ if $x\in\cL$ and $\neg x\not\in\cL$},\\
				-1&\mbox{ if $\neg x\in\cL$ and $x\not\in\cL$},\\
				0&\mbox{ otherwise}.
			\end{cases}
		\end{align*}
		Let $\cC$ be the set of all clauses $a$ such that $\sigma_x=0$ for all $x\in\partial a$ and return $\cL,\cC,\sigma$\;
		%\Return\ $\cL,\cC$ and $\sigma$\;
		\caption{Pessimistic Unit Clause Propagation (`\UCP').}\label{fig_ucp}
	\end{algorithm}
	\DecMargin{1em}

	Clearly, trouble brews if \UCP\ ends up placing both a literal $l$ and its negation $\neg l$ into the set $\cL$.
	%That is, the algorithm will eventually add both a variable $x$ and its negation $\neg x$ to set set $\cL$.
	Our `pessimistic' Unit Clause variant makes no attempt at mitigating such contradictions.
	Instead, Step~3 just constructs a partial assignment where all conflicting literals are set to a dummy value zero.
	Additionally, \UCP\ identifies the set $\cC$ of {\em conflict clauses} that contain conflicted variables only.

	Now consider a 2-CNF $\Phi$ on a set of variables $V(\Phi)$.
	For each possible literal $l\in\{x,\neg x:x\in V(\Phi)\}$ we run \UCP$(\Phi,\cL=\{l\})$.
	Let $\cC(\Phi,\{l\})$ be the set of conflict clauses returned by \UCP.
	Obtain the {\em pruned formula} $\hat\Phi$ from $\Phi$ by removing all clauses in $\cC(\Phi)=\bigcup_{l}\cC(\Phi,\{l\})$.
	%In other words, we prune all clauses that could cause a contradiction once we assign some value to a single variable.
	%In effect, it is easy to verify the following.
	%Since pruning removes all potential conflicts,
	Then it is easy to verify the following.%
	\footnote{See \Sec~\ref{sec_ucp} for a detailed proof.}

	\begin{fact}\label{lem_hatphi_sat}
		For any 2-CNF $\Phi$ the pruned 2-CNF $\hat\Phi$ is satisfiable.
	\end{fact}

	Generally, the pruned formula $\hat\Phi$ could have far fewer clauses than the original formula $\Phi$.
	Accordingly, even if $\Phi$ is satisfiable the number $Z(\hat\Phi)$ of satisfying assignments of $\hat\Phi$ could dramatically exceed $Z(\Phi)$.
	However, the following proposition shows that on a random formula, the impact of pruning is modest.

	\begin{proposition}\label{prop_arnab}
		With probability $1 - o(n^{-1/2})$ we have $|\log Z(\hPHI) - \log Z(\PHI)| \leq n^{1/3}.$
	\end{proposition}

	\subsection{Variance redux}\label{sec_pruned_var}
	The error bound from \Prop~\ref{prop_arnab} is tight enough so that towards the proof of \Thm~\ref{thm_clt} it suffices to establish a central limit theorem for $\log Z(\hPHI)$, i.e., the log of the number of satisfying assignments of the pruned formula.
	Once again the pivotal task to this end is to compute the variance of $\log Z(\hPHI)$.
	Revisiting the telescoping sum \eqref{eq_lem_comb_var}, we obtain the following expression.
	Recalling~\eqref{eqPHI12}, we write $\hPHI_h(M,M')=\widehat{\PHI_h(M,M')}$ for the formula obtained by pruning $\PHI_h(M,M')$.

	\begin{lemma}\label{lem_comb_rgb}
		Let
		\begin{align}\label{eqDeltaM}
			\Delta(M)&=\Erw\brk{\log\bcfr{ Z(\hPHI_1(M,m-M))}{ Z(\hPHI_1(M-1,m-M))} \cdot \log\bcfr{ Z(\hPHI_2(M,m-M))}{ Z(\hPHI_2(M-1,m-M))} },\\
			\Delta'(M)&=\Erw\brk{\log\bcfr{ Z(\hPHI_1(M-1,m-M+1))}{ Z(\hPHI_1(M-1,m-M))} \cdot \log\bcfr{ Z(\hPHI_2(M-1,m-M+1))}{ Z(\hPHI_2(M-1,m-M))} }.\label{eqDelta'M}
		\end{align}
		Then $\displaystyle
		%	\begin{align}\label{eq_rgb}
			\Var\brk{\log Z(\hPHI)}=\sum_{M=1}^m\Delta(M)-\Delta'(M).
			%	\end{align}
		$
	\end{lemma}

	\Lem~\ref{lem_comb_rgb} expresses the variance as a sum of local changes.
	For example, $\PHI_1(M,m-M)$ is obtained from $\PHI_1(M-1,m-M)$ by adding a single random clause, namely $\va_M$.
	%The expression $\Delta'(M)$ admits a similar combinatorial interpretation.
	Thus, $\Delta(M)$ equals the expected change  upon addition of a single {\em shared} clause---modulo the effect of pruning, that is.

	But fortunately, on random formulas only a few clauses get pruned \whp\
	In effect, we can express the impact of these random changes neatly in terms of random satisfying assignments of the `small' formulas $\hPHI_h(M-1,m-M)$ that appear in \eqref{eqDeltaM}--\eqref{eqDelta'M}.
	Specifically, the quotients in \eqref{eqDeltaM}--\eqref{eqDelta'M} boil down to the probabilities that random satisfying assignments of the `small' formulas survive the extra clause that gets added to obtain the 2-CNFs in the respective numerators.
	Thus, with $\SIGMA=(\SIGMA_y)_{y\in V_n}$ denoting a random satisfying assignment of $\hPHI_h(M-1,m-M)$, we obtain the following.

	\begin{proposition}\label{lem_rs}
		Let $1\leq M\leq m$.
		\Whp\ we have
		\begin{align*}
			\frac{ Z(\hPHI_h(M,m-M))}{ Z(\hPHI_h(M-1,m-M))}&=1-\prod_{y\in\partial\va_M}
			\pr\brk{\SIGMA_y\neq\sign(y,\va_M)\mid\hat\PHI_h(M-1,m-M),\va_M}+o(1)\quad(h=1,2),\\
			%		\mu_{\hat\PHI_h(M-1,m-M)}\bc{\cbc{\SIGMA_y\neq\sign(y,\va_M)}}+o(1)&&(h=1,2),\\
			\frac{ Z(\hPHI_1(M-1,m-M+1))}{ Z(\hPHI_1(M-1,m-M))}&=1-\prod_{y\in\partial\va_{m-M+1}'}\pr\brk{\SIGMA_y\neq\sign(y,\va_{m-M+1}')\mid\hPHI_1(M-1,m-M),\va'_{m-M+1}}+o(1),\\
			% \mu_{\hat\PHI_1(M-1,m-M)}\bc{\cbc{\SIGMA_y\neq\sign(y,\va_{m-M+1}')}}+o(1),\\
			\frac{ Z(\hPHI_2(M-1,m-M+1))}{ Z(\hPHI_2(M-1,m-M))}&=1-\prod_{y\in\partial\va_{m-M+1}''}\pr\brk{\SIGMA_y\neq\sign(y,\va_{m-M+1}'')\mid\hPHI_2(M-1,m-M),\va'_{m-M+1}}+o(1).
			%		\frac{Z(\hPHI_2(M-1,m-M+1))}{Z(\hPHI_2(M-1,m-M))}&=1-\prod_{y\in\partial\va_{m-M+1}''}\mu_{\hat\PHI_2(M-1,m-M)}\bc{\cbc{\SIGMA_y\neq\sign(y,\va_{m-M+1}'')}}+o(1).
		\end{align*}
	\end{proposition}

	%\begin{corollary}\label{lem_rs}
	%	For $h=1,2$ and $s,s'\in\{\pm1\}$ we have \aco{we need this for the reduced formula!}
	%	\begin{align*}
		%		\sum_{i,j=1}^n\ex\abs{\mu_{\hPHI_h}(\SIGMA(x_i)=s,\SIGMA(x_j)=s')-\mu_{\hPHI_h}(\SIGMA(x_i)=s)\mu_{\hat\PHI_h}(\SIGMA(x_j)=s')}&=o(n^2).
		%	\end{align*}
	%\end{corollary}

	%we need to get a grip on the distribution of the truth values $\SIGMA_y$ for $y\in\partial\va_M\cup\partial\va_{m-M+1}'\cup\partial\va_{m-M+1}''$ under a random satisfying assignment of $\hat\PHI_h(M-1,m-M)$.
	%since the expression \eqref{eqDeltaM} involves the product of two logarithms,

	\subsection{Local convergence in probability}\label{sec_lwc}
	To evaluate the expressions from \Prop~\ref{lem_rs} we need to get a grip on the {\em joint} distribution of the truth values of $y$ under random satisfying assignments of the two correlated formulas $\hPHI_h(M-1,m-M)$.
	To this end we will devise a Galton-Watson tree $\vT^{\tensor}$ that mimics the {\em joint} distribution of the local structure of $(\hPHI_1(M-1,m-M),\hPHI_2(M-1,m-M))$.
	Subsequently, we will establish Gibbs uniqueness for this Galton-Watson tree to compute the expressions from \Prop~\ref{lem_rs}.

	The Galton-Watson tree $\vT$ from \Sec~\ref{sec_BPop} that describes the local topology of the `plain' random formula $\PHI$ had one type of variable nodes and four types $(\pm1,\pm1)$ of clause nodes.
	To approach the correlated pair $(\hPHI_1(M,m-M-1),\hPHI_2(M,m-M-1))$ we need a Galton-Watson process with three types of variable nodes and a full dozen types of clause nodes.
	Specifically, there are {\em shared}, {\em 1-distinct} and {\em 2-distinct} variable nodes.
	The root $o$ of $\vT^\tensor$ is a shared variable node.
	The clause node types are {\em $(s,s')$-shared}, {\em $(s,s')$ 1-distinct} and {\em $(s,s')$ 2-distinct} for $s,s'\in\PM$.

	In addition to $d\in(0,2)$ the offspring distributions of $\vT^{\tensor}=\vT^\tensor_{d,t}$ involve a second parameter $t\in[0,1]$:
	\begin{itemize}
		\item A {\em shared variable} spawns $\Po(dt/4)$ shared clauses of type $(s,s')$ as well as $\Po(d(1-t)/4)$ $1$-distinct clauses of type $(s,s')$ and $\Po(d(1-t)/4)$ $2$-distinct clauses of type $(s,s')$ for any $s,s'\in\PM$.
		\item An {\em $h$-distinct variable} begets $\Po(d/4)$ $h$-distinct clauses of type $(s,s')$ for any $s,s'\in\PM$ ($h=1,2$).
		\item A {\em shared clause} has precisely one shared variable as its offspring.
		\item An {\em $h$-distinct clause} spawns a single $h$-distinct variable $(h=1,2)$.
	\end{itemize}
	Figure~\ref{fig_var} provides an illustration of the tree $\vT^\tensor$.
	Shared variables/clauses are indicated in red, $1$-distinct variables/clauses in green and $2$-distinct ones in blue.

	From $\vT^\tensor$ we extract a pair $(\vT_1,\vT_2)$ of correlated random trees.
	Specifically, $\vT_h$ is obtained from $\vT^\tensor$ by deleting all $(3 -h)$-distinct variables and clauses.
	Hence, the parameter $t$ determines how `similar' $\vT_1,\vT_2$ are.
	Specifically, if $t=1$ then no $\{1,2\}$-distinct clauses exist and thus $\vT_1,\vT_2$ are identical.
	By contrast, if $t=0$ then $\vT_1,\vT_2$ are independent copies of the tree $\vT$ from \Sec~\ref{sec_BPop}.

	For an integer $\ell\geq0$ obtain $\vT^{\tensor,\,(2\ell)}$, $\vT_1^{(2\ell)}$, $\vT_2^{(2\ell)}$ from $\vT^\tensor,\vT_1,\vT_2$ by omitting all nodes at a distance greater than $2\ell$ from the root $o$.
	As in \Sec~\ref{sec_BPop}, we can interpret these trees as 2-CNFs, with the type $(s,s')$ of a clause indicating the signs of its parent and child variables.
	We say that two possible outcomes $T,T'$ of $\vT^{\tensor,\,(2\ell)}$ are {\em isomorphic} if there is a tree isomorphism that preserves the root $o$ as well as all types.

	Further, a variable $x\in V_n$ is called a {\em $2\ell$-instance} of $T$ in $(\hPHI_1(M,M'),\hPHI_2(M,M'))$ if there exist isomorphisms $\iota_h$ of the 2-CNFs $T_h$ obtained from $T$ by deleting all $(3-h)$-distinct variables/clauses to the depth-$2\ell$ neighbourhoods $\partial^{\leq 2\ell}_{\hPHI_h(M,M')}x$ of $x$ in $\hPHI_h(M,M')$ such that
	\begin{itemize}
		\item the root gets mapped to $x$, i.e., $\iota_1(o)=\iota_2(o)=x$,
		\item for any shared variable $y$ of $T_1,T_2$ the image variables coincide, i.e., $\iota_1(y)=\iota_2(y)$,
		\item for any shared clauses $a$ of $T_1,T_2$ the image $\iota_1(a)=\iota_2(a)\in\{\va_1,\ldots,\va_M\}$ is a shared clause,
		\item for any $1$-distinct clause $a$ whose parent in $T_1$ is a shared variable, $\iota_1(a)\in\{\va'_1,\ldots,\va'_{M'}\}$, and
		\item for any $2$-distinct clause $a$ whose parent in $T_2$ is a shared variable, $\iota_1(a)\in\{\va''_1,\ldots,\va''_{M'}\}$.
	\end{itemize}
	%\pav{Maybe we should specify that this is a regular tree isomorphism. We mentioned it for trees $T$ and $T'$, but technically here we construct isomorphism between tree and formula. From current definion it does not follow that if we map clause $\ra$ clause, then its neighbours $\ra$ neighbours.}

	Let $\vN^{(2\ell)}(T, (\PHI_1(M, M'), \PHI_2(M, M')))$ be the number of $2\ell$-instances of $T$ in $(\PHI_1(M, M'), \PHI_2(M, M'))$.% \nm{In Section 6, the notation $\vN^{(2\ell)}(T, \PHI_1(M,M'), \PHI_2(M,M'))$ is used. }
	The following proposition confirms that $\vT^\tensor$ models the local structure of $(\hPHI_1(M,M'),\hPHI_2(M,M'))$ faithfully.

	\begin{proposition}\label{lemma_lwc}
		Let $\ell>0$ be a fixed integer, let $t\in[0,1]$ and suppose that $M\sim tdn/2$ and $M'\sim (1-t)dn/2$.
		Then \whp\ for all possible outcomes $T$ of  $\vT^{\tensor,\,(2\ell)}$  we have $\vN^{(2\ell)}(T, (\hPHI_1(M, M'), \hPHI_2(M, M'))) \sim n\pr\brk{\vT^{\tensor,\,(2\ell)}\ism T}.$ \end{proposition}

	\subsection{Correlated Belief Propagation}\label{sec_gibbs}
	Now that we have a branching process description of our pair of correlated formulas the next step is to run BP on the random trees $(\vT_1,\vT_2)$ to find the joint distribution of the truth values $\SIGMA_{\vT_1^{(2\ell)},o},\SIGMA_{\vT_2^{(2\ell)},o}$ assigned to the root.
	Hence, let
	\begin{align}\label{eqmyrandomvector}
		\vec\mu^{(2\ell)}=
		\bc{\pr\brk{\SIGMA_{\vT_1^{(2\ell)},o}=1\mid\vT^\tensor},\pr\brk{\SIGMA_{\vT_2^{(2\ell)},o}=1\mid\vT^\tensor}}\in(0,1)^2.
	\end{align}

	Since BP is exact on trees, we could calculate these marginals by iterating \eqref{eqBPop}--\eqref{eqBPop2} for $2\ell$ steps, starting from all-uniform messages.
	But our objective is not merely to calculate the marginals of a specific pair of trees, but the {\em distribution} of the vector \eqref{eqmyrandomvector} for a random $\vT^\tensor$.
	Fortunately, due to the Markovian nature of the Galton-Watson tree $\vT^\tensor$, the bottom-up BP computation on a random tree can be expressed by a fixed point iteration on the space of probability distributions on $\RR^2$.
	The appropriate operator is the $\LLN$-operator from~\eqref{eqlogBPtensor}.
	To be precise, that operator expresses the updates of the log-likelihood ratios of the BP messages from \eqref{eqBPop1}--\eqref{eqBPop2}.
	Thus, let
	$$\mypsi: (z_1,z_2)\in\RR^2\mapsto((1+\tanh(z_1/2))/2,(1+\tanh(z_2/2))/2)\in(0,1)^2$$ be the function that maps log-likelihood ratios back to probabilities.
	Furthermore, for a probability measure $\rho\in\cP(\RR^2)$ let $\mypsi(\rho)$ be the pushforward probability measure on $(0,1)^2$.%
	\footnote{That is, for a measurable $\fA\subset(0,1)^2$ we have $\mypsi(\rho)(\fA)=\rho(\mypsi^{-1}(\fA))$.}

	\begin{proposition}\label{prop_bp}
		Let $\rho_{d,t}^{(0)}\in\cP(\RR^2)$ be the atom at the origin and let $\rho_{d,t}^{(\ell)}=\LLN(\rho_{d,t}^{(\ell-1)})$.
		Then $\MU^{(2\ell)}$ has distribution $\mypsi(\rho_{d,t}^{(\ell)})$.
	\end{proposition}

	\noindent
	We employ the contraction method to show that the sequence $(\rho_{d,t}^{(\ell)})_{\ell\geq1}$ of measures converges.

	\begin{proposition}\label{prop_noela}
		There exists a unique $\rho_{d,t}\in\cP(\RR^2)$ that satisfies \eqref{eqfixedp} and $\lim_{\ell\to\infty}\rho_{d,t}^{(\ell)}=\rho_{d,t}$ weakly.
		%Moreover, $\cB_{d,t}^{\tensor}(\rho_{d,t})<\infty$.
	\end{proposition}

	\noindent
	Furthermore,
	%Since each tree $\vT_h$ separately has the same distribution as the tree $\vT$ from \Sec~\ref{sec_BPop},
	the Gibbs uniqueness property \eqref{eqProp_uniqueness1} extends to $\vT_1$ and $\vT_2$.
	%Hence, \Prop~\ref{prop_uniqueness_old} implies the following.

	\begin{corollary}\label{prop_uniqueness}
		For all $t\in[0,1]$ and $h=1,2$ we have
		\begin{align}\label{eqProp_uniqueness}
			\lim_{\ell\to\infty}\Erw\brk{\max_{\tau\in S(\vT^{(2\ell)}_h)}\abs{\pr\brk{\SIGMA_{\vT^{(2\ell)}_h,o}=1\mid\vT^{\tensor},\SIGMA_{\vT^{(2\ell)}_h,\partial^{2\ell} o}=\tau_{\partial^{2\ell}o}}-\pr\brk{\SIGMA_{\vT_h^{(2\ell)},o}=1\mid\vT^{\tensor}}}}&=0.
		\end{align}
	\end{corollary}

	Combining \Prop s~\ref{prop_bp} and~\ref{prop_noela} and \Cor~\ref{prop_uniqueness}, we are now in a position to pinpoint the joint marginals of $\hat\PHI_1(M,M'),\hat\PHI_2(M,M')$.
	Formally, let
	\begin{align*}
		\pi_{\hat\PHI_1(M,M'),\hat\PHI_2(M,M')}&=\frac1n\sum_{i=1}^n
		\delta_{(\pr[\SIGMA_{\hPHI_1(M,M'),x_i}=1\mid\hat\PHI_1(M,M')],\pr[\SIGMA_{\hPHI_2(M,M'),x_i}=1\mid\hat\PHI_2(M,M')])}\in\cP([0,1]^2)
	\end{align*}
	be the empiricial distribution of the joint marginals of $\hPHI_1(M,M')$ and $\hPHI_2(M,M')$, which we need to know to evaluate the expressions from \Prop~\ref{lem_rs}.
	Furthermore, denote by $W_1(\nix,\nix)$ the Wasserstein $L^1$-distance of two probability measures on $[0,1]^2$.
	%\begin{align*}
	%	\pi_{\hat\PHI_1(M,M'),\hat\PHI_2(M,M')}&=\frac1n\sum_{i=1}^n\delta_{\bc{\pr\brk{\SIGMA_{\hPHI_1(M,M'),x_i}=1\mid\hat\PHI_1(M,M')},\pr\brk{\SIGMA_{\hPHI_2(M,M'),x_i}=1\mid\hat\PHI_2(M,M')}}}
	%\end{align*}
	%be the empirical joint distribution of the marginals.

	\begin{corollary}\label{lem_empirical}
		For any $t\in[0,1]$ and any $M\sim tnd/2$, $M'\sim(1-t)dn/2$ we have
		$$\ex\brk{W_1\bc{\pi_{\hat\PHI_1(M,M'),\hat\PHI_2(M,M')},\mypsi(\rho_{d,t})}}=o(1).$$
	\end{corollary}

	\noindent
	Finally, combining \Prop~\ref{lem_rs} with \Cor~\ref{lem_empirical}, we obtain the variance of $\log Z(\hPHI)$.

	\begin{corollary}\label{cor_var}
		With $\eta(d)^2$ from \eqref{eqthm_clt} we have $\eta(d)>0$ and $\Var\log Z(\hPHI)\sim m\eta_d^2.$
	\end{corollary}

	\begin{figure}
		\includegraphics[height=53mm]{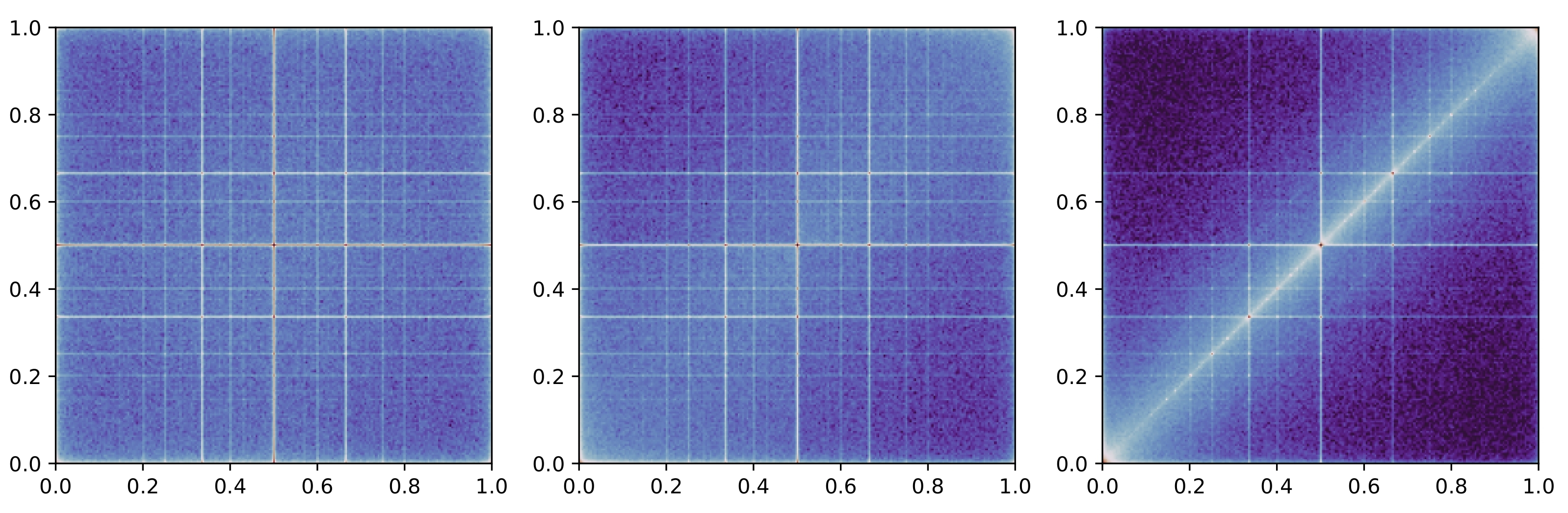}
		\caption{The distributions $\mypsi(\rho_{d,t})$ for $d=1.9$ and $t=0.1 ,0.5,0.9$. }\label{fig_pavel}
	\end{figure}

	Because the proof of \Prop~\ref{prop_noela} is based on a contraction argument, for any $d,t$ the distribution $\rho_{d,t}$ can be approximated effectively within any given accuracy via a fixed point iteration.
	Figure~\ref{fig_pavel} displays approximations to $\mypsi(\rho_{d,t})$ for different values of $t$ and shows how correlations between the two coordinates of the random vector increase with $t$ (brighter diagonal).

	\subsection{The central limit theorem}\label{sec_hh_apply}
	With the variance computation done, we have now overcome the greatest hurdle en route to \Thm~\ref{thm_clt}.
	Indeed, to obtain the desired asymptotic normality we just need to combine the techniques from the variance computation with a generic martingale central limit theorem.

	To this end we set up a filtration $(\fF_{n,M})_{0\leq M\leq m_n}$ by letting $\fF_{n,M}$ be the $\sigma$-algebra generated by $\va_1,\ldots,\va_M$.
	Hence, conditioning on $\fF_{n,M}$ amounts to conditioning on $\va_1,\ldots,\va_M$, while averaging on the remaining clauses $\va_{M+1},\ldots,\va_m$.
	The conditional expectations
	\begin{align}\label{eqZstab1}
		\vZ_{n,M}&=m^{-1/2}\ex\brk{\log Z(\hPHI)\mid\fF_{n,M}}
	\end{align}
	then form a Doob martingale.
	Let $\vX_{n,M}=\vZ_{n,M}-\vZ_{n,M-1}$ be the martingale differences.
	%The steps that we pursued towards the computation of the variance imply the following two statements about the marginalte differences.

	\begin{proposition}\label{prop_varproc}
		For all $0<d<2$  the martingale \eqref{eqZstab1} satisfies
		\begin{align}
			\lim_{n\to\infty}\ex\brk{\max_{1\leq M\leq m}|\vX_{n,M}|}&=0&\mbox{and}&& %\label{eqhh1}&&
			\lim_{n\to\infty}\ex\abs{\eta(d)^2-\sum_{M=1}^{m}\vX_{n,M}^2}&=0.\label{eqmyhh}
		\end{align}
	\end{proposition}

	Thanks to pruning, the first condition from~\eqref{eqmyhh} is easily checked.
	Furthermore, the steps that we pursued towards the proof of \Cor~\ref{cor_var}, i.e., the variance calculation, also imply the second condition without further ado.
	Finally, as \eqref{eqmyhh} demonstrates that the marginal differences are small and that the variance process converges to a deterministic limit, \Thm~\ref{thm_clt} follows from the general martingale central limit theorem from~\cite{Eagleson}.

	\section{Discussion}\label{sec_discussion}

	\noindent
	The hunt for satisfiability thresholds of random constraint satisfaction problems was launched by the experimental work of Cheeseman, Kanefsky and Taylor~\cite{Cheeseman}.
	The 2-SAT threshold was the first one to be caught~\cite{CR,Goerdt}.
	Subsequent successes include the 1-in-$k$-SAT threshold~\cite{ACIM} and the $k$-XORSAT threshold~\cite{DuboisMandler,PittelSorkin}.
	Furthermore, Friedgut~\cite{Friedgut} proved the existence of non-uniform (i.e., $n$-dependent) satisfiability thresholds in considerable generality.
	The plot thickened when physicists employed a compelling but non-rigorous technique called the cavity method to `predict' the exact satisfiability thresholds of many further problems, including the $k$-SAT problem for $k\geq3$~\cite{MPZ}.
	A line of rigorous work~\cite{nae,yuval,KostaSAT} culminated in the verification of this physics prediction for large $k$~\cite{DSS3}.

	Even though the satisfiability threshold of random 2-SAT was determined already in the 1990s, the problem continued to receive considerable attention.
	For example, \Bollobas,  Borgs, Chayes, Kim and~Wilson~\cite{BBCKW} investigated the scaling window around the satisfiability threshold, a point on which a recent contribution by Dovgal, de Panafieu and Ravelomanana elaborates~\cite{Dogval}.
	Abbe and Montanari~\cite{AM} made the first substantial step towards the study of the number of satisfying assignments that $\frac1n\log Z(\PHI)$ converges in probability to a deterministic limit $\varphi(d)$ for Lebesgue-almost all $d\in(0,2)$.
	However, their techniques do not reveal the value $\varphi(d)$.
	Moreover, Montanari and Shah~\cite{MS} obtain a `law-of-large-numbers' estimate of the number of assignments that satisfy all but $o(n)$ clauses for $d<1.16$.
	%More precisely, they proved that in this range of $d$ the value of the random 2-SAT partition function for inverse temperatures $0<\beta\leq n^{\Omega(1)}$ converges to the value predicted by Belief Propagation.%
	%\footnote{This does not yield the number of satisfying assignments, which would require taking the limit $\beta\to\infty$ {\em before} $n\to\infty$.}
	%The proof is based on the Gibbs uniqueness property.
	Finally, the aforementioned article of Achlioptas et al.~\cite{2sat} verifies the prediction from~\cite{MZ} as to the number of satisfying assignments for all $d<2$.
	The main result of the present paper refines these results considerably by establishing a central limit theorem.
	%Further ramifications on random 2-SAT include the study of 2-CNFs with given degrees~\cite{CFS}.

	For random $k$-CNFs with $k\geq3$ an upper bound on the number of satisfying assignments can be obtained via the interpolation method from mathematical physics~\cite{PanchenkoTalagrand}.
	This bound matches the predictions of the cavity method~\cite{MM}.
	However, no matching lower bound is currently known.
	The precise physics prediction called the `replica symmetric solution' has only been verified for `soft' versions of random $k$-SAT where unsatisfied clauses are penalised but not strictly forbidden, and for clause-to-variable ratios well below the satisfiability threshold~\cite{MS,Panchenko2,Talagrand}.
	%The proofs exploit interpolation and contraction arguments and/or the Gibbs uniqueness property.

	Random CSPs such as random $k$-XORSAT or random $k$-NAESAT that exhibit stronger symmetry properties than random $k$-SAT tend to be amenable to the method of moments \cite{nae}.%
	\footnote{Formally, by `symmetry' we mean that the empirical distribution of the marginals of random solutions converges to an atom; cf.~\cite{CKPZ}.}
	Therefore, more is known about their number of solutions.
	For example, due to the inherent connection to linear algebra, the number of satisfying assignments of random $k$-XORSAT formulas is known to concentrate on a single value right up to the satisfiability threshold~\cite{Ayre,DuboisMandler,PittelSorkin}.
	Furthermore, in random $k$-NAESAT, random graph colouring and several related problems the logarithm of the number of solutions superconcentrates, i.e., has only bounded fluctuations for constraint densities up to the so-called condensation threshold, a phase transition that shortly precedes the satisfiability threshold~\cite{silent,CKM,Feli2}.
	The same is true of random $k$-SAT instances with regular literal degrees~\cite{COW}.
	A further example is the symmetric perceptron~\cite{ALS}, where the number of solutions superconcentrates but the limiting distribution is a log-normal with bounded variance.
	Going beyond the condensation transition, Sly, Sun and Zhang~\cite{SSZ} proved that the number of satisfying assignments of random regular $k$-NAESAT formulas matches the `1-step replica symmetry breaking' prediction from physics.
	%Additionally, Nam, Sly and Sun~\cite{NSS,NSS2} investigated the structure of the $k$-NAE solutions within the condensation regime, verifying that a bounded number of solutions clusters dominate.

	Apart from the superconcentration results for symmetric problems from~\cite{silent,COW,CKM,Feli2}, the limiting distribution of the logarithm of the number of solutions has not been known in any random constraint satisfaction problem.
	In particular, \Thm~\ref{thm_clt} is the first central limit theorem for this quantity in any random CSP.
	We expect that the technique developed in the present work, particularly the use of two correlated random instances in combination with spatial mixing, can be extended to other problems.
	The present use of correlated instances is inspired by the work of Chen, Dey and Panchenko~\cite{ChenDeyPanchenko} on the $p$-spin model from mathematical physics, a generalisation of the famous Sherrington-Kirkpatrick model.
	That said, on a technical level the present use of correlated instances is quite different from the approach from~\cite{ChenDeyPanchenko}.
	Specifically, while here we construct correlated 2-CNFs that share a specific fraction of their clauses and employ a martingale central limit theorem, Chen, Dey and Panchenko combine a continuous interpolation of two mixed $p$-spin Hamiltonians with Stein's method.
	%In short, the present approach seems better suited to sparse random structures.

	A further line of work deals with central limit theorems for random optimisation problems.
	Cao~\cite{Cao} provided a general framework based on the `objective method'~\cite{AldousSteele}.
	Unfortunately, the conditions of Cao's theorem tend to be unwieldy for \textsc{Max Csp} problems with hard constraints.
	Recent work of Krea\v ci\v c~\cite{Kreacic} and Glasgow, Kwan, Sah, Sawhney~\cite{GKSS} on the matching number therefore instead resorts to the use of stochastic differential equations.
	A promising question for future work might be whether the present method of considering correlated instances might extend to random optimisation problems.

	\subsection*{Organisation}\label{sec_org}
	In the rest of the paper we carry out the strategy from \Sec~\ref{sec_outline} in detail.
	After some preliminaries in \Sec~\ref{sec_prelim}, we prove \Prop~\ref{prop_arnab} in \Sec~\ref{sec_prop_arnab}.
	Subsequently \Sec~\ref{sec_lemma_lwc} deals with the proof of \Prop~\ref{lemma_lwc}.
	The proof of \Prop~\ref{lem_rs} follows in \Sec~\ref{sec_lem_rs}.
	Moreover, \Sec~\ref{sec_prop_noela} contains the proof of \Prop~\ref{prop_noela}.
	Further, in \Sec~\ref{sec_prop_bp} we prove \Prop~\ref{prop_bp} and \Sec~\ref{sec_prop_var} contains the proofs of \Prop~\ref{prop_varproc} and \Cor~\ref{cor_var}.
	Finally, in \Sec~\ref{sec_hh} we complete the proof of \Thm~\ref{thm_clt}.

	\section{Preliminaries and notation}\label{sec_prelim}

	\subsection{Boolean formulas}\label{sec_prelim_boolean}
	A {\em $2$-SAT formula} or {\em 2-CNF $\Phi$} consists of a finite set $V(\Phi)$ of propositional variables and another set $F(\Phi)$ of clauses.
	Unless specified otherwise, we assume that each clause contains two distinct variables.

	For a clause $a\in F(\Phi)$ we denote by $\partial a=\partial_\Phi a$ the set of variables that appear in clause $a$.
	Similarly, for a variable $x\in V(\Phi)$ let $\partial x=\partial_\Phi x$ signify the set of clauses in which $x$ appears.
	Thus, the formula $\Phi$ induces a bipartite graph on variables and clauses, the so-called {\em incidence graph} of $\Phi$.
	Further, the shortest path metric on the incidence graph induces a metric on the variables and clauses of $\Phi$.
	Accordingly, for a variable or clause $u$ let $\partial^\ell u=\partial^\ell_\Phi u$ be the set of all nodes at a distance precisely $\ell$ from $u$.
	Moreover, let $\partial^{\leq\ell}u=\partial^{\leq\ell}_\Phi u$ be the sub-formula of $\Phi$ obtained by deleting all clauses and variables at a distance greater than $\ell$ from $u$.
	In other words, $\partial^{\leq\ell}u$ is the depth-$\ell$ neighbourhood of $u$.
	%\aco{introduce depth-$\ell$ neighbourhoods} \pav{more specifically: what is the distance between two vertices?}

	We encode the Boolean values `true' and `false' as $\pm1$.
	Accordingly, let $S(\Phi)\subset\PM^{V(\Phi)}$ be the set of satisfying assignments of $\Phi$ and let $Z(\Phi)=|S(\Phi)|$.
	Further, $\sign(x,a)=\sign_\Phi(x,a)\in\{\pm1\}$ denotes the sign with which variable $x$ appears in clause $a$, i.e., $\sign(x,a)=1$ if $x$ appears in $a$ positively and $\sign(x,a)=-1$ if $a$ contains the negation $\neg x$.
	Finally, for a literal $l\in\{x,\neg x\}$ we let $|l|=x$ denote the underlying Boolean variable.

	Assuming $S(\Phi) \neq \emptyset$ let
	\begin{align}\label{eqboltz}
		\mu_{\Phi}(\sigma) &= \vecone\{\sigma \in S(\Phi)\} / Z(\Phi), &&(\sigma \in  \{\pm1\}^{V(\Phi)})
	\end{align}
	be the uniform distribution on $S(\Phi)$.
	We write $\vec\sigma=\vec\sigma_{\Phi}=(\vec\sigma_{\Phi,x})_{x\in V(\Phi)}\in\PM^{V(\Phi)}$ for a sample from $\mu_{\Phi}$, i.e., a uniformly random satisfying assignment of $\Phi$.

	In contrast to $k$-SAT for $k\geq3$, the 2-SAT problem can be solved in polynomial time.
	This is because a 2-SAT instance is unsatisfiable if and only if it contains a peculiar sub-formula called a bicycle.
	To be precise, let $\Phi$ be a CNF with clauses of length one or two.
	A {\em bicycle} of $\Phi$ is an alternating sequence $l_0,a_1,l_1,a_2,\ldots,a_k,l_k$ of literals $l_0,\ldots,l_k$ and clauses $a_1,\ldots,a_k\in F(\Phi)$ such that
	\begin{description}
		\item[BIC1] $l_0=l_k$,
		\item[BIC2] $l_i=\neg l_0$ for some $0<i<k$ and
		\item[BIC3] $a_i\equiv \neg l_{i-1}\vee l_i\equiv l_{i-1}\to l_i$.
	\end{description}
	(Observe that a clause $a$ comprising only a single literal $l$ is logically equivalent to $l\vee l\equiv\neg l\to l$.)
	Hence, the bicycle consists of clauses that are logically equivalent to a chain of implications $l_0\to\neg l_0\to l_0$.

	\begin{fact}[{\cite{APT79}}]\label{fact_bicycle}
		A CNF $\Phi$ with clauses of lengths one or two is unsatisfiable iff $\Phi$ contains a bicycle.
	\end{fact}

	\subsection{Unit Clause Propagation}\label{sec_ucp}
	The \UCP\ algorithm (Algorithm~\ref{fig_ucp}) takes as input a CNF $\Phi$ along with an initial set $\cL_0$ of literals.
	\UCP\ outputs a set $\cL=\cL(\Phi,\cL_0)\supseteq\cL_0$ of literals.
	Let $$\cV(\Phi,\cL_0)=\{|l|:l\in\cL(\Phi,\cL_0)\}$$ be the set of underlying variables.
	In addition to $\cL(\Phi,\cL_0)$, \UCP\ also outputs a partial assignment $$\sigma=\sigma_{\Phi,\cL_0}:\cV(\Phi,\cL_0)\to\{0,\pm1\}$$ that sets each $x\in\cV$ either to a truth value $\pm1$ or to the dummy value $0$.
	Let $$\cV_0(\Phi,\cL_0)=\{x\in\cV(\Phi,\cL_0):\sigma_{\Phi,\cL_0,x}=0\}$$ be the set of variables that receive the dummy value.
	Finally, the algorithm identifies a set $\cC(\Phi,\cL_0)$ of {\em conflict clauses}, i.e., clauses $a$ such that $\partial a\subset\cV_0(\Phi,\cL_0)$.

	We make a note of a few basic facts about \UCP.
	These remarks apply to any CNF $\Phi$ with clauses of length {\em at most} two.
	To get started, we say that a literal $l'$ is {\em implication-reachable} from another literal $l$ if there exists an alternating sequence $l=l_0,a_1,l_1,\ldots,a_k,l_k=l'$ of literals $l_i$ and clauses $a_i$ of $\Phi$ such that $a_i\equiv \neg l_{i-1}\vee l_i\equiv l_{i-1}\to l_i$ for all $1\leq i\leq k$.
	We call this sequence an {\em implication chain} from $l$ to $l'$.
	Observe that a unit clause (clause of length one) comprising a single literal $l$ is equivalent to the implication $\neg l\to l\equiv l\vee l$.
	Furthermore, if $l'$ is implication-reachable from $l$, then $\neg l$ is implication-reachable from $\neg l'$.
	Indeed, if $l=l_0,a_1,l_1,\ldots,a_k,l_k=l'$ is an implication chain from $l$ to $l'$, then its contraposition
	\begin{align*}
		\neg l'=\neg l_k,a_k,\neg l_{k-1},\ldots,\neg l_1,a_1,\neg l
	\end{align*}
	is an implication chain from $\neg l'$ to $\neg l$.

	\begin{lemma}\label{lem_ucp_reachable}
		Let $\Phi$ be a CNF with clauses of length at most two and let $\cL_0$ be a set of literals of $\Phi$.
		Then $\cL(\Phi,\cL_0)$ is the set of all literals $l'$ that are implication-reachable from a literal $l\in\cL_0$.
	\end{lemma}
	\begin{proof}
		This is an easy induction on the length of the shortest implication chain from $l$ to $l'$.
	\end{proof}

	\noindent
	An immediate consequence of \Lem~\ref{lem_ucp_reachable} is that the order in which \UCP\ proceeds is irrelevant.

	Finally, for the sake of completeness, we carry out the proof of Fact~\ref{lem_hatphi_sat}.

	\begin{proof}[Proof of Fact~\ref{lem_hatphi_sat}]
		Fix some order $l_1,\ldots,l_\nu$ of the literals $\{x,\neg x:x\in V(\Phi)\}$.
		Let $\sigma_i$ be the assignments produced by \UCP\ on input $(\Phi,\{l_i\})$.
		We construct an assignment $\sigma:V(\Phi)\to\{0,\pm1\}$ by proceeding as follows for $i=1,\ldots,\nu$.
		For each variable $x$ such that $\{x,\neg x\}\cap\cL(\Phi,l_i)\setminus\bigcup_{1\leq h<i}\cL(\Phi,l_h)\neq\emptyset$ let $\sigma_x=\sigma_i(x)$.
		We claim that for each clause $a\in\hat\Phi$ there is a variable $x\in\partial a$ such that $\sigma_x=\sign(x,a)$.
		Indeed, it is not possible that $\sigma_x=0$ for all $x\in\partial a$; for otherwise $a\in\cC(\Phi,l_i)$ for some $i$, and thus $a$ would not be present in $\hat\Phi$.
		Thus, there exists $x\in\partial a$ such that $\sigma_x\in\PM$.
		If $\sigma_x\neq\sign(x,a)$, then \UCP\ would have included the second literal $k$ that appears in $a$ into the set $\cL$.
		Hence, $\sigma_{|k|}=\sign(|k|,a)$, because otherwise $a$ would have been a contradiction and therefore omitted from $\hat\Phi$.
	\end{proof}
	\subsection{Random 2-SAT}\label{sec_prelim_random}
	Recall from \Sec~\ref{sec_motiv} that $\PHI$ denotes the random $2$-CNF formula with variables $V_n=\{x_1,\ldots,x_n\}$ and clauses $F_m=\{\va_1,\ldots,\va_m\}$, where $m=m_n\sim dn/2$ for a fixed $d>0$.
	{\bf\em We tacitly assume that $0<d<2$, i.e., that we are in the satisfiable regime.}

	In the following sections we will need estimates of the sizes of the sets $|\cV(\PHI,\cL)|$, $|\cC(\PHI,\cL)|$ produced by \UCP\ on the random formula $\PHI$ for singletons $\cL$.
	Thus, suppose we start the \UCP\ algorithm from an initial literal $\cL=\{l\}$.
	Since the ensuing chain of implications traced by \UCP\ is stochastically dominated by a sub-critical branching process (for $d<2$), we obtain the following bound.

	\begin{lemma}[{\cite[Claim~6.8]{2sat}}]\label{lem_forced_tail}
		For any literal $l$ and every $t>8/(2-d)$ we have
		\begin{align*}
			\pr\brk{\abs{\cV(\PHI,\{l\})}>t}\leq(2+o(1))\exp(-dt/40).
		\end{align*}
	\end{lemma}

	\begin{corollary}\label{lem_bicycle}
		With probability $1-o(n^{-2})$ we have
		\begin{align*}
			%	\sum_{i=1}^n\vecone\cbc{|\cC(\PHI,\{x_i\})|+|\cC(\PHI,\{\neg x_i\})|>0}&\leq\log^2n,\\
			\max_{1\leq i\leq n}|\cV(\PHI,\{x_i\})|+|\cV(\PHI,\{\neg x_i\})|&\leq\log^2n.
		\end{align*}
	\end{corollary}
	\begin{proof}
		This is an immediate consequence of \Lem~\ref{lem_forced_tail}.
	\end{proof}

	\noindent
	Finally, the following statement estimates the probability that a random formula is unsatisfiable.

	\begin{lemma}\label{cor_bicycle}
		We have $\pr\brk{\PHI\mbox{ is unsatisfiable}}\leq n^{o(1)-1}$.
	\end{lemma}
	\begin{proof}
		This follows from Fact~\ref{fact_bicycle} and \cite[Claim~6.9]{2sat}.
	\end{proof}

	Recall the Galton-Watson tree $\vT$ from \Sec~\ref{sec_BPop}.
	The following lemma shows that $\vT$ mimics the local structure of the `plain' random formula with $n'$ variables and $m'$ independent random $2$-clauses.
	Also recall that $\partial^{\leq2\ell}_{\Phi} x$ denotes the sub-formula of $\Phi$ comprising all clauses and variables at distance at most $2\ell$ from $x$.

	\begin{lemma}[{\cite[p.~15]{2sat}}]\label{lem_lwc_old}
		Let $\ell \geq 0$ be an integer and let $T$ be a possible outcome of $\vT^{(2\ell)}$.
		Let $\PHI_0$ be a random 2-CNF with $n' \sim n$ variables and $m' \sim dn/2$ clauses.
		Then \whp\ the number $\vN^{(2\ell)}(T,\PHI_0)$ of variables $x_i$ of $\PHI_0$ such that $\partial^{\leq2\ell}_{\PHI_0} x_i \ism T^{2\ell}$ satisfies
		\begin{align*}
			\vN^{(2\ell)}(T,\PHI_0)= n\pr\brk{ \vT^{(2\ell)} \ism T } + o(n).
		\end{align*}
	\end{lemma}

	As a final preparation we need an upper bound on the maximum variable degree.

	\begin{lemma}\label{lem_lwc_d_old}
		With probability $1-o(n^{-10})$ the degree of any variable node $x_i$, $i  = 1, \ldots, n$ in $\PHI$ is bounded by $\log^2 n$.
	\end{lemma}
	\begin{proof}
		The number of clauses that contain a given variable $x_i$ has distribution $\Bin(m,2/n)$.
		Therefore, the assertion follows from the Chernoff bound.
	\end{proof}

	\begin{corollary}\label{lem_lwc_d}
		With probability $1-o(n^{-10})$ the degree of any variable node $x_i$, $i  = 1, \ldots, n$ in $(\PHI_1, \PHI_2)$ is bounded by $\log^2 n$.
	\end{corollary}
	\begin{proof}
		Since $\PHI_1,\PHI_2$ separately are distributed as $\PHI$, the assertion follows from \Lem~\ref{lem_lwc_d_old} and the union bound.
	\end{proof}

	\subsection{Convergence of probability measures}\label{sec_prelim_wasserstein}
	For a measurable subset $\Omega$ of Euclidean space $\RR^k$ we let $\cP(\Omega)$ denote the space of all probability distributions on $\Omega$ equipped with the Borel $\sigma$-algebra.
	Moreover, for $p\geq1$ we define $\cW_p(\Omega)$ to be the set of all $\mu \in \cP(\Omega)$ such that $\int_\Omega\|x\|_2^p\dd\mu(x)<\infty$.
	We equip $\cW_p(\Omega)$ with the Wasserstein metric
	\begin{align}\label{eqWasserstein}
		W_p(\mu, \mu') = \inf_{\vX,\vX'}\Erw\brk{\|\vec X-\vec X'\|_2^p}^{1/p}&&
		(\mu, \mu' \in \cW_p(\Omega)),
	\end{align}
	where the infimum is taken over all pairs of random variables $\vX,\vX'$ that are defined on some common probability space such that $\vX$ has distribution $\mu$ and $\vX'$ has distribution $\mu'$.

	The infimum in \eqref{eqWasserstein} is attained for any $\mu,\mu'$.
	Random vectors $\vec X, \vec X'$ for which the infimum is attained are called \textit{optimal couplings}.
	Such optimal couplings exist for all $\mu, \mu'$~\cite{Bickel}.

	The spaces $(\cW_p(\Omega), W_p)$ are complete metric spaces~\cite{Bickel}.
	Finally, convergence in  $(\cW_p(\Omega), W_p)$ implies weak convergence of the corresponding probability measures.

	For a measure $\rho \in \cP(\Omega)$ and a measurable function $f: \Omega \to \Omega'$ from $\Omega$ to another probability space $\Omega'$ we denote by $f(\rho)$  the pushforward measure of $\rho$.
	Thus, the measure $f(\rho)$ that assigns mass $\rho(f^{-1}(A))$ to measurable $A \subseteq \Omega'$.

	Throughout the paper we let
	\begin{align*}
		\mypsi: \RR^2 \to (0,1)^2, \mypsi\begin{pmatrix}x_1 \\ x_2\end{pmatrix} = \begin{pmatrix} \frac{1+\tanh(x_1/2)}{2}\\ \frac{1+\tanh(x_2/2)}{2}\end{pmatrix}, \qquad
		\myphi: (0,1)^2 \to \RR^2,  \myphi\begin{pmatrix}x_1 \\ x_2\end{pmatrix} = \begin{pmatrix} \log\frac{x_1}{1-x_1}\\ \log \frac{x_2}{1-x_2}\end{pmatrix}.
	\end{align*}

	\section{Proof of \Prop~\ref{prop_arnab}}\label{sec_prop_arnab}

	\noindent
	In this section we estimate the difference between the number of satisfying assignments of the pruned random formula $\hPHI$ and the original formula $\PHI$.
	We begin with a basic observation about the Unit Clause Propagation algorithm, and then estimate the number of clauses that the pruning process removes.
	Apart from proving \Prop~\ref{prop_arnab}, the considerations in this section also pave the way for the proof of the variance formula in \Sec~\ref{sec_prop_var}.

	\subsection{Tracing Unit Clause Propagation}\label{sec_prune_ucp}
	For a 2-CNF $\Phi$ and a set of literals $\cL_0$ consider a set of conflict clauses $\cC=\cC(\Phi,\cL_0)$ that \UCP\ produces along with a set $\cV=\cV(\Phi,\cL_0)$ of conflict variables.
	Let $\Phi-\cC$ be the formula obtained from $\Phi$ by deleting the clauses from $\cC$.
	Clearly $Z(\Phi)\leq Z(\Phi-\cC)$.
	Conversely, the following lemma puts a bound on how much bigger $Z(\Phi-\cC)$ may be.

	\begin{lemma}%[{\cite[Fact~6.4]{2sat}}]
		\label{lem_forced_remove}
		Assume that $\Phi$ is a satisfiable 2-CNF.
		For any set $\cL_0$ of literals we have
		\begin{align}\label{eq_lem_forced_remove}
			Z(\Phi-\cC(\Phi,\cL_0))\leq 2^{|\cV(\Phi,\cL_0)| \cdot \vecone\{\cC(\Phi,\cL_0) \not= \emptyset\}}Z(\Phi).
		\end{align}
	\end{lemma}

	Towards the proof of \Lem~\ref{lem_forced_remove} let $\cL=\cL(\Phi,\cL_0)$ be the final set of literals that \UCP\ produces.
	Moreover, let $\sigma:\cV\to\{0,\pm1\}$ be the function that \UCP\ outputs and let $\cV_0=\{x\in\cV:\sigma_x=0\}$.
	Further, let $\Phi_0$ be a CNF with variable set $\cV$ that contains the following clauses:
	\begin{enumerate}[(i)]
		\item any clause $a\in F(\Phi)$ with $\partial a\subset\cV$,
		\item a unit clause $l$ for every literal $l$ with $|l|\in\cV$ such that $\Phi$ contains a clause $a\equiv l\vee l'$ with $|l'|\not\in\cV$.
	\end{enumerate}
	Thus, $\Phi_0$ contains clauses of length one or two.

	\begin{claim} \label{claim_forced_remove}
		The formula $\Phi_0$ possesses a satisfying assignment $\tau$ such that $\tau_x=\sigma_x\mbox{ for all }x\in\cV\setminus\cV_0.$
	\end{claim}
	\begin{proof}
		Obtain a formula $\Phi_1$ by adding to $\Phi_0$ a unit clause $x$ for every variable $x\in\cV$ with $\sigma_x=1$ and a unit clause $\neg x$ for every $x\in\cV$ with $\sigma_x=-1$.
		Then we just need to show that $\Phi_1$ is satisfiable.

		Assume otherwise.
		Then by Fact~\ref{fact_bicycle} $\Phi_1$ contains a bicycle $l_0,a_1,l_1,a_2,\ldots,a_k,l_k$.
		%We claim that $|l_i|\in\cV_0$ for all $i$.
		%	We are going to argue that the bicycle does not contain any unit clauses.
		This bicycle is logically equivalent to an implication chain
		\begin{align}\label{eq_lem_forced_remove_101}
			l_0\to l_1\to\cdots\to\neg l_0\to\cdots\to l_{k-1}\to l_k=l_0.
		\end{align}
		The contraposition of this chain reads
		\begin{align}\label{eq_lem_forced_remove_102}
			\neg l_0=\neg l_k\to\neg l_{k-1}\to\cdots\to l_0\to\cdots\to\neg l_1\to\neg l_0.
		\end{align}
		Since $\Phi$ is satisfiable, Fact~\ref{fact_bicycle} shows that the bicycle~\eqref{eq_lem_forced_remove_101} cannot be contained in $\Phi$.
		Therefore, the bicycle contains a unit clause $l_i\in F(\Phi_1)\setminus F(\Phi)$ for some $1\leq i\leq k$.
		Hence, the constructions of $\Phi_0$ and $\Phi_1$ ensure that $l_i\in\cL(\Phi,\cL_0)$.
		Indeed, letting $\cU$ be the the set of all literals $l_i$ that appear in~\eqref{eq_lem_forced_remove_101} as unit clauses, we obtain $\cU\subset\cL(\Phi,\cL_0)$.

		We claim that in fact $l_0,\ldots,l_k\in\cL(\Phi,\cL_0)$.
		To see this, pick any $0\leq j\leq k$ such that $l_j$ does not appear as a unit clause in $\Phi_1$.
		Define $l_{-i}=l_{k-i}$ for $0\leq i< k$ and let $1-k\leq i<j$ be the largest index such that $l_i\in\cU$.
		Then $\Phi$ contains the implication chain $l_i\to\cdots\to l_j$.
		Therefore, \Lem~\ref{lem_ucp_reachable} implies that $l_j\in\cL(\Phi,\cL_0)$.
		Analogously, considering the contraposition~\eqref{eq_lem_forced_remove_102}, we conclude that the negations of the literals $l_0,\ldots,l_k$ belong to $\cL(\Phi,\cL_0)$.
		In summary,
		\begin{align}\label{eq_lem_forced_remove_103}
			l_0,\neg l_0,\ldots,l_k,\neg l_k\in\cL(\Phi,\cL_0).
		\end{align}

		But~\eqref{eq_lem_forced_remove_103} implies that $|l_0|,\ldots,|l_k|\in\cV_0$.
		Consequently, none of these literals belongs to a unit clause $u\in F(\Phi_1)\setminus F(\Phi_0)$.
		Furthermore, none of the literals $l_i,\neg l_i$ belongs to a unit clause $a\in F(\Phi_0)\setminus F(\Phi)$.
		This is because if $\Phi$ contains a clause $l_i\vee l'$ or $\neg l_i\vee l'$ and $l_i,\neg l_i\in\cL(\Phi,\cL_0)$, then \UCP\ added $l'$ to $\cL(\Phi,\cC_0)$ as well.
		Thus, we conclude that the bicycle \eqref{eq_lem_forced_remove_101}	consists of clauses of $\Phi$ only.
		But by Fact~\ref{fact_bicycle} this contradicts the fact that $\Phi$ is satisfiable.
	\end{proof}

	\begin{proof}[Proof of \Lem~\ref{lem_forced_remove}]
		Clearly, if $\cC(\Phi,\cL_0) = \emptyset$, the statement is true.
		Hence, assume that $\cC(\Phi,\cL_0) \neq \emptyset$ and let $\tau$ be a satisfying assignment of $\Phi_0$ from Claim~\ref{claim_forced_remove}.
		Consider a satisfying assignment $\chi$ of $\Phi-\cC$ and let $\chi':V(\Phi)\setminus\cV\to\PM$ be the restriction of $\chi$ to $V(\Phi)\setminus\cV$.
		We extend $\chi'$ to a satisfying assignment $\chi''$ of $\Phi$ by letting
		\begin{align*}
			\chi''_x&=\vecone\{x\in\cV\}\tau_x+\vecone\{x\not\in\cV\}\chi'_x;
		\end{align*}
		clearly, $\chi''$ satisfies all clauses $a$ such that $\partial a\cap\cV=\emptyset$, because all these clauses are contained in $\Phi-\cC$.
		Moreover, $\chi''$ satisfies all $a$ such that $\partial a\subset\cV$, because these clauses belong to $\Phi_0$.
		Further, if $a=l\vee l'$ is a clause such that $|l|\in\cV$ but $|l'|\not\in\cV$, then $\tau_{|l|}=\sigma_{|l|}$.
		Since $|l'|\not\in\cV$, this means that $\sigma_x=\sign(|l|,a)$, as otherwise \UCP\ would have added $l'$ to $\cL$.
		Therefore, $\chi''_{|l|}=\tau_{|l|}=\sigma_{|l|}$ satisfies $a$.
		Since the map $\chi\mapsto\chi'$ only discards the values of the variables in $\cV$, we obtain the bound \eqref{eq_lem_forced_remove}.
	\end{proof}

	%Ultimately we are going to derive \Prop~\ref{prop_arnab} from \Lem~\ref{lem_forced_remove}.
	%To this end we need to estimate the sizes of the sets $|\cV(\PHI,\cL_0)|$, $|\cC(\PHI,\cL_0)|$ on the random formula $\PHI$ for singletons $\cL_0$.
	%More precisely, suppose we start the \UCP\ algorithm from an initial literal $\cL_0=\{l\}$.
	%It can be verified that the ensuing chain of implications generated by \UCP\ is stochastically dominated by a sub-cricitcal branching process.
	%Thus, we obtain the following bound.
	%
	%\begin{lemma}[{\cite[Claim~6.8]{2sat}}]\label{lem_forced_tail}
	%	For any literal $l$ and every $t>8/(2-d)$ we have
	%	\begin{align*}
		%		\pr\brk{\abs{\cV(\PHI,\{l\})}>t}\leq(2+o(1))\exp(-dt/40).
		%	\end{align*}
	%\end{lemma}
	%
	%\begin{corollary}\label{lem_bicycle}
	%	With probability $1-o(n^{-2})$ we have
	%\begin{align*}
	%%	\sum_{i=1}^n\vecone\cbc{|\cC(\PHI,\{x_i\})|+|\cC(\PHI,\{\neg x_i\})|>0}&\leq\log^2n,\\
	%	\max_{1\leq i\leq n}|\cV(\PHI,\{x_i\})|+|\cV(\PHI,\{\neg x_i\})|&\leq\log^2n.
	%\end{align*}
	%\end{corollary}
	%\begin{proof}
	%	This follows from \Lem~\ref{lem_forced_tail}.
	%\end{proof}

	%\noindent
	%Finally, the following statement estimates the probability that a random formula is unsatisfiable.
	%
	%\begin{lemma}\label{cor_bicycle}
	%	We have $\pr\brk{\PHI\mbox{ is unsatisfiable}}\leq n^{o(1)-1}$.
	%\end{lemma}
	%\begin{proof}
	%	This follows from Fact~\ref{fact_bicycle} and \cite[Claim~6.9]{2sat}.
	%\end{proof}
	%\nm{Why not make this a corollary?}

	\subsection{Cycles in random formulas}\label{sec_topology} To prove \Prop~\ref{prop_arnab} we need a good estimate of the total number of clauses that will be removed from $\PHI$ to obtain $\hPHI$.
	This estimate is provided by the following lemma.

	\begin{lemma}\label{lem_unicycle}
		Fix any $\delta>0$.
		With probability $1-o(n^{-1})$ the number of literals $l$ such that $\cC(\PHI,\cbc l)\neq\emptyset$ is smaller than $n^\delta$.
	\end{lemma}
	\begin{proof}
		Let $\vR$ be the number of literals $l$ such that $\cC(\PHI,\cbc l)\neq\emptyset$.
		We are going to show that for any fixed (i.e., $n$-independent) $\ell\geq1$ for large enough $n$ we have,
		\begin{align}\label{eq_lem_unicycle_1}
			\ex\brk{\prod_{i=1}^\ell(\vR-i+1)}\leq\bc{\ell\log^3n}^\ell.
		\end{align}
		Providing $\ell\geq2/\delta$ and $n$ is sufficiently large, Markov's inequality then shows that
		\begin{align*}
			\pr\brk{\vR\geq n^\delta}&\leq\pr\brk{\prod_{i=1}^\ell(\vR-i+1)\geq (n^\delta/2)^\ell}\leq\bcfr{2\ell\log^3n}{n^\delta}^\ell=o(n^{-1}),
		\end{align*}
		which implies the assertion.

		Thus, we are left to prove \eqref{eq_lem_unicycle_1}.
		By symmetry it suffices to bound the probability of the event $$\fE=\bigcap_{i=1}^\ell \{\cC(\PHI,\{x_i\}) \not= \emptyset \}$$ that \UCP\ will produce at least one conflict clause from each of the literals $x_1,\ldots,x_\ell$; then
		\begin{align}\ex\brk{\prod_{i=1}^\ell(\vR-i+1)}\leq (2n)^\ell\pr\brk{\fE}.\label{eq_lem_unicycle_1a}\end{align}

		In order to estimate the probability of $\fE$ we are going to launch \UCP\ from the initial set $\cL=\{x_1,\ldots,x_\ell\}$.
		While the order in which the literals and clauses are processed does not affect the ultimate outcome of \UCP, for the present analysis we assume that \UCP\ processes the literals one at a time, each time pursuing {\em all} the clauses $l\vee \neg l'$ that contain the negation of a specific $l'$.
		We also presume that the literals are processed in the same order as they get inserted into the set $\cL$.
		In other words, \UCP\ proceeds in breadth-first-search order.
		Let $\fH_t$ be the history of the execution of \UCP\ up to and including the point where the first $t$ literals and their adjacent clauses have been explored.
		Formally, $\fH_t$ is the $\sigma$-algebra generated by these first $t$ literals that get added to $\cL$ and their adjacent clauses.

		\Lem~\ref{lem_forced_tail} implies that with probability $1 - o(n^{-1})$ the set $\cL$ returned by \UCP\ has size at most $L = \ell \log^2 n$.
		Let $\fE_t$ be the event that at time $t$ we explored a clause that contains two variables from $\cL$ and $|\cL| \leq L$.
		Moreover, let $\vS = \sum_t \vecone\{\fE_t\}$.
		Let $0<t_1 <  \ldots < t_\ell\leq L$ be distinct time steps.
		Then
		\begin{align}\label{eq_lem_unicycle_2}
			\pr\brk{ \fE } \leq \pr\brk{ \vS \geq \ell } \leq \sum_{t_1,\ldots,t_\ell} \pr\brk{ \bigcap_{i=1}^\ell \fE_{t_i} } = \sum_{t_1,\ldots,t_\ell}  \prod_{i=1}^\ell \pr\brk{ \fE_{t_i} \Big| \bigcap_{j=1}^{i-1} \fE_{t_j} }.
		\end{align}
		%\nm{In the last display and later in this proof, should the $t's$ still be bold after bounding with the binomial coefficient ${L \choose \ell}$ over all possible positions? They appear to be fixed from this point onwards.}
		To bound the r.h.s.\ of \eqref{eq_lem_unicycle_2} we will estimate the probability of $\fE_{t+1}$ given the history $\fH_t$ of the process up to time $t$, showing that for all $t\geq0$,
		\begin{align}\label{eq_lem_unicycle_3}
			\pr\brk{ \fE_{t + 1} | \fH_t} \leq {L \over n}.
		\end{align}
		In fact the probability that in step $t + 1$ we will run into already discovered variable is bounded by the probability that the literal explored during that step shares a clause with an already explored variable, which is bounded by $Lm/n\leq L/n$; for if more than $L$ literals have been already explored the event $\fE_{t+1}$ does not occur by definition.
		Finally, because the event $\fE_t$ is $\fH_t$-measurable, \eqref{eq_lem_unicycle_3} implies
		\begin{align*}
			\prod_{i=1}^\ell \pr\brk{ \fE_{t_i} \Big| \bigcap_{j=1}^{i-1} \fE_{t_j}} \leq \left( {L \over n} \right)^{\ell}.
		\end{align*}
		Thus \eqref{eq_lem_unicycle_2} gives $\pr\brk{ \fE } \leq {L \choose \ell} {L^\ell \over n^\ell} \leq \bc{\eul L^2 \over \ell n }^\ell$, which together with~\eqref{eq_lem_unicycle_1a} implies \eqref{eq_lem_unicycle_1}.
	\end{proof}

	\begin{proof}[Proof of \Prop~\ref{prop_arnab}]
		\Lem~\ref{cor_bicycle} shows that $\pr\brk{Z(\PHI)=0}=o(n^{-1/2})$.
		Thus, we may condition on the event that $\PHI$ is satisfiable.
		Furthermore, \Lem~\ref{lem_forced_remove} shows that given that $\PHI$ is satisfiable we have
		\begin{align}\label{eq_prop_arnab_1}
			\log Z(\hPHI)-\log Z(\PHI)\leq\sum_{i=1}^n\vecone\{\cC(\PHI,\{x_i\})\neq\emptyset\}|\cV(\PHI,\{x_i\})|+\vecone\{\cC(\PHI,\{\neg x_i\})\neq\emptyset\}|\cV(\PHI,\{\neg x_i\})|.
		\end{align}
		%\nm{Include indicators also in statement of \Lem~\ref{lem_forced_remove}?}
		Finally, \Cor~\ref{lem_bicycle} and \Lem~\ref{lem_unicycle} (applied with $\delta<1/3$) imply that with probability $1-o(n^{-1/2})$,
		\begin{align}\label{eq_prop_arnab_2}
			\sum_{i=1}^n&\vecone\{\cC(\PHI,\{x_i\})\neq\emptyset\}|\cV(\PHI,\{x_i\})|+\vecone\{\cC(\PHI,\{\neg x_i\})\neq\emptyset\}|\cV(\PHI,\{\neg x_i\})|\\
			&\leq\bc{\abs{\cbc{x\in V_n:\cC(\PHI,\{x\})\neq\emptyset}}+\abs{\cbc{x\in V_n:\cC(\PHI,\{\neg x\})\neq\emptyset}}}
			\max_{1\leq i\leq n}\abs{\cV(\PHI,\{x_i\})|+|\cV(\PHI,\{\neg x_i\})}=o(n^{1/3}).\nonumber
		\end{align}
		Thus, the assertion follows from \eqref{eq_prop_arnab_1}--\eqref{eq_prop_arnab_2}.
	\end{proof}

	\section{Proof of \Prop~\ref{lemma_lwc}}\label{sec_lemma_lwc}

	\noindent
	The proof of \Prop~\ref{lemma_lwc} is based on a combination of  a coupling and a second moment argument.
	As a first step we observe that we do not need to worry about trees of very high maximum degree.

	\begin{lemma}\label{lem_lwc_a}
		For any $\eps>0$, $\ell\geq0$ there exists $L>0$ such that for all $t\in[0,1]$ with probability at least $1-\eps$ the tree $\vT^{\tensor,\,(2\ell)}$ has maximum degree less than $L$.
	\end{lemma}
	\begin{proof}
		The construction of the tree $\vT^{\tensor,\, (2\ell)}$ in \Sec~\ref{sec_lwc} ensures that every variable node has a Poisson number of clauses as offspring.
		The mean of this Poisson variable is always bounded by $2d$.
		Hence, Bennett's inequality shows that for any $L>2d$ the probability that a specific variable has more than $L$ offspring is bounded by $\exp(-L^2/(4d+L))$.
		Thus, choosing $L$ sufficiently large so that $\eps>L^{2\ell}\exp(-L^2/(4d+L))$ and applying the union bound, we obtain the assertion (combined with the chain rule starting from the root).
	\end{proof}

	Thus, in the following we confine ourselves to trees $T$ with a maximum degree bounded by a large enough number $L$.
	First we are going to count the number of copies of such trees $T$ in $(\PHI_1(M,M'),\PHI_2(M,M'))$ via the method of moments.
	The following lemma estimates the first moment.
	%Recall the definition of the Galton-Watson tree $\vT$ from \Sec~\ref{sec_BPop}.
	%\begin{lemma}[{\cite[p.~15]{2sat}}]\label{lem_lwc_old}
	%	Let $\PHI_0$ be a uniformly random 2-CNF with $n' \sim n$ variables and $m' \sim dn/2$ clauses.
	%	Then the number $N^{(2\ell)}(T,\PHI_0)$ of variables $x_i$ of $\PHI_0$ such that $\partial^{\leq 2\ell}_{\PHI_0} x_i \ism T^{(2\ell)}$   satisfies
	%	\begin{align*}
		%		N^{(2\ell)}(T,\PHI_0) = n\pr\brk{ \vT^{(2\ell)} \ism T } + o(n)
		%	\end{align*}
	%	for any $\ell \geq 0$ and any possible outcome $T$ of $\vT^{(2\ell)}$.
	%\end{lemma}
	%We need an extension of \Lem~\ref{lem_lwc_old} to the pairs $(\PHI_1(M, M'), \PHI_2(M, M'))$ of correlated random formulas and $(\vT_1,\vT_2)$ of correlated random trees.

	\begin{lemma}\label{lem_lwc_b}
		For any fixed integers $L,\ell$, %such that $\ell$ is even,
		any possible outcome $T$ of $\vT^{\tensor,\,(2\ell)}$ of maximum degree at most $L$ and any $M \sim tdn/2$, $M' \sim (1-t) dn/2$ we have
		\begin{align*}
			\ex[\vN^{(2\ell)}(T, (\PHI_1(M, M'), \PHI_2(M, M')))]\sim n\pr\brk{\vT^{\tensor,\,(2\ell)}\ism T}.
		\end{align*}
	\end{lemma}
	\begin{proof}
		We proceed by induction on $\ell$.
		In the case $\ell=0$ the tree $T$ consists of nothing but the root, so that there is nothing to show.
		Hence, let $\ell \geq 1$.
		Let $\lambda_{0,s_1, s_2}$ be the number of shared children of the root $o$ of $T$ where $o$ appears with sign $s_1 \in \{-1, +1\}$ and the other variable appears with sign $s_2 \in \{-1, +1\}$.
		Also let $\lambda_{h, s_1, s_2}$ be the number of $h$-distinct children of $o$ ($h=1,2$),  where $o$ appears with sign $s_1 \in \{-1, +1\}$ and the other variable appears with sign $s_2 \in \{-1, +1\}$.

		Consider the event $\fE$ that variable $x_1$ is a $2\ell$-instance of $T$.
		Further, consider the event $\fR$ that $x_1$ occurs in precisely $\lambda_{0,s_1, s_2}$ clauses among $\va_1, \ldots, \va_M$, where the sign of $x_1$ is $s_1$ and the sign of the other variable is $s_2$, precisely in $\lambda_{1, s_1, s_2}$ clauses among $\va'_1, \ldots, \va'_{M'}$, where the sign of $x_1$ is $s_1$ and the sign of the other variable is $s_2$ and precisely in $\lambda_{2, s_1, s_2}$ clauses among $\va''_1, \ldots, \va''_{M'}$, where the sign of $x_1$ is $s_1$ and the sign of the other variable is $s_2$.
		Since $M\sim dtn/2$ and $\lambda_{h,\pm1,\pm1}\leq L$ for $h\in\{0,1,2\}$ we have
		\begin{align}\nonumber
			\pr\brk\fR&\sim\prod_{s_1, s_2\in \{\pm1\}} \pr\brk{\Bin(M,(2n)^{-1})=\lambda_{0, s_1, s_2}}\pr\brk{\Bin(m-M,(2n)^{-1})=\lambda_{1, s_1, s_2}}\pr\brk{\Bin(m-M,(2n)^{-1})=\lambda_{2, s_1, s_2}}\\
			&\sim \prod_{s_1, s_2 \in \{\pm1\}}  \pr\brk{\Po(dt/4)=\lambda_{0, s_1, s_2}} \pr\brk{\Po(d(1-t)/4)=\lambda_{1, s_1, s_2}} \pr\brk{\Po(d(1-t)/4)=\lambda_{2, s_1, s_2}}.\label{eq_lem_lwc_b_1}
		\end{align}
		Let $\lambda_h = \lambda_{h, -1, -1} + \lambda_{h, -1, +1} + \lambda_{h, +1, -1} + \lambda_{h, +1, +1}$ for $h \in \{0, 1, 2\}$.
		Given $\fR$ let $(\vv_{0,i})_{1\leq i\leq\lambda_0}$ be the second variables (other than $x_1$) contained in neighbours of $x_1$ among $\va_1, \ldots, \va_M$.
		Analogously, let $(\vv_{1,i})_{1\leq i\leq\lambda_h}$ be the second variables contained in neighbours of $x_1$ among $\va'_1, \ldots, \va'_M$ and $(\vv_{2,i})_{1\leq i\leq\lambda_h}$ be the second variables contained in neighbours of $x_1$ among $\va''_1, \ldots, \va''_M$.
		By $\PHI_h^{-}$, $h = 1, 2$ define a random formula obtained from $\PHI_h(M, M')$ by deleting $x_1$ and its adjacent clauses.
		Let $\fF$ be the event that the distance between any two of $\vv_{0, 1}, \ldots, \vv_{2,\lambda_2}$ in both $\PHI_1^{-}$ and $\PHI_2^{-}$ is at least $2\ell$.
		A routine union bound argument shows that
		\begin{align}\label{eq_lem_lwc_b_2}
			\pr\brk{\fF} = 1 - o(1).
		\end{align}
		Further, let $T_{0,i}$ be the sub-tree obtained from $T$ comprising the $i$-th shared grandchild of $o$ and its descendants.
		Consider the event $\fH_0$ that $\fR$ and $\fF$ occur and $\vv_{0,i}$ is a $(2\ell-2)$-instance of $T_{0, i}$ in $(\PHI_1^{-}, \PHI_2^{-})$ for any $i = 1, \ldots, \lambda_0$.
		Since the depth and the maximum degree of $T$ are bounded, by induction we obtain
		\begin{align*}
			\pr\brk{ \text{$\vv_{0,i}$ is a $(2\ell-2)$-instance of $T_{0,i}$ in $(\PHI_1^{-}, \PHI_2^{-})$}} = \pr\brk{\vT^{\tensor,\, (2\ell-2)} \ism T_{0, i}} + o(1).
		\end{align*}
		for $i = 1, \ldots, \lambda_0 $.
		%\nm{(The asymptotics for fixed vertices used here appears to be slightly stronger than what the induction hypothesis on the expectation would give.)}
		%\pav{What exactly do you mean here?}
		%	\pav{Maybe we can say a couple of words why we can apply induction. In other words, why are we removing not too much (total size of $\lambda_0$ trees is sublinear)}
		Thus
		\begin{align}\label{eq_lem_lwc_b_3}
			\pr\brk{\fH_0 | \fF\cap\fR} =
			\pr\brk{\bigcap_{i=1}^{\lambda_0} \text{$\vv_{0,i}$ is a $(2\ell-2)$-instance of $T_{0,i}$ in $(\PHI_1^{-}, \PHI_2^{-})$}} = \prod_{i=0}^{\lambda_0}\pr\brk{\vT^{\tensor,\, (2\ell-2)} \ism T_{0, i}} + o(1).
		\end{align}

		Analogously, let $T_{h,i}$ be the sub-tree of $T$ pending on the $i$-th $h$-distinct grandchild of the root.
		Consider the events $\fH_h$ that $\fF$ and $\fR$ occur and that the depth $(2\ell-2)$-neighbourhood of $\vv_{h,i}$ is isomorphic to $T_{h, i}$ in $\PHI_h^{-}$ for any $i = 1, \ldots, \lambda_h$, $h = 1, 2$.
		Since $\vv_{1, i}$ and $\vv_{2, j}$ are chosen independently for all $i$ and $j$, using the same embedding process as above in combination with Lemma~\ref{lem_lwc_old} we obtain
		\begin{align}\label{eq_lem_lwc_b_4}
			\pr\brk{\fH_1 | \fF \cap\fR \cap \fH_0} &= \prod_{i=1}^{\lambda_1}\pr\brk{\vT^{(2\ell-2)} \ism T_{1, i}} + o(1)
			\\
			\pr\brk{\fH_2 | \fF\cap\fR \cap \fH_0 \cap \fH_1} &= \prod_{i=1}^{\lambda_2}\pr\brk{\vT^{(2\ell-2)} \ism T_{2, i}} + o(1).\label{eq_lem_lwc_b_5}
		\end{align}
		Finally, combining \eqref{eq_lem_lwc_b_1}--\eqref{eq_lem_lwc_b_5} we obtain
		\begin{align*}
			\pr\brk\fE&\sim \pr\brk\fR\prod_{i=1}^{\lambda_0}\pr\brk{\vT^{\tensor,\, (2\ell-2)} \ism T_{0, i}}\prod_{i=1}^{\lambda_1}\pr\brk{\vT^{(2\ell-2)} \ism T_{1, i}}\prod_{i=1}^{\lambda_2}\pr\brk{\vT^{(2\ell-2)} \ism T_{2, i}}\sim\pr\brk{\vT^{\tensor,\,(2\ell-2)}\ism T}.
		\end{align*}
		As $\ex[\vN^{(2\ell)}(T, (\PHI_1(M, M'), \PHI_2(M, M')))]=n\pr\brk\fE$ the assertion follows from the linearity of expectation.
	\end{proof}

	We also need an estimate of the second moment of $\vN^{(2\ell)}(T, (\PHI_1(M, M'), \PHI_2(M, M')))$.

	\begin{lemma}\label{lem_lwc_c}
		For any fixed integers $L,\ell$ and any possible outcome $T$ of $\vT^{(2\ell)}$ of maximum degree at most $L$ and any $M \sim tdn/2$, $M' \sim (1-t) dn/2$ we have
		\begin{align*}
			\ex[\vN^{(2\ell)}(T, (\PHI_1(M, M'), \PHI_2(M, M')))^2]\sim n^2\pr\brk{\vT^{\tensor,\,(2\ell)}\ism T}^2.
		\end{align*}
	\end{lemma}
	\begin{proof}
		Consider the event $\fE_{ij}$ that both variables $x_i$ and $x_j$ are $2\ell$-instances of $T$ for $i, j = 1, \ldots, n$.
		Now we can rewrite the second moment as follows.
		\begin{align}\label{eq_lem_lwc_c_1}
			\Erw\brk{\vN^{(2\ell)}(T, (\PHI_1(M, M'), \PHI_2(M, M')))^2} = \sum_{i, j = 1}^{n} \pr\brk{\fE_{ij}} = n \pr\brk{\fE_{11}}  + n(n-1)\pr\brk{\fE_{12}}.
		\end{align}
		From \Lem~\ref{lem_lwc_b} we know that $\pr\brk{\fE_{11}} = \pr\brk{\vT^{\tensor,\,(2\ell)}\ism T} + o(1)$, so we only need to estimate $\pr\brk{\fE_{12}}$.
		Let $\fF$ be such event that the distance between $x_1$ and $x_2$  is at least $2\ell$.
		A routine union bound argument shows that
		\begin{align*}
			\pr\brk{\fF} = 1 - o(1).
		\end{align*}
		This fact completes the proof of \Lem~\ref{lem_lwc_c}.
		\begin{align}\label{eq_lem_lwc_c_2}\nonumber
			\pr\brk{ \fE_{12} } &= 	\pr\brk{ \fE_{12} \mid \fF } + o(1) = \pr\brk{\text{$x_1$ is a $2\ell$-instance of $T$} \mid \fF}\cdot  \pr\brk{\text{$x_2$ is a $2\ell$-instance of $T$}\mid \fF} + o(1)
			\\
			&= \pr\brk{\vT^{\tensor,\,(2\ell)}\ism T}^2 + o(1).
		\end{align}
		Thus the assertion follows from \eqref{eq_lem_lwc_c_1} and \eqref{eq_lem_lwc_c_2}.
	\end{proof}

	\begin{proof}[Proof of \Prop~\ref{lemma_lwc}]
		From \Lem s~\ref{lem_lwc_a}--\ref{lem_lwc_c} in combination with Chebyshev's inequality it follows that for any $\ell \geq 0,T$ \whp\
		\begin{align}\label{eq_lem_lwc_1}
			\vN^{(2\ell)}(T, (\PHI_1(M, M'), \PHI_2(M, M'))) \sim n\pr\brk{\vT^{\tensor,\,(2\ell)}\ism T}.
		\end{align}
		We need to extend this to the pruned formulas $(\hPHI_1(M, M'), \hPHI_2(M, M'))$.
		Let $\vN^{(2\ell), +}(T, (\PHI_1, \PHI_2))$ be the number of variable nodes $x$ such that $x$ is an $2\ell$-instance of $T$ in $(\hPHI_1, \hPHI_2)$ but not in $(\PHI_1, \PHI_2)$.
		Similarly, let $\vN^{(2\ell), -}(T, (\PHI_1, \PHI_2))$ be the number of variable nodes $x$ such that they are $2\ell$-instances of $T$ in $(\PHI_1, \PHI_2)$ but not in $(\hPHI_1, \hPHI_2)$.
		Then
		\begin{align}\label{eq_lem_lwc_2}
			\vN^{(2\ell)}(T, (\hPHI_1, \hPHI_2)) = \vN^{(2\ell)}(T, (\PHI_1, \PHI_2)) + \vN^{(2\ell), +}(T, (\PHI_1, \PHI_2)) - \vN^{(2\ell), -}(T, (\PHI_1, \PHI_2)).
		\end{align}
		Note that both $\vN^{(2\ell), +}(T, (\PHI_1, \PHI_2))$ and $\vN^{(2\ell), -}(T, (\PHI_1, \PHI_2))$ do not exceed the number of variable nodes $x$ whose depth-$2\ell$ neighbourhood in $(\PHI_1, \PHI_2)$ contains at least one clause from $\bigcup_{l\in\{x_i,\neg x_i, ~1\leq i\leq n\}} \cC(\PHI,\{l\})$.
		Moreover, \Lem s~\ref{lem_bicycle} and~\ref{lem_unicycle} show that \whp\
		\begin{align}\label{eq_lem_lwc_3}
			\abs{\bigcup_{l\in\{x_i,\neg x_i, ~1\leq i\leq n\}} \cC(\PHI,\{l\})} \leq n^{0.1}.
		\end{align}

		It follows from Lemma~\ref{lem_lwc_d} that \whp\ the $2\ell$-depth neighbourhood of each vertex consists of no more then $\log^{4\ell + 4} n$ vertices.
		Combining this fact with \eqref{eq_lem_lwc_3} we conclude that
		\begin{align}\label{eq_lem_lwc_4}
			\vN^{(2\ell), +}(T, (\PHI_1, \PHI_2)) \leq n^{0.1} \log^{4\ell + 4} n, \qquad\qquad \vN^{(2\ell), -}(T, (\PHI_1, \PHI_2)) \leq n^{0.1} \log^{4\ell + 4} n
		\end{align}
		\whp.\
		Finally, the assertion follows from \eqref{eq_lem_lwc_1}, \eqref{eq_lem_lwc_2} and \eqref{eq_lem_lwc_4}.
	\end{proof}

	\section{Proof of \Prop~\ref{lem_rs}}\label{sec_lem_rs}

	\noindent
	We will deal with $\frac{ Z(\hPHI_h(M,m-M))}{ Z(\hPHI_h(M-1,m-M))}$ in detail; the arguments for the other two quotients are similar.

	\begin{lemma}\label{lem_addaM}
		Let $h\in\{1,2\}$.
		\Whp\ $\hat\PHI_h(M,m-M)$ is obtained from $\hat\PHI_h(M-1,m-M)$ by adding a clause $\va_M$.
	\end{lemma}
	\begin{proof}
		Let $\vl,\vl'$ be the constituent literals of $\va_M$, i.e., $\va_M = \vl \vee \vl'$.
		Moreover, let $\fQ$ be the event that $\hat\PHI_h(M,m-M)$ does \emph{not} result from $\hat\PHI_h(M-1,m-M)$ by adding clause $\va_M$.
		Thus, on the event $\fQ$ the additional clause $\va_M$ triggers the pruning of clauses that do not get pruned from $\hPHI(M-1,m-M)$ (including potentially $\va_M$ itself).

		We are going to construct events $\fE,\fE'$ whose probabilities are easy to estimate such that
		\begin{align}\label{eq_lem_addaM_100}
			\fQ\subseteq\fE\cup\fE'.
		\end{align}
		To this end, for a literal $l$ let $\cL_l=\cL(\PHI_h(M-1, m - M), \{l\})$ be the final set of literals that \UCP$(\PHI_h(M-1, m - M), \{l\})$ produces.
		Call $l$ a {\em trigger} of $\neg\vl$ if $\neg\vl\in\cL_l$.
		Further, let $\fE$ be the event that there exists a trigger $l$ of $\neg\vl$ such that
		\begin{description}
			\item[E1] $\cC(\PHI(M-1, m - M),\{l\})\cup\cC(\PHI(M-1, m - M),\{l,\vl'\})\neq\emptyset$, or
			\item[E2] $\neg\vl'\in\bigcup_{\lambda\in\{l,\vl'\}}\cL_\lambda.$
		\end{description}
		Define $\fE'$ analogously with the roles of $\vl,\vl'$ swapped.

		We claim that these events $\fE,\fE'$ satisfy~\eqref{eq_lem_addaM_100}.
		To see this, assume that neither $\fE$ nor $\fE'$ occurs.
		We claim that then
		\begin{align}\label{eq_lem_addaM_101}
			\cC(\PHI_h(M-1,m-M),\{l\})=\cC(\PHI_h(M,m-M),\{l\})
		\end{align}
		for all literals $l$; if so, then clearly $\fQ$ does not occur either.

		Thus, assume that~\eqref{eq_lem_addaM_101} is false and that $l$ is a literal such that
		\begin{align}\label{eq_lem_addaM_200}
			\cC(\PHI_h(M-1,m-M),\{l\})\neq\cC(\PHI_h(M,m-M),\{l\}).
		\end{align}
		Then $l$ must be a trigger of $\neg\vl$ or of $\neg\vl'$; for otherwise the presence of the extra clause $\va_M$ has no impact on the set of conflict clauses.
		Hence, suppose that $l$ is a trigger of $\neg\vl$.
		Then the presence of clause $\va_M$ in $\PHI_h(M,m-M)$ causes \UCP\ to add $\vl'$ to $\cL(\PHI_h(M,m-M),\{l\})$.
		Since the event $\fE$  does not occur, neither does {\bf E1} and we conclude that $\cC(\PHI_h(M-1,m-M),\{l\})=\cC(\PHI_h(M-1,m-M),\{l,\vl'\})=\emptyset$.
		Hence, none of the clauses $a\in F(\PHI_h(M-1,m-M))$ is a conflict clause and thus~\eqref{eq_lem_addaM_200} implies that
		\begin{align*}
			\{\va_M\}=\cC(\PHI_h(M,m-M),\{l\})\setminus\cC(\PHI_h(M-1,m-M),\{l\}).
		\end{align*}
		But this is not possible either.
		For if $\va_M\in\cC(\PHI_h(M,m-M),\{l\})$, then \Lem~\ref{lem_ucp_reachable} shows that one of $l,\vl,\vl'$ is a trigger of $\neg\vl'$, and thus {\bf E2} occurs.
		Thus, we obtain~\eqref{eq_lem_addaM_101}.

		%	Observe that $\pr\brk{\fQ}$ can be bounded by the following expression.
		%	\aco{this is not correct. what if we hit the same path twice?}
		%	\begin{align}\label{eq_lem_addaM_1}
			%		\pr\brk{\fQ} \leq \pr\brk{ \bigcup_{l:~ \cC(\PHI_h(M,m-M), \cbc{l}) \neq \emptyset} \{l_M \in \cL_l\} \cup \{\neg l_M \in \cL_l\} \cup \{l_M' \in \cL_l\} \cup \{\neg l_M' \in \cL_l\}}.
			%	\end{align}
		To complete the proof we are going to show that
		\begin{align}\label{eq_lem_addaM_110}
			\pr\brk\fE,\pr\brk{\fE'}&=o(1).
		\end{align}
		Indeed, \Lem~\ref{lem_unicycle} shows that the number of literals $l$ such that $\cC(\PHI_h(M-1,m-M), \cbc{l}) \neq \emptyset$ can be bounded by $n^{0.1}$ \whp.
		Furthermore, \Cor~\ref{lem_bicycle} shows that $|\cV(\PHI_h(M-1,m-M),\{l\})|\leq\log^2n$ \whp\ for all $l$.
		Hence, \whp\ the total number of literals $\lambda$ that have a trigger $l$ such that $\cC(\PHI_h(M-1,m-M),\{l\})\neq\emptyset$ is bounded by $O(n^{0.1}\log^2n)$.
		Consequently, the probability that the random literal $\neg\vl$ possesses such a trigger is bounded by $O(n^{-0.9}\log^2n)$.
		Moreover, since $\vl'$ is a random literal as well, \Lem~\ref{lem_unicycle} shows that $\pr\brk{\cC(\PHI(M-1,m-M),\{\vl'\})=\emptyset}=1-O(n^{-0.9})$.
		Additionally, \whp\ for any trigger $l$ of $\neg\vl$ we have $\cV(\PHI(M-1,m-M),\{l\})\cap\cV(\PHI(M-1,m-M),\{\vl'\})=\emptyset$, because $\vl,\vl'$ are drawn independenly of $\PHI(M-1,m-M)$.
		Similarly,
		\begin{align*}
			\pr\brk{\neg\vl'\in\cL(\PHI(M-1,m-M),\{l\})}=o(1)&&\mbox{and}&&\pr\brk{\neg\vl'\in\cL(\PHI(M-1,m-M),\{\vl'\})}=o(1).
		\end{align*}
		Combining these estimates, we conclude that $\pr\brk\fE=o(1)$.
		By symmetry, the same estimate holds for $\fE'$.
		Thus, we obtain \eqref{eq_lem_addaM_110}.
		Finally, the assertion follows from~\eqref{eq_lem_addaM_100} and~\eqref{eq_lem_addaM_110}.
	\end{proof}

	\begin{corollary}\label{cor_addaM}
		Let $h\in\{1,2\}$.
		\Whp\ we have
		\begin{align*}
			\frac{ Z(\hPHI_h(M,m-M))}{ Z(\hPHI_h(M-1,m-M))}&=\mu_{\hPHI_h(M-1,m-M)}\bc{\cbc{\SIGMA\models\va_M}}.
		\end{align*}
	\end{corollary}
	\begin{proof}
		From \Lem~\ref{lem_addaM} we know that \whp
		\begin{align}\label{eq_corr_addaM}
			Z(\hPHI_h(M,m-M)) = Z(\hPHI(M - 1,m-M)+\va_M).
		\end{align}
		Assuming that \eqref{eq_corr_addaM} is correct, $ Z(\hPHI_h(M,m-M))$ equals the number of satisfying assignments of $\hPHI_h(M-1,m-M)$ that also happen to satisfy $\va_M$.
	\end{proof}

	Additionally, we need the following asymptotic independence property, known as `replica symmetry' in physics parlance.

	\begin{lemma}\label{lem_rs_single}
		Let $h\in\{1,2\}$.
		For all $s,s'\in\PM$ we have
		\begin{align*}
			\frac1{n^2}\sum_{i,j=1}^n\ex\abs{\mu_{\hPHI_h(M-1,m-M)}(\{\SIGMA_{x_i}=s,\,\SIGMA_{x_j}=s'\})-\mu_{\hPHI_h(M-1,m-M)}(\{\SIGMA_{x_i}=s\})\mu_{\hPHI_h(M-1,m-M)}(\{\SIGMA_{x_j}=s'\})}&=o(1).
		\end{align*}
	\end{lemma}
	\begin{proof}
		We adapt an argument from~\cite{Mossel} to the present setting.
		By exchangeability it suffices to prove that
		\begin{align*}
			\ex\abs{\mu_{\hPHI_h(M-1,m-M)}(\{\SIGMA_{x_1}=s,\,\SIGMA_{x_2}=s'\})-\mu_{\hPHI_h(M-1,m-M)}(\{\SIGMA_{x_1}=s\})\mu_{\hPHI_h(M-1,m-M)}(\{\SIGMA_{x_2}=s'\})} = o(1).
		\end{align*}
		The proof rests on the Gibbs uniqueness property.
		Indeed, \Prop~\ref{lemma_lwc} shows that for any fixed $\ell$ the depth-$2\ell$ neighbourhood $\partial^{\leq2\ell}x_i$ of $x_i$ in $\hPHI_h(M-1,m-M)$ is within total variation distance $o(1)$ of the Galton-Watson tree $\vT_h^{(2\ell)}$.
		Furthermore, the distribution of $\vT_h^{(2\ell)}$ by itself is identical to the distribution of the Galton-Watson tree $\vT^{(2\ell)}$.
		Additionally, \Prop~\ref{prop_uniqueness_old} shows that $\vT^{(2\ell)}$ enjoys the Gibbs uniqueness property~\eqref{eqProp_uniqueness1}.
		Consequently, taking $\ell=\ell(n)\to\infty$ sufficiently slowly as $n\to\infty$, we see that \whp
		\begin{align}\label{eq_lem_rs_single_100}
			\sum_{s\in\PM}\max_{\kappa\in S(\hPHI_h(M-1,m-M))}\abs{\mu_{\hPHI_h(M-1,m-M)}(\{\SIGMA_{x_1}=s ~|~\SIGMA_{\partial^{\ell} x_1} = \kappa_{\partial^{2\ell}x_1}\}) - \mu_{\hPHI_h(M-1,m-M)}(\{\SIGMA_{x_1}=s\})}&=o(1).
		\end{align}

		Furthermore, providing $\ell=\ell(n)\to\infty$ slowly enough, the distance between $x_1,x_2$ exceeds $4\ell$ \whp\
		In this case, \eqref{eq_lem_rs_single_100} gives
		\begin{align}\nonumber
			&\mu_{\hPHI_h(M,m-M)}(\{\SIGMA_{x_1}=s,\,\SIGMA_{x_2}=s'\}) = \mu_{\hPHI_h(M,m-M)}(\{\SIGMA_{x_1}=s ~|~ \SIGMA_{x_2}=s'\}) \cdot \mu_{\hPHI_h(M,m-M)}(\{\SIGMA_{x_2}=s'\})
			\\
			&=  \mu_{\hPHI_h(M,m-M)}(\{\SIGMA_{x_2}=s'\}) \cdot \sum_{\kappa \in \{\pm1\}^{\partial^{2\ell} x_1}} \mu_{\hPHI_h(M,m-M)}(\{\SIGMA_{x_1}=s ~|~\SIGMA_{\partial^{2\ell} x_1} = \kappa,  \SIGMA_{x_2}=s'\})\nonumber\\
			&\qquad\qquad\qquad\qquad\qquad\qquad\qquad\qquad \cdot \mu_{\hPHI_h(M,m-M)}(\{\SIGMA_{\partial^{2\ell} x_1} = \kappa ~|~  \SIGMA_{x_2}=s'\})\nonumber\\
			&=  \mu_{\hPHI_h(M,m-M)}(\{\SIGMA_{x_2}=s'\}) \cdot \sum_{\kappa \in \{\pm1\}^{\partial^{2\ell} x_1}} \mu_{\hPHI_h(M,m-M)}(\{\SIGMA_{x_1}=s ~|~\SIGMA_{\partial^{2\ell} x_1} = \kappa\})\nonumber\\
			&\qquad\qquad\qquad\qquad\qquad\qquad\qquad\qquad \cdot \mu_{\hPHI_h(M,m-M)}(\{\SIGMA_{\partial^{2\ell} x_1} = \kappa ~|~  \SIGMA_{x_2}=s'\})\nonumber\\
			&= \mu_{\hPHI_h(M-1,m-M)}(\{\SIGMA_{x_1}=s\})\mu_{\hPHI_h(M-1,m-M)}(\{\SIGMA_{x_2}=s'\}) \cdot (1 + o(1)),
		\end{align}
		as claimed.
	\end{proof}

	\begin{proof}[Proof of \Prop~\ref{lem_rs}]
		The proposition follows from \Cor~\ref{cor_addaM} and \Lem~\ref{lem_rs_single}.
	\end{proof}

	\section{Proof of \Prop~\ref{prop_noela}}\label{sec_prop_noela}

	\noindent
	In this section, we prove \Prop~\ref{prop_noela} via a contraction argument. For this, recall the operator $\LLN$ from~\eqref{eqlogBPtensor}.
	For notational convenience we let
	\begin{align*}
		\vV&=\begin{pmatrix}
			\sum_{i=1}^{\vd} \vec s_i \log\bc{\frac{1+\vec r_i \tanh(\vec \xi_{\rho,i,1}/2)}{2}} +  \sum_{i=1}^{\vec d'} \vec s_i' \log\bc{\frac{1+\vec r'_i \tanh(\vec\xi_{\rho,i,1}'/2)}{2}}\\
			\sum_{i=1}^{\vd} \vec s_i \log\bc{\frac{1+\vec r_i \tanh(\vec \xi_{\rho,i,2}/2)}{2}} +  \sum_{i=1}^{\vec d''} \vec s_i'' \log\bc{\frac{1+\vec r_i'' \tanh(\vec\xi''_{\rho,i,2}/2)}{2}}
		\end{pmatrix}.
	\end{align*}
	The main step towards \Prop~\ref{prop_noela} is the following lemma:
	\begin{lemma}\label{lem_noela}
		$\LLN$ is a contraction on the space $(\cW_2(\RR^2), W_2)$ for all $0<d<2$ and $0 \leq t \leq 1$.
	\end{lemma}

	Indeed, it immediately follows from \Lem~\ref{lem_noela} and Banach's fixed point theorem that for every $d \in (0,2)$ and $t \in [0,1]$, there is a unique $\rho_{d,t} \in \cW_2(\RR^2)$ with $\rho_{d,t} = \LLN(\rho_{d,t})$, and that for any $\rho \in \cW_2(\RR^2)$ and $\ell \to \infty$, the $\ell$-fold application of $\LLN$ to $\rho$ converges to $\rho_{d,t}$ in Wasserstein distance.

	We prove  \Lem~\ref{lem_noela} in the following subsection, and conclude the section with the proof of \Prop~\ref{prop_noela}.

	\subsection{Proof of \Lem~\ref{lem_noela}}
	We first check that the operator $\LLN$ is well-defined in the sense that it maps the space $(\cW_2(\RR^2), W_2)$ to itself.

	\begin{claim}\label{claim_haodong}
		The operator $\LLN$ maps the space $(\cW_2(\RR^2), W_2)$ to itself.
	\end{claim}

	\begin{proof}
		Let $\rho\in (\cW_2(\RR^2), W_2)$ and $\vec V$ be a random vector with distribution $\LLN(\rho)$. By the definition of $\LLN$,
		\begin{align}\label{eq_hd_1}
			\Erw\brk{\norm{\vV}_2^2}=&\Erw\Bigg[\bc{\sum_{i=1}^{\vd} \vec s_i \log\bc{\frac{1+\vec r_i \tanh(\vec \xi_{\rho,i,1}/2)}{2}} +  \sum_{i=1}^{\vec d'} \vec s_i' \log\bc{\frac{1+\vec r_i' \tanh(\vec\xi_{\rho,i,1}'/2)}{2}}}^2\nonumber\\
			&+\bc{\sum_{i=1}^{\vd} \vec s_i \log\bc{\frac{1+\vec r_i \tanh(\vec \xi_{\rho,i,2}/2)}{2}} +  \sum_{i=1}^{\vec d''} \vec s_i'' \log\bc{\frac{1+\vec r_i'' \tanh(\vec\xi''_{\rho,i,2}/2)}{2}}}^2\Bigg].
		\end{align}
		By the independence of the random variables $(\vec s_i)_{i \geq 1}$, $(\vs'_i)_{i \geq 1}$ and $(\vec s''_i)_{i \geq 1}$ from everything else, all cross-terms in the evaluation of the squares in \eqref{eq_hd_1} vanish, e.g. for $i \not= j$,
		\begin{align*}
			\Erw&\brk{\vec s_{i} {\vec s}_{j}  \log\bc{\frac{1+\vec r_{i} \tanh(\vec \xi_{\rho,i,1}/2)}{2}} \log\bc{\frac{1+\vec r_{j} \tanh(\vec \xi_{\rho,j,1}/2)}{2}}}\\
			&=\Erw\brk{\vec s_{i}}\Erw\brk{ {\vec s}_{j}  \log\bc{\frac{1+\vec r_{i} \tanh(\vec \xi_{\rho,i,1}/2)}{2}} \log\bc{\frac{1+\vec r_{j} \tanh(\vec \xi_{\rho,j,1}/2)}{2}}}=0.
		\end{align*}
		%for $i\neq j$. The same results hold for the other cross terms as well.
		As a consequence, (\ref{eq_hd_1}) in combination with the independence of the Poisson random variables gives that
		\begin{align}
			\Erw\brk{\norm{\vV}_2^2} % = %&\Erw\Bigg[\sum_{i=1}^{\vd}  \log^2\bc{\frac{1+\vec r_i \tanh(\vec \xi_{\rho,i,1}/2)}{2}} +  \sum_{i=1}^{\vec d'}  \log^2\bc{\frac{1+\vec r_i' \tanh(\vec\xi_{\rho,i,1}'/2)}{2}}\nonumber\\
			%&+\sum_{i=1}^{\vd}  \log^2\bc{\frac{1+\vec r_i \tanh(\vec \xi_{\rho,i,2}/2)}{2}} +  \sum_{i=1}^{\vec d''}  \log^2\bc{\frac{1+\vec r_i'' \tanh(\vec\xi''_{\rho,i,2}/2)}{2}}\Bigg]\nonumber\\
			=&\Erw\Bigg[td  \log^2\bc{\frac{1+\vec r_1 \tanh(\vec \xi_{\rho,1,1}/2)}{2}} +  (1-t)d  \log^2\bc{\frac{1+\vec r_1' \tanh(\vec \xi'_{\rho,1,1}/2)}{2}}\nonumber\\
			&+td \log^2\bc{\frac{1+\vec r_1 \tanh(\vec \xi_{\rho,1,2}/2)}{2}} +  (1-t)d \log^2\bc{\frac{1+\vec r_1'' \tanh(\vec\xi''_{\rho,1,2}/2)}{2}}\Bigg].\label{eq_hd_2}
		\end{align}
		%	Note that
		%	\begin{align}\label{eq_hd_derivaties}
			%		\frac\partial{\partial x}\log\frac{1+\tanh(x/2)}{2}=\frac{1-\tanh(x/2)}{2}\quad\mbox{and}\quad
			%		\frac\partial{\partial x}\log\frac{1-\tanh(x/2)}{2}=-\frac{1+\tanh(x/2)}{2}.
			%	\end{align}
		Finally, conditioning on the value of $\vec r_1$ and an application of the fundamental theorem of calculus, followed by the Cauchy-Schwarz inequality give
		\begin{align*}
			\Erw\brk{\log^2\bc{\frac{1+\vec r_1 \tanh(\vec \xi_{\rho,1,1}/2)}{2}}}=&\frac{1}{2}\Erw\brk{\log^2\bc{\frac{1+ \tanh(\vec \xi_{\rho,1,1}/2)}{2}}+\log^2\bc{\frac{1- \tanh(\vec \xi_{\rho,1,1}/2)}{2}}}\\
			= &\frac{1}{2}\Erw\brk{\bc{\int_0^{\vec \xi_{\rho,1,1}}\frac{1-\tanh(x/2)}{2}\dif x - \log 2}^2+\bc{\int_0^{\vec \xi_{\rho,1,1}}\frac{1+\tanh(x/2)}{2}\dif x + \log 2}^2}\\
			\leq &2\Erw\brk{\log^2 2+\vec \xi_{\rho,1,1}^2}.
		\end{align*}
		Analogous bounds can be derived for the remaining three terms in (\ref{eq_hd_2}).
		Therefore, for any vector $\vec \xi\in\RR^2$ with distribution $\rho$,
		\begin{align*}
			\Erw\brk{\norm{\vV}_2^2}&\leq \Erw\brk{2td\bc{\vec \xi_{\rho,1,1}^2+\vec \xi_{\rho,1,2}^2}+2(1-t)d\bc{{\vec\xi'}^{\,2}_{\rho,1,1}+{\vec\xi''}^{\,2}_{\rho,1,2}}}+4d\log^22=2d\Erw\brk{\norm{\vec \xi}_2^2}+4d\log^22<\infty.
		\end{align*}
		%	Hence, the operator $\LLN$ maps the space $(\cW_2(\RR^2), W_2)$ to itself.
	\end{proof}

	\begin{proof}[Proof of \Lem~\ref{lem_noela}]
		Let $\rho, \nu \in \cW_2(\RR^2)$ be arbitrary.
		To show contraction, consider three independent sequences of optimally coupled pairs $( \vec\xi_{\rho,i},  \vec\xi_{\nu,i})_{i \geq 1}$, $( \vec\xi_{\rho,i}',  \vec\xi_{\nu,i}')_{i \geq 1}$ and $( \vec\xi_{\rho,i}'',  \vec\xi_{\nu,i}'')_{i \geq 1}$ such that for each $\vec\zeta=\vec\xi,\vec\xi',\vec\xi''$, the $\vec\zeta_{\rho,i}=(\vec\zeta_{\rho,i,1},\vec\zeta_{\rho,i,2})\in\RR^2$ have distribution $\rho$, the $\vec\zeta_{\nu,i}=(\vec\zeta_{\nu,i,1},\vec\zeta_{\nu,i,2})\in\RR^2$ have distribution $\nu$ and
		\begin{align}
			W_2(\rho, \nu) = \Erw\brk{\|\vec \zeta_{\rho,i} - \vec \zeta_{\nu,i}\|_2^2}^{1/2}.\label{Eq:W0}
		\end{align}
		%\noela{Wasserstein distance seems to only depend on optimal couplings of marginals : $W_2(\rho, \nu)^2 =$ $W_2(\rho^{(1)}, (\nu)^{(1)})^2 +$  $ W_2(\rho^{(2)}, (\nu)^{(2)})^2$.}

		Let $\vd \sim \Po(td)$ and $\vec d',\vec d'' \sim \Po((1-t)d)$ all be independent.
		Moreover, let $(\vec s_i)_{i \geq 1}$, $(\vec r_i)_{i \geq 1}$, $(\vs'_i)_{i \geq 1}$, $(\vec r'_i)_{i \geq 1}$, $(\vec s''_i)_{i \geq 1}$ and $(\vec r''_i)_{i \geq 1}$ be independent sequences of i.i.d. Rademacher random variables with parameter $1/2$; all not explicitly coupled random variables are assumed to be independent.
		Then with $\hat\rho=\LLN(\rho)$ and $\hat\nu=\LLN(\nu)$ we obtain
		\begin{align}
			W_2(\hat\rho, \hat\nu)^2 &\leq \Erw\brk{\bc{ \sum_{i=1}^{\vd} \vec s_i \log\bc{\frac{1+\vec r_i \tanh(\vec \xi_{\rho,i,1}/2)}{1+\vec r_i \tanh(\vec \xi_{\nu,i,1}/2)}} +\sum_{i=1}^{\vec d'} \vec s_i' \log\bc{\frac{1+\vec r_i' \tanh(\vec \xi'_{\rho,i,1}/2)}{1+\vec r_i' \tanh(\vec \xi'_{\nu,i,1}/2)}}}^2} \nonumber \\
			& \quad + \Erw\brk{\bc{ \sum_{i=1}^{\vd} \vec s_i \log\bc{\frac{1+\vec r_i \tanh(\vec \xi_{\rho,i,2}/2)}{1+\vec r_i \tanh(\vec \xi_{\nu,i,2}/2)}} +\sum_{i=1}^{\vec d''} \vec s_i'' \log\bc{\frac{1+\vec r_i'' \tanh(\vec \xi''_{\rho,i,2}/2)}{1+\vec r_i'' \tanh(\vec \xi''_{\nu,i,2}/2)}}}^2}. \nonumber
		\end{align}
		Analogous to the derivation of (\ref{eq_hd_2}), by the independence of the random signs, the expectations of the cross-terms cancel. Combined with the independence of the Poisson random variables, we conclude that
		\begin{align}
			W_2(\hat\rho, \hat\nu)^2
			&\leq td\Erw\brk{\log^2\bc{\frac{1+\vec r_1 \tanh(\vec \xi_{\rho,1,1}/2)}{1+\vec r_1 \tanh(\vec \xi_{\nu,1,1}/2)}}}  + (1-t)d \Erw\brk{ \log^2\bc{\frac{1+\vec r_1' \tanh(\vec \xi'_{\rho,1,1}/2)}{1+\vec r_1' \tanh(\vec \xi'_{\nu,1,1}/2)}}} \nonumber \\
			& \quad + td \Erw\brk{ \log^2\bc{\frac{1+\vec r_1 \tanh(\vec \xi_{\rho,1,2}/2)}{1+\vec r_1 \tanh(\vec \xi_{\nu,1,2}/2)}}} + (1-t)d \Erw\brk{ \log^2\bc{\frac{1+\vec r_1''\tanh(\vec \xi''_{\rho,1,2}/2)}{1+\vec r_1'' \tanh(\vec \xi''_{\nu,1,2}/2)}}}.\label{Eq:W1}
		\end{align}
		Moreover, conditioning on the value of $\vec r_1$ and an application of the fundamental theorem of calculus yield
		\begin{align}
			\log^2 \frac{1+ \tanh(\vec\xi_{\rho,1,1}/2)}{1+ \tanh(\vec \xi_{\nu,1,1}/2)} &= \brk{\int_{\vec \xi_{\nu,1,1}}^{\vec \xi_{\rho,1,1} } \frac{\partial \log (1+ \tanh(z/2))}{\partial z} \dd z }^2 = \brk{\int_{\vec \xi_{\rho,1,1} \wedge \vec \xi_{\nu,1,1}}^{\vec \xi_{\rho,1,1} \vee \vec \xi_{\nu,1,1}} \frac{1- \tanh(z/2)}{2} dz }^2, \label{Eq:W2} \\
			\log^2 \frac{1- \tanh(\vec \xi_{\rho,1,1}/2)}{1- \tanh(\vec \xi_{\nu,1,1} /2)} &= \brk{\int_{\vec \xi_{\nu,1,1}}^{\vec \xi_{\rho,1,1} } \frac{\partial \log (1- \tanh(z/2))}{\partial z} dz }^2 = \brk{\int_{\vec \xi_{\rho,1,1} \wedge \vec \xi_{\nu,1,1}}^{\vec \xi_{\rho,1,1} \vee \vec \xi_{\nu,1,1} } \frac{1+ \tanh(z/2)}{2} dz }^2. \label{Eq:W3}
		\end{align}
		Combining (\ref{Eq:W2}) and (\ref{Eq:W3}) and applying the Cauchy-Schwarz inequality, we obtain
		\begin{align}
			\Erw \brk{\log^2\bc{\frac{1+\vec r_1 \tanh(\vec \xi_{\rho,1,1}/2)}{1+\vec r_1 \tanh( \vec \xi_{\nu,1,1}/2)}}} &\leq \frac{1}{2} \Erw\brk{\bc{\vec \xi_{\rho,1,1} - \vec \xi_{\nu,1,1}}^2}. \label{Eq:W4}
		\end{align}
		An identical argument can be made for $\vec\xi_{\rho,1,2}$, $\vec\xi_{\nu,1,2}$, $\vec\xi'_{\rho,1,1}$, $\vec\xi'_{\nu,1,1}$ and $\vec\xi''_{\rho,1,2}$, $\vec\xi''_{\nu,1,2}$.
		Finally, (\ref{Eq:W0}),(\ref{Eq:W1}) and (\ref{Eq:W4}) yield
		\begin{align}
			W_2(\hat\rho, \hat\nu)^2 &\leq \frac{td}{2} \Erw \brk{ \bc{ \vec\xi_{\rho,1,1}- \vec \xi_{\nu,1,1}}^2 + \bc{ \vec\xi_{\rho,1,2}-  \vec\xi_{\nu,1,2}}^2 } + \frac{(1-t)d}{2} \Erw\brk{\bc{ \vec\xi'_{\rho,1,1}- \vec \xi'_{\nu,1,1}}^2 + \bc{\vec \xi''_{\rho,1,2}-  \vec\xi''_{\nu,1,2}}^2 }\nonumber \\
			&= \frac{td}{2} W_2(\rho, \nu)^2 +\frac{(1-t)d}{2} W_2(\rho, \nu)^2 = \frac{d}{2} W_2(\rho, \nu)^2,\label{eq_hd_connor}
		\end{align}
		which implies contraction because $d<2$.
	\end{proof}

	\subsection{Proof of \Prop~\ref{prop_noela}}
	The uniqueness of $\rho_{d,t} \in \cW_2(\RR^2)$ with $\LLN(\rho_{d,t}) = \rho_{d,t}$  (which yields \eqref{eqfixedp})  follows from  \Lem~\ref{lem_noela} and the Banach fixed point theorem. As the Dirac measure in zero is an element of $\cW_2(\RR^2)$,  \Lem~\ref{lem_noela} also implies the weak convergence of $(\rho_{d,t}^{(\ell)})_{\ell \geq 0}$ to $\rho_{d,t}$.

\section{Proof of \Prop~\ref{prop_bp}}\label{sec_prop_bp}

\noindent
As a first step toward the proof of \Prop~\ref{prop_bp} we are going to introduce an operator on probability distributions on the unit square that resembles the Belief Propagation update equations~\eqref{eqBPop1}--\eqref{eqBPop2}.
We will see that this operator is closely related to the operator from~\eqref{eqlogBPtensor}.
Specifically, \eqref{eqlogBPtensor} is the log-likelihood version of the new operator.
Subsequently, we will show that the Belief Propagation operator correctly implements marginal computations on the Galton-Watson trees $(\vT_1,\vT_2)$.

\subsection{Density evolution}\label{sec_density}
Recall that $\cP((0,1)^2)$ is the space of all Borel probability measures on the unit square $(0,1)^2$.
We define an operator
\begin{align}\label{eqBPtensor}
	\BP_{d,t}^{\tensor}&:\cP((0,1)^2)\to\cP((0,1)^2),&\pi\mapsto\hat\pi=\BP_{d,t}^\tensor(\pi)
\end{align}
as follows.
For $s\in\PM$ let
\begin{align*}
	(\MU_{\pi,s,i,1},\MU_{\pi,s,i,2})_{i\geq1},
	(\MU_{\pi,s,i,1}',\MU_{\pi,s,i,2}')_{i\geq1},
	(\MU_{\pi,s,i,1}'',\MU_{\pi,s,i,2}'')_{i\geq1}
\end{align*}
be three sequences of random vectors with distribution $\pi$.
Further, let $(\vd_{s},\vd_{s}',\vd_{s}'')_{s\in\PM}$ be Poisson variables with $\ex[\vd_s]=td/2$ and $\ex[\vd_s']=\ex[\vd_s'']=(1-t)d/2$.
Finally, let $((\vr_{s,i},\vr_{s,i}',\vr_{s,i}''))_{s\in\PM,i\geq1}$ be uniformly distributed on $\{\pm1\}^3$.
All of these random variables are mutually independent.
Then $\hat\pi\in\cP((0,1)^2)$ is the distribution of the random vector
\begin{align}\label{eqBPtensor_update}
	\Bigg(&
	\frac{2^{-\vd_{-1}-\vd_{-1}'}\prod_{i=1}^{\vec d_{-1}} (1+\vr_{-1,i}(2\MU_{\pi,-1,i,1}-1))\prod_{i=1}^{\vec d_{-1}'}(1+\vr_{-1,i}'(2\vec\mu_{\pi,-1,i,1}'-1))}
	{\sum_{s\in\{\pm1\}}2^{-\vd_{s}-\vd_{s}'}\prod_{i=1}^{\vec d_{s}} (1+\vr_{s,i}(2\vec\mu_{\pi,s,i,1}-1))\prod_{i=1}^{\vec d_{s}'}(1+\vr_{s,i}'(2\MU_{\pi,s,i,1}'-1))} ,\\
	&\qquad		\frac{2^{-\vd_{-1}-\vd_{-1}''}\prod_{i=1}^{\vec d_{-1}} (1+\vr_{-1,i}(2\MU_{\pi,-1,i,2}-1))\prod_{i=1}^{\vec d_{-1}''}(1+\vr_{-1,i}''(2\vec\mu_{\pi,-1,i,2}''-1))}
	{\sum_{s\in\{\pm1\}}2^{-\vd_{s}-\vd_{s}''}\prod_{i=1}^{\vec d_{s}} (1+\vec r_{s,i}(2\vec\mu_{\pi,s,i,2}-1))\prod_{i=1}^{\vec d_{s}''}(1+\vr_{s,i}''(2\MU_{\pi,s,i,2}''-1))}
	\Bigg)\in(0,1)^2.\nonumber
\end{align}
%In addition, define a function $\cB_{d,t}^\tensor:\cP((0,1)^2)\to(0,\infty]$ by letting
%\begin{align}\nonumber
%		\cB^\tensor_{d,t}(\pi)&=
%		\ex\Big[\log\bc{1-(1+\vr_{-1,1}(2\MU_{\pi,-1,1,1}-1))(1+\vr_{-1,2}(2\MU_{\pi,-1,2,1}-1))/4}\\
%							  &\qquad\qquad			\cdot\log\bc{1-(1+\vr_{-1,1}(2\MU_{\pi,-1,1,2}-1))(1+\vr_{-1,2}(2\MU_{\pi,-1,2,2}-1))/4}\Big].
%\label{eqfunctional}
%	\end{align}
%Further, define
%	\begin{align*}
%		\cB^\tensor_{d,t}(\pi)&=
%		\ex\brk{\log\bc{1-\sum_{\rho\in\{\pm1\}}\MU_{\pi,1,1}(\vec s_{1,1},\rho)\sum_{\rho\in\{\pm1\}}\MU_{\pi,2,1}(\vec s_{2,1},\rho)}\log\bc{1-\sum_{\rho\in\{\pm1\}}\MU_{\pi,1,1}(\rho,\vec s_{1,1})\sum_{\rho\in\{\pm1\}}\MU_{\pi,2,1}(\rho,\vec s_{2,1})}},
%	\end{align*}
%provided that the expectation exists.
Let $\uni^\tensor\in\cP((0,1)^2)$ denote the atom on the centre $(\frac12,\frac12)$ of the unit square.
We write $\BP_{d,t}^{\tensor (\ell)}$ %\achat{($\BP_{d,t}^{\tensor,\,\ell}$)}
for the $\ell$-fold application of the operator $\BP_{d,t}^\tensor$.
We are going to perform a fixed point iteration using the operator $\BP_{d,t}^{\tensor}$, starting from $\uni^\tensor$.
This fixed point iteration is known as {\em density evolution} in physics jargon~\cite{MM}.
Let
\begin{align*}
\pi_{d,t}^{(\ell)}=\BP_{d,t}^{\tensor (\ell)}(\uni^\tensor).
\end{align*}

\begin{lemma}\label{lem_trafo}
Let $d \in (0,2)$, $t \in [0,1]$ and set $\pi_{d,t} = \mypsi(\rho_{d,t})$, where $\rho_{d,t}$ is
%$$\pi_{d,t} = \cL\bc{\begin{pmatrix}
		%\frac{1 + \tanh(\vec\xi_{\rho_{d,t}, 1}/2)}{2}\\
		%\frac{1 + \tanh(\vec\xi_{\rho_{d,t},2}/2)}{2}
		%\end{pmatrix}},$$
		%where $\vec \xi_{\rho_{d,t}} = (\vec \xi_{\rho_{d,t},1}, \vec \xi_{\rho_{d,t},2})\sim \rho_{d,t}$ is
		the unique fixed point of $\LLN$ from \Prop~\ref{prop_noela}.

		Then $\pi_{d,t}$ is a fixed point of $\BP_{d,t}^{\tensor}$, and
		\begin{align*}
			\pi_{d,t}&=\lim_{\ell\to\infty}\pi_{d,t}^{(\ell)}.
		\end{align*}
		%$\pi_{d,t} \in \cP((0,1)^2)$ is a fixed point of $ \BP_{d,t}^{\tensor}(\pi_{d,t})$, i.e.  $\pi_{d,t} = \BP_{d,t}^{\tensor}(\pi_{d,t})$.
	\end{lemma}
	\begin{proof}
		For $s\in\PM$ let
		\begin{align}\label{eqNoelasXis}
			(\vec\xi_{\rho,s,i,1},\vec\xi_{\rho,s,i,2})_{i\geq1},
			(\vec\xi_{\rho,s,i,1}',\vec\xi_{\rho,s,i,2}')_{i\geq1},
			(\vec\xi_{\rho,s,i,1}'',\vec\xi_{\rho,s,i,2}'')_{i\geq1}
		\end{align}
		be three sequences of random vectors with distribution $\rho_{d,t}$.
		Further, let $(\vd_{s},\vd_{s}',\vd_{s}'')_{s\in\PM}$ be Poisson variables with $\ex[\vd_s]=td/2$ and $\ex[\vd_s']=\ex[\vd_s'']=(1-t)d/2$ and let $(\vd,\vd',\vd'')$ be Poisson variables with $\ex[\vd]=td$ and $\ex[\vd']=\ex[\vd'']=(1-t)d$.
		Finally, let $((\vr_{s,i},\vr_{s,i}',\vr_{s,i}''))_{s\in\PM,i\geq1}$, $((\vs_{i},\vs_{i}',\vs_{i}''))_{i\geq1}$ and $((\vr_{i},\vr_{i}',\vr_{i}''))_{i\geq1}$ all be uniformly distributed on $\{\pm1\}^3$.
		All of these random variables are mutually independent.
		Throughout the proof, we write $\myphi = (\myphi_1, \myphi_2)$ and $\mypsi = (\mypsi_1, \mypsi_2)$ for $\myphi_h:(0,1) \to \RR$, $\myphi_h(x) = \log \frac{x}{1-x}$, and $\mypsi_h:\RR \to (0,1)$, $\mypsi_h(x) = (1+\tanh(x/2))/2$, where $h \in \{1,2\}$.

		Then, since $\mypsi(\vec\xi_{\rho,s,1,1}, \vec\xi_{\rho,s,1,2}) = (\mypsi_2(\vec\xi_{\rho,s,1,1}), \mypsi_2(\vec\xi_{\rho,s,1,2}))$ has distribution $\mypsi(\rho_{d,t})$, using the definitions of $\myphi$ and $\mypsi,$
		\begin{align*}
			\BP_{d,t}^{\tensor}(\mypsi(\rho_{d,t})) &\disteq
			\Bigg(
			\frac{2^{-\vd_{-1}-\vd_{-1}'}\prod_{i=1}^{\vec d_{-1}} (1+\vr_{-1,i}(2\mypsi_1(\vec\xi_{\rho,-1,i,1})-1))\prod_{i=1}^{\vec d_{-1}'}(1+\vr_{-1,i}'(2\mypsi_1(\vec\xi_{\rho,-1,i,1}')-1))}
			{\sum_{s\in\{\pm1\}}2^{-\vd_{s}-\vd_{s}'}\prod_{i=1}^{\vec d_{s}} (1+\vr_{s,i}(2\mypsi_1(\vec\xi_{\rho,s,i,1})-1))\prod_{i=1}^{\vec d_{s}'}(1+\vr_{s,i}'(2\mypsi_1(\vec\xi_{\rho,s,i,1}')-1))} ,\\
			&\qquad		\frac{2^{-\vd_{-1}-\vd_{-1}''}\prod_{i=1}^{\vec d_{-1}} (1+\vr_{-1,i}(2\mypsi_1(\vec\xi_{\rho,-1,i,2})-1))\prod_{i=1}^{\vec d_{-1}''}(1+\vr_{-1,i}''(2\mypsi_1(\vec\xi_{\rho,-1,i,2}'')-1))}
			{\sum_{s\in\{\pm1\}}2^{-\vd_{s}-\vd_{s}''}\prod_{i=1}^{\vec d_{s}} (1+\vec r_{s,i}(2\mypsi_1(\vec\xi_{\rho,s,i,2})-1))\prod_{i=1}^{\vec d_{s}''}(1+\vr_{s,i}''(2\mypsi_1(\vec\xi_{\rho,s,i,2}'')-1))}
			\Bigg)\\
			%	&= \psi \begin{pmatrix} \sum_{i=1}^{\vec d_{-1}} \log\frac{1+\vr_{-1,i}(2\MU_{\pi,-1,i,1}-1)}{2} + \sum_{i=1}^{\vec d_{-1}'}\log\frac{1+\vr_{-1,i}'(2\psi(\vec\xi_{\pi,-1,i,1}'-1)}{2} -\sum_{i=1}^{\vec d_{1}} \log\frac{1+\vr_{1,i}(2\MU_{\pi,1,i,1}-1)}{2} + \sum_{i=1}^{\vec d_{1}'}\log\frac{1+\vr_{1,i}'(2\psi(\vec\xi_{\pi,1,i,1}'-1)}{2}  \\
				%	\sum_{i=1}^{\vec d_{-1}} \log\frac{1+\vr_{-1,i}(2\MU_{\pi,-1,i,2}-1)}{2} + \sum_{i=1}^{\vec d_{-1}'}\log\frac{1+\vr_{-1,i}''(2\psi(\vec\xi_{\pi,-1,i,2}''-1)}{2} -\sum_{i=1}^{\vec d_{1}} \log\frac{1+\vr_{1,i}(2\MU_{\pi,1,i,2}-1)}{2} + \sum_{i=1}^{\vec d_{1}''}\log\frac{1+\vr_{1,i}''(2\psi(\vec\xi_{\pi,1,i,2}''-1)}{2}
				%		\end{pmatrix} \\
			&\disteq \mypsi \begin{pmatrix} \sum_{i=1}^{\vec d} \vec s_i\log\frac{1+\vr_{i}(2\mypsi_1(\vec\xi_{\rho,i,1})-1)}{2} + \sum_{i=1}^{\vec d'}\vec s_i' \log\frac{1+\vr_{i}'(2\mypsi_1(\vec\xi_{\rho,i,1}')-1)}{2}  \\
				\sum_{i=1}^{\vec d} \vec s_i\log\frac{1+\vr_{i}(2\mypsi_1(\vec\xi_{\rho,i,2})-1)}{2} + \sum_{i=1}^{\vec d''}\vec s_i'' \log\frac{1+\vr_{i}''(2\mypsi_1(\vec\xi_{\rho,i,2}'')-1)}{2}
			\end{pmatrix}\\
			&\disteq \mypsi \begin{pmatrix} \sum_{i=1}^{\vec d} \vec s_i\log\frac{1+\vr_{i}\tanh(\vec\xi_{\rho,i,1}/2)}{2} + \sum_{i=1}^{\vec d'}\vec s_i' \log\frac{1+\vr_{i}'\tanh(\vec\xi_{\rho,i,1}'/2)}{2}  \\
				\sum_{i=1}^{\vec d} \vec s_i\log\frac{1+\vr_{i}\tanh(\vec\xi_{\rho,i,2})/2)}{2} + \sum_{i=1}^{\vec d''}\vec s_i'' \log\frac{1+\vr_{i}''\tanh(\vec\xi_{\rho,i,2}''/2)}{2}
			\end{pmatrix}.
		\end{align*}
		Since %$\varphi_1$ and $\mypsi_1$ are mutually inverse and
		$\rho_{d,t}$ is a fixed point of $\LLN$, the last argument vector of $\mypsi$ has distribution $\rho_{d,t}$, and  we get that
		\begin{align*}
			\BP_{d,t}^{\tensor}(\mypsi(\rho_{d,t})) & = \mypsi(\rho_{d,t}).
		\end{align*}
		So $\mypsi(\rho_{d,t})$ is a fixed point of $\BP_{d,t}^{\tensor}$.

		Next, let $\luni^\tensor\in\cP(\RR^2)$ denote the atom in $(0,0)$. Then since $\mypsi(0,0)=(\frac12,\frac12)$, we have $\uni^\tensor = \mypsi(\luni^\tensor)$. By a computation analogous to the first part of the proof, one can show inductively that for all $\ell \geq 1$, $\pi_{d,t}^{(\ell)} = \BP_{d,t}^{\tensor}(  \mypsi(\luni^\tensor)) =  \mypsi({\LLN}^{\,(\ell)}(\luni^\tensor))$. As %\achat{here}. Moreover, the proof of \Prop~\ref{prop_noela} will show that $\LLN$ is a contraction, so that
		\begin{align*}
			\rho_{d,t}&=\lim_{\ell\to\infty}{\LLN}^{\,(\ell)}(\luni^\tensor),
		\end{align*}
		%\achat{here as well}\\
		the second part of the claim now follows from the continuous mapping theorem.
		%Furthermore, because \nm{the proof of} \Prop~\ref{prop_noela} shows that $\LLN$ is a contraction, we have
		%\begin{align*}
		%	 \rho_{d,t}&=\lim_{\ell\to\infty}{\LLN}^{\,\ell}(\luni^\tensor).
		%\end{align*}
		%Since $\psi(0,0)=(\frac12,\frac12)$, we obtain the desired result.
	\end{proof}
	%\subsection{Proof of \Lem~\ref{lem_trafo}}\label{sec_lem_trafo}

	\subsection{Belief Propagation on the Galton-Watson tree}\label{sec_bp_gw}
	The proof of \Prop~\ref{prop_bp} relies on the fact that Belief Propagation is `exact' on trees.
	The following fact, which is a direct consequence of~\cite[\Thm~14.1]{MM}, furnishes the precise statement that we will use.

	\begin{fact}\label{fact_BP}
		Assume that the bipartite graph associated with the 2-CNF $\Phi$ is a (finite) tree.
		Let $z\in V(\Phi)$ be a variable and let $\ell\geq1$ be an integer such that no variable or clause of $\Phi$ has distance greater than $2\ell$ from $z$.
		Let
		\begin{align*}
			\mu_{x\to a}^{(0)}(s)&=\mu_{a\to x}^{(0)}(s)=\frac12&&\mbox{for all }x\in V(\Phi),\,a\in\partial x,\,s\in\PM.
		\end{align*}
		Furthermore, obtain the messages $(\mu_{x\to a}^{(i+1)}(s),\mu_{a\to x}^{(i+1)}(s))_{x,a,s}$ by applying the BP operator~\eqref{eqBPop} to \\ $(\mu_{x\to a}^{(i)}(s),\mu_{a\to x}^{(i)}(s))_{x,a,s}$.
		Then for all $i\geq2\ell$ we have
		\begin{align}\label{eq_fact_BP}
			\mu_{\Phi}(\{\SIGMA_z=s\})&=\frac{\prod_{a\in\partial z}\mu_{a\to z}^{(i)}(s)}{\prod_{a\in\partial z}\mu_{a\to z}^{(i)}(1)+\prod_{a\in\partial z}\mu_{a\to z}^{(i)}(-1)}.
		\end{align}
	\end{fact}

	As a preparation toward the proof of \Prop~\ref{prop_bp} we establish the following `univariate' variant of the proposition.

	\begin{lemma}\label{lem_rs_uni}
		Let $h=1,2$.
		Let $\pi_{d,t,h}^{(\ell)}$ be the distribution of the $h$-th component of a random vector with distribution $\pi_{d,t}^{(\ell)}$.
		Then $\mu_{\vT_h^{(2\ell)}}(\{\SIGMA_o=1\})$ has distribution $\pi_{d,t,h}^{(\ell)}$.
	\end{lemma}
	\begin{proof}
		We proceed by induction on $\ell$.
		For $\ell=0$ there is nothing to show because both $\mu_{\vT_h^{(0)}}(\{\SIGMA_o=1\})=\frac12$ and $\pi_{d,t, h}^{(\ell)}$ is the atom on $1/2$.
		To go from $\ell-1$ to $\ell\geq1$ we exploit the fact that $\vT_h$ by itself has the same distribution as the `plain' Galton-Watson tree $\vT$ from \Sec~\ref{sec_BPop}, after all distinctions between different types of clauses and variables are dropped.
		In effect, the tree $\vT_{h,x}$ pending on any grandchild $x\in\partial^2o$ of the root has the same distribution as $\vT$ itself, and these trees are mutually independent for all $x\in\partial^2o$.
		Consequently, by induction we know that the marginal $\mu_{\vT_{h,x}^{(2(\ell-1))}}(\{\SIGMA_x=1\})$ of $x$ in $\vT_{h,x}^{(2(\ell-1))}$ has distribution $\pi_{d,t,h}^{(\ell-1)}$.

		Now let $a_x\in\partial x\cap\partial o$ be the clause that links $x$ and $o$.
		Fact~\ref{fact_BP} implies that the marginals $\mu_{\vT_{h,x}^{(2(\ell-1))}}(\{\SIGMA_x=s\})$ coincide with the messages $\mu_{\vT^{(2\ell)},x\to a_x}^{(2(\ell-1))}(s)$.
		Indeed, the marginal formula~\eqref{eq_fact_BP} for the tree $\vT_{h,x}^{(2(\ell-1))}$ coincides with the message update formula \eqref{eqBPop2}, because clause $a_x$ is not part of $\vT_{h,x}$.
		Furthermore, because the trees pending on the different grandchildren of the root $o$ are mutually independent, the incoming messages $(\mu_{\vT_{h,x}^{(2\ell)},x\to a_x}^{(2(\ell-1))}(\pm1))_{x\in\partial^2o}$ are mutually independent.
		Moreover, the quotient from \eqref{eq_fact_BP}, which, upon substituting in the update equation~\eqref{eqBPop1}, can be rewritten as
		\begin{align}\label{eq_lem_rs_uni_1}
			\mu_{\vT_h^{(2\ell)}}(\{\SIGMA_o=1\})&=\frac{\prod_{x\in\partial^2 o}\vecone\{\sign(o,a_x)=1\}+\vecone\{\sign(o,a_x)=-1\}\mu_{\vT_{h,x}^{(2\ell)},x\to a_x}^{(2(\ell-1))}(\sign(x,a_x))}{\sum_{s\in\PM}\prod_{x\in\partial^2 o}\vecone\{\sign(o,a_x)=s\}+\vecone\{\sign(o,a_x)=-s\}\mu_{\vT_{h,x}^{(2\ell)},x\to a_x}^{(2(\ell-1))}(\sign(x,a_x))}.
		\end{align}
		Also recall that $o$ has $\Po(d)$ children.
		Hence, comparing~\eqref{eq_lem_rs_uni_1} with the $h$-component of \eqref{eqBPtensor_update}, we conclude that $\mu_{\vT_h^{(2\ell)}}(\{\SIGMA_o=1\})$ has distribution $\pi_{d,t,h}^{(\ell)}$.
	\end{proof}

	\begin{lemma}\label{lem_gw_bp}
		Let $t\in[0,1]$.
		Then $\MU^{(2\ell)}$ has distribution $\pi_{d,t}^{(\ell)}$.
	\end{lemma}
	\begin{proof}
		As in the proof of \Lem~\ref{lem_rs_uni} we proceed by induction on $\ell$.
		For $\ell=0$ there is nothing to show.
		To go from $\ell-1$ to $\ell\geq1$ we reuse the rewritten update equation~\eqref{eq_lem_rs_uni_1}.
		In each of the correlated trees $\vT_1,\vT_2$ the root $o$ has $\Po((1-t)d)$ $\{1,2\}$-distinct grandchildren.
		By construction, the trees pending on these grandchildren are mutually independent copies of the tree $\vT$.
		Hence, the same consideration as in the proof of \Lem~\ref{lem_rs_uni} shows that the messages that the $\{1,2\}$-distinct grandchildren pass up are independent with distribution $\pi_{d,t,h}^{(\ell-1)}$.
		%\nm{Maybe not literally \Lem~\ref{lem_rs_uni}, because this was for the tree with shared and distinct clauses and a shared variable as a root? Maybe it would be good to make a remark about ignoring the types somewhere?}
		The same is true of the messages $\vec\mu'_{\pi_{d,t}^{(\ell-1)},\pm1,i,h}, \vec\mu''_{\pi_{d,t}^{(\ell-1)},\pm1,i,h}$ from~\eqref{eqBPtensor_update}.
		Consequently, the contribution of the $\{1,2\}$-distinct children in~\eqref{eqBPtensor_update} matches the corresponding contribution to the update equation \eqref{eq_lem_rs_uni_1}.

		With respect to the shared grandchildren, we apply induction as in the proof of \Lem~\ref{lem_rs_uni}.
		Indeed, the trees pending on the shared grandchildren $x\in\partial_{\vT^\tensor}^2o$ have the same distribution as the original tree $\vT^\tensor$ and are mutually independent.
		Therefore, by the induction hypothesis, the pair of messages $(\mu_{\vT_h^{(2\ell-2)},x\to a_x}^{(2(\ell-1))}(1))_{h=1,2}$ that a shared grandchild $x$ sends towards the root has distribution $\pi_{d,t}^{(\ell-1)}$.
		Finally, since $o$ has $\Po(dt)$ shared grandchildren, matching the expressions \eqref{eq_lem_rs_uni_1} and~\eqref{eqBPtensor_update} completes the proof.
	\end{proof}

	\begin{proof}[Proof of \Prop~\ref{prop_bp}]
		The assertion follows from \Lem s~\ref{lem_trafo} and~\ref{lem_gw_bp}.
	\end{proof}

\section{Proof that the variance is finite}\label{sec_no_prop}
The main goal of this section is to show that both the evaluation of the functional $\cB_{d,t}^{\tensor}$ on $\rho_{d,t}$ as well as the integration to obtain $\eta(d)^2$ yield finite values for any $d \in (0,2)$ and $t \in [0,1]$.

\begin{lemma}\label{lem_bethe_fin}
	For any $d \in (0,2)$ and $t \in [0,1]$, $\cB_{d,t}^{\tensor}(\rho_{d,t}) < \infty$. Moreover, for any $d \in (0,2)$, $\eta(d)^2 < \infty$.
\end{lemma}

\subsection{Proof of \Lem~\ref{lem_bethe_fin}}
Let $\rho_{d,t}^{(\ell)} \in \cW_2(\RR^2)$ be the result of $\ell$ iterations of $\LLN$ launched from $\luni^{\tensor}$, the atom at $(0,0)$. In the proof of \Lem~\ref{lem_bethe_fin}, the following properties of the fixed point $\pi_{d,t}$ will be used:

\begin{claim}\label{claim_sym_prop}
	Let $\pi_{d,t} = \mypsi(\rho_{d,t})$ %be the unique distribution in $\cW_2((0,1)^2)$ with $\pi_{d,t} = \BP_{d,t}^{\tensor}(\pi_{d,t})$
	and $\vec \mu_{\pi_{d,t}} = (\vec \mu_{\pi_{d,t},1}, \vec\mu_{\pi_{d,t},2})$ be a random vector with distribution $\pi_{d,t}$. Then
	\begin{align*}
		\vec \mu_{\pi_{d,t},1} \disteq \vec\mu_{\pi_{d,t},2} \qquad \text{ and } \qquad \vec\mu_{\pi_{d,t}1}\disteq 1 - \vec\mu_{\pi_{d,t},1}.
	\end{align*}
\end{claim}

\begin{proof}
	Recall the definition of $\BP_{d,t}^{\tensor}$ from \eqref{eqBPtensor_update}. %For $s\in\PM$ let
	%\begin{align*}
	%	(\MU_{\pi,s,i,1},\MU_{\pi,s,i,2})_{i\geq1},
	%	(\MU_{\pi,s,i,1}',\MU_{\pi,s,i,2}')_{i\geq1},
	%	(\MU_{\pi,s,i,1}'',\MU_{\pi,s,i,2}'')_{i\geq1}
	%\end{align*}
	%be three sequences of random vectors with distribution $\pi$.
	%Further, let $(\vd_{s},\vd_{s}',\vd_{s}'')_{s\in\PM}$ be Poisson variables with $\ex[\vd_s]=td/2$ and $\ex[\vd_s']=\ex[\vd_s'']=(1-t)d/2$.
	%Finally, let $(\vr_{s,i},\vr_{s,i}',\vr_{s,i}'')_{s\in\PM,i\geq1}$ be uniformly distributed on $\{\pm1\}$.
	%All of these random variables are mutually independent. Then $\hat\pi\in\cP((0,1)^2)$ is the distribution of the random vector
	%\begin{align}\label{eq_claim_distribution_hatpi}
	%	\Bigg(&
	%		\frac{2^{-\vd_{-1}-\vd_{-1}'}\prod_{i=1}^{\vec d_{-1}} (1+\vr_{-1,i}(2\MU_{\pi,-1,i,1}-1))\prod_{i=1}^{\vec d_{-1}'}(1+\vr_{-1,i}'(2\vec\mu_{\pi,-1,i,1}'-1))}
	%		{\sum_{s\in\{\pm1\}}2^{-\vd_{-s}-\vd_{-s}'}\prod_{i=1}^{\vec d_{-s}} (1+\vr_{-s,i}(2\vec\mu_{\pi,-s,i,1}-1))\prod_{i=1}^{\vec d_{-s}'}(1+\vr_{-s,i}'(2\MU_{\pi,-s,i,1}'-1))} ,\\
	%		 &\qquad		\frac{2^{-\vd_{-1}-\vd_{-1}''}\prod_{i=1}^{\vec d_{-1}} (1+\vr_{-1,i}(2\MU_{\pi,-1,i,2}-1))\prod_{i=1}^{\vec d_{-1}''}(1+\vr_{-1,i}''(2\vec\mu_{\pi,-1,i,2}''-1))}
	%		{\sum_{s\in\{\pm1\}}2^{-\vd_{-s}-\vd_{-s}''}\prod_{i=1}^{\vec d_{-s}} (1+\vec r_{-s,i}(2\vec\mu_{\pi,-s,i,2}-1))\prod_{i=1}^{\vec d_{-s}''}(1+\vr_{-s,i}''(2\MU_{\pi,-s,i,2}''-1))}
	%	\Bigg)\in(0,1)^2.\nonumber
	%\end{align}
	The first claim then follows from the following limiting argument: By \Lem~\ref{lem_trafo},  %and Claim~\ref{claim_dirac} imply
	$\pi_{d,t} = \lim_{\ell \to \infty} \BP_{d,t}^{\tensor (\ell)}(\uni^\tensor)$, where $\uni^\tensor$ is the Dirac measure on $(1/2,1/2)$.
	%Let us start from the point in which the distributions $\vec \mu_{\pi,s,i,1}$ and $\vec\mu_{\pi,s,i,2}$ are uniform
	As the marginal distributions of the initial distribution $\uni^\tensor$ are identical,  inspection of the update rule (\ref{eqBPtensor_update}) yields that also the marginal distributions of $ \BP_{d,t}^{\tensor (1)}(\uni^\tensor)$ are identical. Analogously, it is immediate from  (\ref{eqBPtensor_update}) that any $\pi$ with two identical marginals will be mapped to a measure $\BP_{d,t}^{\tensor}(\pi)$ with two identical marginals, such that the marginal distributions of $ \BP_{d,t}^{\tensor (\ell)}(\uni^\tensor)$ for any $\ell \geq 0$ are identical. Hence, also in the limit,
	$$ \vec \mu_{\pi_{d,t},1} \disteq \vec\mu_{\pi_{d,t},2}. $$
	On the other hand,  \Lem~\ref{lem_trafo} also implies that $\pi_{d,t} = \BP_{d,t}^{\tensor}(\pi_{d,t})$, so that the distribution of $\vec \mu_{\pi_{d,t},1}$ is the same as the distribution of
	\begin{align}\label{eq_mupidt1}
		\frac{2^{-\vd_{-1}-\vd_{-1}'}\prod_{i=1}^{\vec d_{-1}} (1+\vr_{-1,i}(2\MU_{\pi_{d,t},-1,i,1}-1))\prod_{i=1}^{\vec d_{-1}'}(1+\vr_{-1,i}'(2\vec\mu_{\pi_{d,t},-1,i,1}'-1))}
		{\sum_{s\in\{\pm1\}}2^{-\vd_{-s}-\vd_{-s}'}\prod_{i=1}^{\vec d_{-s}} (1+\vr_{-s,i}(2\vec\mu_{\pi_{d,t},-s,i,1}-1))\prod_{i=1}^{\vec d_{-s}'}(1+\vr_{-s,i}'(2\MU_{\pi_{d,t},-s,i,1}'-1))}
	\end{align}
	while the distribution of $1-\vec \mu_{\pi_{d,t},1}$ is the same as the distribution of
	\begin{align}\label{eq_mupidt2}
		\frac{2^{-\vd_{1}-\vd_{1}'}\prod_{i=1}^{\vec d_{1}} (1+\vr_{1,i}(2\MU_{\pi_{d,t},1,i,1}-1))\prod_{i=1}^{\vec d_{1}'}(1+\vr_{1,i}'(2\vec\mu_{\pi_{d,t},1,i,1}'-1))}
		{\sum_{s\in\{\pm1\}}2^{-\vd_{-s}-\vd_{-s}'}\prod_{i=1}^{\vec d_{-s}} (1+\vr_{-s,i}(2\vec\mu_{\pi_{d,t},-s,i,1}-1))\prod_{i=1}^{\vec d_{-s}'}(1+\vr_{-s,i}'(2\MU_{\pi_{d,t},-s,i,1}'-1))}\enspace.
	\end{align}
	This immediately shows that (\ref{eq_mupidt1}) and (\ref{eq_mupidt2}) have the same distribution. As a consequence, the second claim holds as well.
\end{proof}

\begin{proof}[Proof of \Lem~\ref{lem_bethe_fin}]
	Recall that $\pi_{d,t}=\mypsi(\rho_{d,t})$. Let $$\vec\mu_{\pi_{d,t},1} = (\vec\mu_{\pi_{d,t},1,1}, \vec\mu_{\pi_{d,t},1,2})= \mypsi(\vec\xi_{\rho_{d,t},1,1}, \vec\xi_{\rho_{d,t},1,2}),$$
	and
	$$\vec\mu_{\pi_{d,t},2} = (\vec\mu_{\pi_{d,t},2,1}, \vec\mu_{\pi_{d,t},2,2}) = \mypsi(\vec\xi_{\rho_{d,t},2,1}, \vec\xi_{\rho_{d,t},2,2}).$$
	Then they are independent random vectors with distribution $\pi_{d,t}$. Let $\vec r_{1}$, $\vec r_{2}$ be independent Rademacher random variables with parameter $1/2$ , independent of $\vec\mu_{\pi_{d,t},1}$ and $\vec\mu_{\pi_{d,t},2}$. Conditioning on the values of $\vec r_{1}$ and $\vec r_{2}$ yields the upper bound
	\begin{align*}
		|\cB_{d,t}^{\tensor}(\rho_{d,t})| &\leq \frac{1}{4}\Erw\brk{\abs{\log\bc{1-\vec\mu_{\pi_{d,t},1,1}\vec\mu_{\pi_{d,t},2,1}}\log\bc{1-\vec\mu_{\pi_{d,t},1,2}\vec\mu_{\pi_{d,t},2,2}}}} \\
		&+ \frac{1}{4}\Erw\brk{\abs{\log\bc{1-\bc{1-\vec\mu_{\pi_{d,t},1,1}}\vec\mu_{\pi_{d,t},2,1}}\log\bc{1-\bc{1-\vec\mu_{\pi_{d,t},1,2}}\vec\mu_{\pi_{d,t},2,2}}}} \\
		&+ \frac{1}{4}\Erw\brk{\abs{\log\bc{1-\vec\mu_{\pi_{d,t},1,1}\bc{1-\vec\mu_{\pi_{d,t},2,1}}}\log\bc{1-\vec\mu_{\pi_{d,t},1,2}\bc{1-\vec\mu_{\pi_{d,t},2,2}}}}} \\
		&+ \frac{1}{4}\Erw\brk{\abs{\log\bc{1-\bc{1-\vec\mu_{\pi_{d,t},1,1}}\bc{1-\vec\mu_{\pi_{d,t},2,1}}}\log\bc{1-\bc{1-\vec\mu_{\pi_{d,t},1,2}}\bc{1-\vec\mu_{\pi_{d,t},2,2}}}}}.
	\end{align*}
	The Cauchy-Schwarz inequality further gives that
	\begin{align*}
		|\cB_{d,t}^{\tensor}(\rho_{d,t})| &\leq \frac{1}{4}\Erw\brk{\log^2\bc{1-\vec\mu_{\pi_{d,t},1,1}\vec\mu_{\pi_{d,t},2,1}}}^{1/2} \Erw\brk{\log^2\bc{1-\vec\mu_{\pi_{d,t},1,2}\vec\mu_{\pi_{d,t},2,2}}}^{1/2} \\
		&+ \frac{1}{4}\Erw\brk{\log^2\bc{1-\bc{1-\vec\mu_{\pi_{d,t},1,1}}\vec\mu_{\pi_{d,t},2,1}}}^{1/2}\Erw\brk{\log^2\bc{1-\bc{1-\vec\mu_{\pi_{d,t},1,2}}\vec\mu_{\pi_{d,t},2,2}}}^{1/2} \\
		&+ \frac{1}{4}\Erw\brk{\log^2\bc{1-\vec\mu_{\pi_{d,t},1,1}\bc{1-\vec\mu_{\pi_{d,t},2,1}}}}^{1/2}\Erw\brk{\log^2\bc{1-\vec\mu_{\pi_{d,t},1,2}\bc{1-\vec\mu_{\pi_{d,t},2,2}}}}^{1/2} \\
		&+ \frac{1}{4}\Erw\brk{\log^2\bc{1-\bc{1-\vec\mu_{\pi_{d,t},1,1}}\bc{1-\vec\mu_{\pi_{d,t},2,1}}}}^{1/2}\Erw\brk{\log^2\bc{1-\bc{1-\vec\mu_{\pi_{d,t},1,2}}\bc{1-\vec\mu_{\pi_{d,t},2,2}}}}^{1/2}.
	\end{align*}
	As $\vec\mu_{\pi_{d,t},1,1}\disteq \vec\mu_{\pi_{d,t},1,2}$ and $\vec\mu_{\pi_{d,t},1,1}\disteq 1 - \vec\mu_{\pi_{d,t},1,1}$ thanks to Claim \ref{claim_sym_prop}, we further get
	\begin{align*}
		|\cB_{d,t}^{\tensor}(\rho_{d,t})| &\leq \Erw\brk{\log^2\bc{1-\vec\mu_{\pi_{d,t},1,1}\vec\mu_{\pi_{d,t},2,1}}}.
	\end{align*}
	Next, recalling~\eqref{eqNoelasXis},
	\begin{align*}
		& \Erw\brk{\log^2\bc{1-\vec\mu_{\pi_{d,t},1,1}\vec\mu_{\pi_{d,t},2,1}}}  \leq \Erw\brk{\log^2\bc{1-\vec\mu_{\pi_{d,t},1,1}}} \\
		&\leq  \Erw\brk{\vecone\cbc{\vec \mu_{\pi_{d,t},1,1} \leq \frac{1}{2}} \abs{\log^2\bc{1-\vec \mu_{\pi_{d,t},1,1}}}} + \Erw\brk{\vecone\cbc{\vec \mu_{\pi_{d,t},1,1}> \frac{1}{2}} \abs{\log^2\bc{1-\vec \mu_{\pi_{d,t},1,1}}}} \\
		& \leq  \Erw\brk{\vecone\cbc{\vec \mu_{\pi_{d,t},1,1} \leq \frac{1}{2}} \log^2 2 } + \Erw\brk{\vecone\cbc{\vec \mu_{\pi_{d,t},1,1} > \frac{1}{2}} \bc{\log\frac{\vec \mu_{\pi_{d,t},1,1} }{\bc{1-\vec \mu_{\pi_{d,t},1,1} }} - \log\vec \mu_{\pi_{d,t},1,1} }^2} \\
		&\leq \log^2 2 + \Erw\brk{\log^2\bc{\frac{\vec \mu_{\pi_{d,t},1,1} }{1 - \vec \mu_{\pi_{d,t},1,1} }}} - 2 \Erw\brk{\vecone\cbc{\vec \mu_{\pi_{d,t},1,1} > \frac{1}{2}}\log\frac{\vec \mu_{\pi_{d,t},1,1} }{\bc{1-\vec \mu_{\pi_{d,t},1,1} }} \log\vec \mu_{\pi_{d,t},1,1}}\\
		& \leq \log^2 2 + \Erw\brk{\log^2\bc{\frac{\vec \mu_{\pi_{d,t},1,1} }{1 - \vec \mu_{\pi_{d,t},1,1} }}} = \log^22 +  \Erw\brk{\vec \xi_{\rho_{d,t},1,1}^2}.
	\end{align*}
	However $ \Erw[\vec \xi_{\rho_{d,t},1,1}^2] < \infty$, since $\rho_{d,t} \in \cW_2(\RR^2)$. Thus, $\cB_{d,t}^{\tensor}(\rho_{d,t}) < \infty$.
	
	Moreover, it is easy to see that the distribution of the marginal sample $\vec \xi_{\rho_{d,t},1,1}$ is independent of for any $t \in [0,1]$. Call this distribution $\rho_d$ and let $\vec \xi_{\rho_{d}}$ be a sample from $\rho_d$. Then the previous upper bound yields
	\begin{align*}
	\eta(d)^2 \leq \int_0^1 |\cB_{d,t}^{\tensor}(\rho_{d,t})| \mathrm{d}t + |\cB_{d}^{\tensor}(\rho_{d,0})| \leq  \int_0^1\log^22 +  \Erw\brk{\vec \xi_{\rho_{d}}^2} \mathrm{d}t + \log^22 +  \Erw\brk{\vec \xi_{\rho_{d}}^2} \leq 2 \bc{\log^22 +  \Erw\brk{\vec \xi_{\rho_{d}}^2}} < \infty.
	\end{align*}
\end{proof}

	\section{Proof of \Prop~\ref{prop_varproc}}\label{sec_prop_var}

	\noindent
	We combine the results from the previous sections in order to analyse the variance process.
	As a first step we derive a rough upper bound on the potential change in the number of satisfying assignments upon insertion of a single clause (\Lem s~\ref{lem_forced_add} and~\ref{lem_tail}).
	Subsequently we derive a combinatorial formula for the squared martingle difference (\Lem~\ref{lem_comb_proc}), which easily implies \Lem~\ref{lem_comb_rgb}.
	The combinatorial formula puts us in a position to obtain an $L^2$-bound on the squared martingale difference (\Lem~\ref{lem_comb_l2}).
	With these ingredients in place, we complete the proof of \Prop~\ref{prop_varproc} and of \Thm~\ref{thm_clt} in \Sec~\ref{sec_var_proc}.

	\subsection{A pessimistic estimate}\label{sec_prune}

	Let $\Phi,\Psi$ be two 2-CNFs on the same set of variables such that $\Psi$ is obtained from $\Phi$ by adding a single clause $e$.
	We are going to need a baseline estimate of the difference $|\log Z(\hat\Phi)-\log Z(\hat\Psi)|$.
	The principal difficulty here is to assess the impact of the additional clause on the pruning operation.
	The issue is that the extra clause may also cause additional pruning.
	Indeed, while clearly
	\begin{align*}
		F(\hat\Psi)\setminus\{e\}\subset F(\hat\Phi),
	\end{align*}
	i.e., any clause $a\neq e$ that survives pruning on $\Psi$ also remains present in $\hat\Phi$, the pruned formula $\hat\Psi$ may ironically end up having strictly fewer clauses than $\hat\Phi$.

	To get a handle on the potential repercussions of pruning, let $\{v,v'\}=\partial e$ be the variables that appear in clause $e$.
	Let $\cN(\Phi,v)$ be the set of all literals $l$ such that $\{v,\neg v\}\cap\cL(\Phi,\cbc l)\not= \emptyset$.
	Thus, \UCP\ may reach $v$ or $\neg v$ once $l$ is deemed true.
	Observe that $v\in\cN(\Phi,v)$ and $\neg v\in\cN(\Phi,v)$.
	Define $\cN(\Phi,v')$ analogously.
	Further, let
	\begin{align}\label{eqLphil}
		\cN(\Phi,e)=\bigcup_{l\in\cN(\Phi,v)\cup\cN(\Phi,v')}\cL(\Phi,\cbc l).
	\end{align}
	Thus, $\cN(\Phi,e)$ contains all literals that \UCP\ can reach by tracing the implications of a literal from $\cN(\Phi,v)\cup\cN(\Phi,v')$.
	The definition of the sets $\cN(\Phi,v)$, $\cN(\Phi,v')$ ensures that
	\begin{align}\label{eqclaim_PhiPsi1}
		v,\neg v&\in \cN(\Phi,v),&v',\neg v'&\in\cN(\Phi,v').
	\end{align}

	\begin{lemma}\label{lem_forced_add}
		Let $\Phi$ be a 2-CNF formula.
		Suppose that $\Psi$ is obtained from $\Phi$ by adding a single clause $e$.
		Then
		\begin{align}\label{eqlem_forced_add}
			\abs{\log(Z(\hat\Phi))-\log(Z(\hat\Psi))}\leq\abs{\cN(\Phi,e)}\log 2.
		\end{align}
	\end{lemma}

	As a first step towards the proof of \Lem~\ref{lem_forced_add} we observe that \eqref{eqLphil} can be rewritten as follows.

	\begin{claim}\label{claim_PhiPsi}
		We have
		\begin{align}\label{eqLpsil}
			\cN(\Phi,e)=\bigcup_{l\in\cN(\Phi,v)\cup\cN(\Phi,v')}\cL(\Psi,\cbc l).
		\end{align}
	\end{claim}
	\begin{proof}
		Clearly $\cL(\Psi,\cbc l)\supseteq\cL(\Phi,\cbc l)$ for every literal $l$.
		Hence, we just need to show that
		\begin{align}\label{eqLpsil_1}
			\cL(\Psi,\cbc l)\subseteq\cN(\Phi,e)&&\mbox{for all }l\in\cN(\Phi,v)\cup\cN(\Phi,v').
		\end{align}
		Hence, let $l\in\cN(\Phi,v)\cup\cN(\Phi,v')$ and let $l'\in\cL(\Psi,\cbc l)$.
		Then \Lem~\ref{lem_ucp_reachable} shows that there exists an implication chain
		\begin{align}\label{eqLpsil_2}
			l=l_0,a_1,l_1,\ldots,a_k,l_k=l'
		\end{align}
		comprising literals $l_i$ and clauses $a_i\in F(\Psi)$ such that $a_i\equiv  l_{i-1}\to l_i$ for all $1\leq i\leq k$.
		If $a_i\neq e$ for all $i$, then the chain~\eqref{eqLpsil_2} is contained in $\Phi$ and thus $l'\in\cL(\Phi,\{l\})\subseteq\cN(\Phi,e)$.
		Otherwise let $1\leq j\leq k$ be the largest index such that $a_j=e$.
		Then $l_j$ is one of the constituent literals of $a_j$ and thus $l_j\in\{v,\neg v,v',\neg v'\}$.
		Furthermore, the implication chain $l_j,a_{j+1},l_{j+1},\ldots,a_k,l_k=l'$ from $l_j$ to $l'$ is contained in $\Phi$.
		Therefore, \eqref{eqclaim_PhiPsi1} shows in combination with \Lem~\ref{lem_ucp_reachable} and~\eqref{eqLphil} that $l'\in\cL(\Phi,\{l_j\})\subset\cN(\Phi,e)$.
	\end{proof}

	We proceed to show that $\cN(\Phi,e)$ contains the variables of all clauses $a\in F(\Phi)$ on which the pruning processes run on $\Phi,\Psi$ differ.

	\begin{claim}\label{claim_forced_add_a}
		For any clause $a\in F(\hat\Phi)\setminus F(\hat\Psi)$ we have $\partial a\subset \cN(\Phi,e)$.
	\end{claim}
	\begin{proof}
		Consider a clause $a\in F(\Phi)$ that was removed by pruning applied to $\Psi$ but not by pruning applied to $\Phi$.
		Let $w,w'$ be the constituent literals of $a$, i.e., $ a\equiv w\vee w'$.
		Then \UCP$(\Psi,\cbc l)$ added $a$ to the set $\cC(\Psi,\cbc l)$ of conflict clauses for some literal $l$.
		Hence,
		\begin{align}\label{eq_claim_forced_add_a_1}
			w,\neg w,w',\neg w'\in\cL(\Psi,\cbc l).
		\end{align}
		Consequently, \Lem~\ref{lem_ucp_reachable} shows that for each literal $k\in\{\neg w,\neg w'\}$, \UCP$(\Psi,\{l\})$ traverses an implication chain
		$$l_{0,k}=l,a_{0,k},l_{1,k},a_{1,k},\ldots,l_{j_k,k} =k$$
		of literals $l_{i,k}\in\cL(\Psi,\cbc l)$ and clauses $a_{i,k}\equiv\neg l_{i,k}\vee l_{i+1,k}\equiv l_{i,k}\to l_{i+1,k}$ for $0\leq i<j_k$.
		Because $a\not\in\cC(\Phi,l)$, at least one of these two sequences includes the clause $e$ and thus at least one of $v,\neg v$ and one of $v',\neg v'$.
		Hence, $l\in\cN(\Phi,v)\cup\cN(\Phi,v')$.
		Therefore, combining (\ref{eqLpsil}) and~\eqref{eq_claim_forced_add_a_1}, we conclude that $\partial a=\{|w|,|w'|\}\subset\cN(\Phi,e)$.
		%	Let $k'\in\{w,\neg w\}\setminus\{\neg k\}$ be the second literal of $a$.
		%	Since $k\in\cL(\Psi,\{l\})$ and $e\in F(\Psi)$, \Lem~\ref{lem_ucp_reachable} implies that $k'\in\cL(\Psi,\{l\})$ as well.
		%	Finally, since the clause
		%	Since also $v,\neg v,v',\neg v'\in\cN(\Phi,v)\cup\cN(\Phi,v')$, we conclude that $l_{i,k}\in\cN(\Phi,e)$ for all $i,k$.
		%	Thus, $\partial a\subset\cN(\Phi,e)$, as claimed.
	\end{proof}

	Let $\tilde\Phi$ be the formula obtained from $\hat\Phi$ by removing all variables $x\in V(\hat\Phi)$ such that $\{x,\neg x\}\cap\cN(\Phi,e)\neq\emptyset$ along with their adjacent clauses.

	\begin{claim}\label{claim_forced_sat}
		For any $\tilde\sigma\in S(\tilde\Phi)$ there exists $\sigma\in S(\hat\Psi)$ such that $\sigma_x=\tilde\sigma_x$ for all $x\in V(\tilde\Phi)$.
	\end{claim}
	\begin{proof}
		Let $\check\Psi$ be a CNF with variables $$V(\check\Psi)=V(\hat\Psi)\setminus V(\tilde\Phi)=\{x\in V(\hat\Phi):\{x,\neg x\}\cap\cN(\Phi,e)\neq\emptyset\}.$$
		The clauses of $\check\Psi$ include all $a\in F(\hat\Psi)$ such that $\partial a\subset V(\check\Psi)$.
		Additionally, for every clause $a\in F(\hat\Psi)$ that contains exactly one literal $l$ with  $|l|\in V(\check\Psi)$ we include the literal $l$ as a unit clause into $\check\Psi$.
		In light of Claim~\ref{claim_forced_add_a}, to prove the assertion it suffices to show that $\check\Psi$ is satisfiable.
		For then we could extend any $\sigma\in S(\tilde\Phi)$ to a satisfying assignment of $\hat\Psi$ by simply setting the variables $x\in V(\hat\Psi)\setminus V(\tilde\Phi)$ in accordance with a satisfying assignment of $\check\Psi$.

		As in the proof of Fact~\ref{lem_hatphi_sat}, to construct a satisfying assignment of $\check\Psi$ we fix an order $l_1,\ldots,l_k$ of the literals $\cN(\Phi,e)$.
		Let $\sigma_i$ be the assignment that \UCP\ outputs on input $\Psi,\{l_i\}$.
		Further, define a $\{0,\pm1\}$-valued assignment $(\sigma_x)_{x\in V(\check\Psi)}$ by letting $\sigma_x=\sigma_{i,x}$ for the least index $i$ such that $\{x,\neg x\}\cap\cL(\Psi,l_i)\neq\emptyset$.

		We claim that
		\begin{align}\label{eq_claim_forced_sat_1}
			\forall a\in F(\check\Psi)\,\exists x\in\partial_{\check\Psi}a\,:\,\sigma_x=\sign(x,a)\enspace;
		\end{align}
		thus, we can turn $\sigma$ into a satisfying assignment of $\check\Psi$ by assigning those variables $y$ with $\sigma_y=0$ arbitrarily.
		To verify \eqref{eq_claim_forced_sat_1}, we consider two cases separately.
		\begin{description}
			\item[Case 1: $|\partial_{\check\Psi} a|=2$]
			then $a\in F(\Psi)$.
			Let $\partial a=\{x,x'\}$ and let $i$ be the smallest index such that $\cL(\Psi,l_i)\cap\{x,\neg x,x',\neg x'\}\neq\emptyset$.
			Also let $l,l'$ be the constitutent literals of $a$ such that $|l|=x$ and $|l'|=x'$.
			Suppose that $l\in\cL(\Psi,\{l_i\})$.
			If $\neg l\not\in\cL(\Psi,l_i)$, then $\sigma_x=\sign(x,a)$ by construction.
			Hence, assume that $l,\neg l\in\cL(\Psi,l_i)$.
			Then the construction in Steps 1--2 of \UCP\ ensures that $l'\in\cL(\Psi,l_i)$ as well.
			Moreover, if $\neg l'\not\in\cL(\Psi,l_i)$, then $\sigma_{x'}=\sign(x',a)$.
			Finally, the case $l,l',\neg l,\neg l'\in\cL(\Psi,l_i)$ cannot occur because otherwise $a$ would have been pruned, i.e., $a\not\in F(\hat\Psi)$.
			\item[Case 2: $|\partial_{\check\Psi}a|=1$] there exists a clause $b\in F(\hat\Psi)$ and literals $l,l'$ with $|l'|\not\in V(\check\Psi)$ such that $b=l\vee l'$ and $a=l$.
			Let $i$ be the least index such that $\{l,\neg l\}\cap\cL(\Psi,\{l_i\})\neq\emptyset$.
			If $\neg l\in\cL(\Psi,\{l_i\})$, then \UCP$(\Psi,\{l_i\})$ would have added $l'$ to the set $\cL(\Psi,\{l_i\})$ as well and thus $|l'|\in V(\check\Psi)$.
			But $|l'|\not\in V(\check\Psi)$.
			Hence, $\{l,\neg l\}\cL(\Psi,\{l_i\})=\{l\}$ and thus $\sigma_{|l|}=\sigma_{i,|l|}=\sign_\Psi(|l|,b)=\sign_{\check\Psi}(|l|,a)$.
		\end{description}
		Thus, in either case $\sigma$ satisfies clause $a$.
	\end{proof}

	In perfect analogy to the above let $\tilde\Psi$ be the formula obtained from $\hat\Psi$ by removing all variables $x\in V(\hat\Psi)$ such that $\{x,\neg x\}\cap\cN(\Phi,e)\neq\emptyset$, along with their adjacent clauses.

	\begin{claim}\label{claim_forced_sat_Phi}
		For any $\tilde\sigma\in S(\tilde\Psi)$ there exists $\sigma\in S(\hat\Phi)$ such that $\sigma_x=\tilde\sigma_x$ for all $x\in V(\tilde\Psi)$.
	\end{claim}
	\begin{proof}
		Let $\check\Phi$ be a CNF with variables $V(\check\Phi)=V(\hat\Phi)\setminus V(\tilde\Psi)$.
		Include in $\check\Phi$ all $a\in F(\hat\Phi)$ with $\partial a\subset V(\check\Phi)$.
		Moreover, for every $a\in F(\hat\Phi)$ that contains exactly one literal $l$ with $|l|\in V(\check\Psi)$ add  $l$ as a clause to $\check\Phi$.
		As in the proof of Claim~\ref{claim_forced_sat} it suffices to construct a satisfying assignment of $\check\Phi$.
		Due to \eqref{eqLphil} the same argument as in the proof of Claim~\ref{claim_forced_sat} extends.
	\end{proof}

	\begin{proof}[Proof of \Lem~\ref{lem_forced_add}]
		We use Claim~\ref{claim_forced_sat}	to prove that $Z(\hat\Phi)\leq2^{|\cN(\Phi,e)|}Z(\hat\Psi)$; similar reasoning based on Claim~\ref{claim_forced_sat_Phi} yields the reverse bound.
		To show the desired bound split a satisfying assignment $\sigma\in S(\hat\Phi)$ up into two parts $\tilde\sigma=(\sigma_x)_{\{x, \neg x\} \cap \cN(\Phi,e) = \emptyset}$, $\check\sigma=(\sigma_x)_{\{x, \neg x\} \cap \cN(\Phi,e) \not= \emptyset}$.
		Claim~\ref{claim_forced_sat} shows that the number of possible first parts $\tilde\sigma$ for $\sigma\in S(\hat\Phi)$ is bounded by $Z(\hat\Psi)$, because every $\tilde\sigma$ extends to a satisfying assignment of $\hat\Psi$.
		Moreover, the total number of possible second parts is bounded by $2^{|\cN(\Phi,e)|}$.
	\end{proof}

	\subsection{A tail bound}\label{sec_tail}

	As a next step we are going to derive a bound on the r.h.s.\ of~\eqref{eqlem_forced_add} on random formulas.
	More specifically, obtain the formula $\PHI'$ from $\PHI$ by deleting the last clause $\va_m$.
	Let $\vN'=|\cN(\PHI',\va_m)|$.

	\begin{lemma}\label{lem_tail}
		There exists $c=c(d)>0$ such that for all $t>c$ we have
		\begin{align*}
			\pr\brk{\vN'>t^2}\leq c\exp(-t/c).
		\end{align*}
	\end{lemma}

	As a first step we are going to estimate the size of the set $\cN(\PHI',x_1)$ that contains all literals $l$ such that $\cL(\PHI',l)\cap\{x_1,\neg x_1\}\neq\emptyset$.

	\begin{claim}\label{claim_taila}
		There exists $c_1=c_1(d)>0$ such that for all $t>c_1$ we have $\pr\brk{|\cN(\PHI',x_1)|>t}\leq c_1\exp(-t/c_1).$
	\end{claim}
	\begin{proof}
		We use a classical branching process argument.
		Let $\cR$ be the set of literals $l$ such that $x_1\in\cL(\PHI',l)$.
		By symmetry it suffices to bound $|\cR|$.

		For every $l\in\cR$ there exists an alternating sequence $l=l_0,a_1,l_1,a_2,\ldots,l_k=x_1$ of literals and clauses such that $a_i\equiv\neg l_{i-1}\vee l_{i}$.
		Flipping the negations along this sequence yields a reverse sequence
		$l_0'=\neg x_1=\neg l_k,a_1'=a_k,l_{1}'=\neg l_{k-1},\ldots,l_k'=\neg l$ such that $a_i'\equiv \neg l_{i-1}'\vee l_i'$.
		Hence, $\cR$ is precisely the set of literals $l$ that are reachable from $x_1$ via such an alternating sequence $l_0',a_1',\ldots,l_k'$.
		Furthermore, for any literal $l$ the expected number of clauses $\va_i$ such that $\va_i\equiv l\vee l'$ for some other literal $l'$ equals $m/2n\sim d/2$.
		Therefore, $|\cR|$ is stochastically dominated by the progeny of a branching process with offspring $\Po(d/2)$.
		Standard branching process tail bounds therefore yield the desired bound on $|\cR|$.
	\end{proof}

	\begin{claim}\label{claim_tailb}
		There exists $c_2=c_2(d)>0$ such that for all $t>c_2$ and for every literal $l\neq x_1$ we have
		\begin{align*}
			\pr\brk{|\cL(\PHI',l)|>t\mid x_1\in\cL(\PHI',l)}\leq c_2\exp(-t/c_2).
		\end{align*}
	\end{claim}
	\begin{proof}
		We combine a branching process argument with Bayes' formula.
		Specifically, because the formula $\PHI'$ is random, the set $\cL(\PHI',l)\setminus\{l\}$ is random given its size.
		Hence, for an integer $\ell$ we have
		\begin{align}\label{eq_claim_tailb_1}
			\pr\brk{x_1\in\cL(\PHI',l)\mid|\cL(\PHI',l)|=\ell}&=\frac{\ell-1}{2n-1}.
		\end{align}
		Furthermore, the size $|\cL(\PHI',l)|$ is stochastically dominated by the progeny of a branching process with offspring $\Po(d/2)$.
		Therefore, there exists $c_2'=c_2'(d)>0$ such that for all $t>c_2'$ we have
		\begin{align}\label{eq_claim_tailb_2}
			\pr\brk{|\cL(\PHI',l)|>t}\leq c_2'\exp(-t/c_2').
		\end{align}
		Moreover, for any $d>0$ there exists $c_2''=c_2''(d)>0$ such that
		\begin{align}\label{eq_claim_tailb_3}
			\pr\brk{x_1\in\cL(\PHI',l)}\geq c_2''/n.
		\end{align}
		Hence, combining \eqref{eq_claim_tailb_1}--\eqref{eq_claim_tailb_3} with Bayes' rule, we obtain for $\ell>c_2'$,
		\begin{align*}
			\pr\brk{|\cL(\PHI',l)|=\ell\mid x_1\in\cL(\PHI',l)}
			&\leq\frac{\pr\brk{x_1\in\cL(\PHI',\ell)\mid|\cL(\PHI',l)|=\ell}\pr\brk{|\cL(\PHI',l)|=\ell}}{\pr\brk{x_1\in\cL(\PHI',l)}}\leq\frac{c_2'}{c_2''}\ell\exp(-\ell/c_2'),
		\end{align*}
		which implies the assertion.
	\end{proof}

	\begin{proof}[Proof of \Lem~\ref{lem_tail}]
		Let $\fR(\ell)$ be the event that there exists $l\in\cN(\PHI',x_1)\setminus\{x_1\}$ such that $|\cL(\PHI',l)|>\ell$.
		Claim~\ref{claim_taila} implies that there exists $c_3=c_3(d)>0$ such that
		\begin{align}\label{eq_lem_tail_1}
			\pr\brk{x_1\in\cL(\PHI',l)}\leq c_3/n.
		\end{align}
		Hence, by Claim~\ref{claim_tailb}, \eqref{eq_lem_tail_1} and the union bound,
		\begin{align}\label{eq_lem_tail_2}
			\pr\brk{\fR(\ell)}\leq\sum_{l\neq x_1}\pr\brk{x_1\in\cL(\PHI',l),\,|\cL(\PHI',l)|>\ell}\leq 2c_2c_3\exp(-\ell/c_2).
		\end{align}
		Furthermore, Claim~\ref{claim_taila} shows that
		\begin{align}\label{eq_lem_tail_3}
			\pr\brk{\cN(\PHI',x_1)\setminus\{x_1\}>\ell}&\leq c_1\exp(-\ell/c_1).
		\end{align}
		Combining \eqref{eq_lem_tail_2} and \eqref{eq_lem_tail_3}, we obtain
		\begin{align}\label{eq_lem_tail_4}
			\pr\brk{\sum_{l\in\cN(\PHI',x_1)}|\cL(\PHI',l)|>\ell^2}&
			\leq\pr\brk{\fR(\ell)}+\pr\brk{\cN(\PHI',x_1)\setminus\{x_1\}>\ell}
			\leq c_1\exp(-\ell/c_1)+2c_2c_3\exp(-\ell/c_2).
		\end{align}
		By symmetry the same bound holds with $x_1$ replaced by $\neg x_1$.
		Therefore, the assertion follows from \eqref{eq_lem_tail_4} and the union bound.
	\end{proof}

	\subsection{The squared martingale difference}\label{sec_comb_proc}

	We derive a combinatorial formula for the squared martingale differences $\vX_i^2$.
	Let
	\begin{align*}
		\vDelta(M)&=\log\bcfr{Z(\hPHI_1(M,m-M))}{Z(\hPHI_1(M-1,m-M))} \cdot \log\bcfr{Z(\hPHI_2(M,m-M))}{Z(\hPHI_2(M-1,m-M))},\\
		\vDelta'(M)&=\log\bcfr{Z(\hPHI_1(M-1,m-M+1))}{Z(\hPHI_1(M-1,m-M))} \cdot \log\bcfr{Z(\hPHI_2(M-1,m-M+1))}{Z(\hPHI_2(M-1,m-M))} ,\\
		\vDelta''(M)&=\log\bcfr{Z(\hPHI_1(M,m-M))}{Z(\hPHI_1(M-1,m-M))} \cdot \log\bcfr{Z(\hPHI_2(M-1,m-M+1))}{Z(\hPHI_2(M-1,m-M))} .
	\end{align*}

	\begin{lemma}\label{lem_comb_proc}
		We have $m\vX_M^2=\ex\brk{\vDelta(M)+\vDelta(M)'-2\vDelta''(M)\mid\fF_M}.$
	\end{lemma}
	\begin{proof}
		This follows from a direct computation.
	\end{proof}

	\begin{proof}[Proof of \Lem~\ref{lem_comb_rgb}]
		\Lem~\ref{lem_comb_rgb} is an immediate consequence of \Lem~\ref{lem_comb_proc}.
	\end{proof}

	\subsection{An $L^2$-bound}\label{sec_comb_l2}
	The following $L^2$-bound will enable us to deal with error terms.

	\begin{lemma}\label{lem_comb_l2}
		Uniformly for all $1\leq M\leq m$ we have $\ex\brk{\vec\Delta(M)^2+\vec\Delta'(M)^2+\vec\Delta''(M)^2}=O(1)$.
	\end{lemma}
	\begin{proof}
		We will bound $\ex[\vec\Delta(M)^2]$; the bounds on the other two terms follow analogously.
		Invoking the Cauchy-Schwarz inequality, we obtain
		\begin{align}\nonumber
			\ex\brk{\vec\Delta(M)^2}&=\ex\brk{\log^2\bcfr{Z(\hPHI_1(M,m-M))}{Z(\hPHI_1(M-1,m-M))} \cdot \log^2\bcfr{Z(\hPHI_2(M,m-M))}{Z(\hPHI_2(M-1,m-M))}}\\
			&\leq\ex\brk{\log^4\bcfr{Z(\hPHI_1(M,m-M))}{Z(\hPHI_1(M-1,m-M))}}^{1/2}\ex\brk{\log^4\bcfr{Z(\hPHI_2(M,m-M))}{Z(\hPHI_2(M-1,m-M))}}^{1/2}\nonumber\\
			&=\ex\brk{\log^4\bcfr{Z(\hPHI_1(M,m-M))}{Z(\hPHI_1(M-1,m-M))}}.\label{eq_lem_comb_l2_1}
		\end{align}
		Furthermore, the random formula $\PHI_1(M,m-M)$ is obtained from $\PHI_1(M-1,m-M)$ by adding a single random clause $\va_M$, which is independent of $\PHI_1(M-1,m-M)$.
		Therefore, \Lem~\ref{lem_forced_add} implies that
		\begin{align}\label{eq_lem_comb_l2_2}
			\log\bcfr{Z(\hPHI_1(M,m-M))}{Z(\hPHI_1(M-1,m-M))}&\leq|\cN(\PHI_1(M-1,m-M),\va_M)|\log2.
		\end{align}
		Moreover, since $|\cN(\PHI_1(M-1,m-M),\va_M)|$ has the same distribution as the random variable $\vN'$ from \Lem~\ref{lem_tail}, we obtain
		\begin{align}\label{eq_lem_comb_l2_3}
			\ex\brk{|\cN(\PHI_1(M-1,m-M),\va_M)|^4}&=O(1)
		\end{align}
		uniformly for all $M$.
		Finally, the assertion follows from \eqref{eq_lem_comb_l2_1}--\eqref{eq_lem_comb_l2_3}.
	\end{proof}

	To facilitate the following steps we introduce trunacted versions of $\vDelta(M),\vDelta'(M),\vDelta''(M)$: for $B>0$ and $x>0$ let
	\begin{align*}
		\Lambda_B(x)&=\begin{cases}
			-B&\mbox{if $\log(x)<-B$},\\
			B&\mbox{if $\log(x)>B$},\\
			\log(x)&\mbox{otherwise}.
		\end{cases}
	\end{align*}
	Further, let
	\begin{align*}
		\vDelta_B(M)&=\Lambda_B\bc{\frac{Z(\hPHI_1(M,m-M))}{Z(\hPHI_1(M-1,m-M))}} \cdot \Lambda_B\bc{\frac{Z(\hPHI_2(M,m-M))}{Z(\hPHI_2(M-1,m-M))}},\\
		\vDelta'_B(M)&=\Lambda_B\bc{\frac{Z(\hPHI_1(M-1,m-M+1))}{Z(\hPHI_1(M-1,m-M))}} \cdot \Lambda_B\bc{\frac{Z(\hPHI_2(M-1,m-M+1))}{Z(\hPHI_2(M-1,m-M))}} ,\\
		\vDelta''_B(M)&=\Lambda_B\bc{\frac{Z(\hPHI_1(M,m-M))}{Z(\hPHI_1(M-1,m-M))}} \cdot \Lambda_B\bc{\frac{Z(\hPHI_2(M-1,m-M+1))}{Z(\hPHI_2(M-1,m-M))}}.
	\end{align*}
	Combining \Lem~\ref{lem_comb_l2} with the Cauchy-Schwarz inequality, we obtain the following.

	\begin{corollary}\label{cor_comb_l2}
		For any $\eps>0$ there exists $B>0$ such that for all $1\leq M\leq m$ we have
		\begin{align*}
			\ex\abs{\vDelta(M)-\vDelta_B(M)} +\ex\abs{\vDelta'(M)-\vDelta'_B(M)}+ \ex\abs{\vDelta''(M)-\vDelta''_B(M)}<\eps.
		\end{align*}
	\end{corollary}

	\subsection{The variance process}\label{sec_var_proc}
	In light of \Lem~\ref{lem_comb_proc}, to prove \Prop~\ref{prop_varproc} we need to show that
	\begin{align*}
		\frac1m\sum_{M=1}^m\ex\brk{\vDelta(M)+\vDelta'(M)-2\vDelta''(M)\mid\fF_{M}}&\to\eta(d)^2\qquad\mbox{in probability.}
	\end{align*}
	To this end we divide the above sum up into batches $\bar\Sigma(L,L')=\Sigma(L,L')+\Sigma'(L,L')-2\Sigma''(L,L')$, where
	\begin{align*}
		\Sigma(L,L')&=\frac1{L'-L}\sum_{M=L}^{L'-1}\ex\brk{\vDelta(M)\mid\fF_{M}},\\
		\Sigma'(L,L')&=\frac1{L'-L}\sum_{M=L}^{L'-1}\ex\brk{\vDelta'(M)\mid\fF_{M}},\\
		\Sigma''(L,L')&=\frac1{L'-L}\sum_{M=L}^{L'-1}\ex\brk{\vDelta''(M)\mid\fF_{M}}.
	\end{align*}
	Then for any sequence $1=L_0<\cdots<L_k=m$ we have
	\begin{align*}
		\frac1n\sum_{M=1}^m\ex\brk{\vDelta(M)+\vDelta'(M)-2\vDelta''(M)\mid\fF_{M}}&=\sum_{i=1}^{k}\frac{L_i-L_{i-1}}{n}\bar\Sigma(L_{i-1},L_i).
	\end{align*}

	The following lemma is the centerpiece of the proof.

	\begin{lemma}\label{lem_varproc}
		For any $\eps>0$ there exists $\omega>0$ such that uniformly for all $1\leq L<L'\leq m$ with $\omega\leq L'-L\leq2\omega$ we have
		\begin{align*}
			\ex\abs{\Sigma(L,L')-\cB^\tensor_{d,t}(\pi_{d,t})}+
			\ex\abs{\Sigma'(L,L')-\cB^\tensor_{d,0}(\pi_{d,0})}+
			\ex\abs{\Sigma''(L,L')-\cB^\tensor_{d,0}(\pi_{d,0})}
			&<\eps+o(1),&&\mbox{where }t=L/m.
		\end{align*}
	\end{lemma}

	We will carry out the details for the first term $\ex|\Sigma(L,L')-\cB^\tensor_{d,t}(\pi_{d,t})|$, which is the most delicate; similar but slightly simpler steps yield the other two estimates.
	We begin by replacing $\vDelta(M)$ by its truncated version $\vDelta_B(M)$.
	Accordingly, let
	\begin{align*}
		\Sigma_B(L,L')&=\frac1{L'-L}\sum_{M=L}^{L'-1}\ex\brk{\vDelta_B(M)\mid\fF_{M}},\\
		\Sigma'_B(L,L')&=\frac1{L'-L}\sum_{M=L}^{L'-1}\ex\brk{\vDelta'_B(M)\mid\fF_{M}},\\
		\Sigma''_B(L,L')&=\frac1{L'-L}\sum_{M=L}^{L'-1}\ex\brk{\vDelta''_B(M)\mid\fF_{M}}.
	\end{align*}

	\begin{claim}\label{claim_varproc_a}
		For any $\eps>0$ there exists $B_0>0$ such that for all $B>B_0$ and all $L,L'>0$ we have $$\ex\abs{\Sigma(L,L')-\Sigma_B(L,L')}<\eps+o(1).$$
	\end{claim}
	\begin{proof}
		This is an immediate consequence of \Cor~\ref{cor_comb_l2}.
	\end{proof}

	We proceed to relate the change in the pruned partition function to the marginal distribution of the truth values of the variables of the additional clause $\va_M$.

	\begin{claim}\label{claim_varproc_c}
		Let $B>0$.
		\Whp\ we have
		\begin{align*}
			\Lambda_B\bcfr{Z(\hPHI_h(M,m-M))}{Z(\hPHI_h(M-1,m-M))}&=\Lambda_B\bc{1-\prod_{y\in\partial\va_M}\mu_{\hPHI_h(M-1,m-M)}\bc{\SIGMA_{y}\neq\sign(y,\va_M)}}+o(1)&&(h=1,2).
		\end{align*}
	\end{claim}
	\begin{proof}
		Since the function $\Lambda_B$ is bounded and continuous, this follows from \Prop~\ref{lem_rs}.
	\end{proof}

	A combinatorial interpretation of $\Sigma(L,L')$ is that the sum gauges the cumulative effect of adding a total of $L'-L$ `shared' clauses, one after the other.
	Claim~\ref{claim_varproc_c} expresses the effect of adding a shared clause in terms of the marginals of the formula $\hPHI_h(M-1,m-M)$.
	So long as the total number $L'-L$ of clauses added is not too large, we may expect that this marginal distribution does not shift all to much as we add clauses one by one.
	This is what the following claim verifies.

	\begin{claim}\label{claim_varproc_d}
		Let $t=M/m$.
		If $L'-L=O(1)$, then \whp\ we have $$\sum_{M=L}^{L'-1}W_1(\pi_{\hat\PHI_1(M-1,m-M),\hat\PHI_2(M-1,m-M)},\pi_{d,t})=o(1).$$
	\end{claim}
	\begin{proof}
		This follows from \Cor~\ref{lem_empirical}.
	\end{proof}

	As a next step we truncate the functional $\cB_{d,t}^\tensor$ from \eqref{eqfunctional}.
	Hence, for $B>0$ let
	\begin{align}\nonumber
		\cB^\tensor_{B,d,t}(\pi)&=
		\ex\Big[\Lambda_B\bc{1-(\vecone\{\vr_{-1,1}=-1\}+\vr_{-1,1}\MU_{\pi,-1,1,1})(\vecone\{\vr_{-1,2}=-1\}+\vr_{-1,2}\MU_{\pi,-1,2,1})}\\
		&\qquad\qquad			\Lambda_B\bc{1-(\vecone\{\vr_{-1,1}=-1\}+\vr_{-1,1}\MU_{\pi,-1,1,2})(\vecone\{\vr_{-1,2}=-1\}+\vr_{-1,2}\MU_{\pi,-1,2,2})}\Big].
		\label{eqfunctional_trunc}
	\end{align}

	\begin{claim}\label{claim_varproc_b}
		For any $\eps>0$ there exists $B_0>0$ such that for all $B>B_0$ and all $t\in[0,1]$ we have $$\abs{\cB^\tensor_{d,t}(\pi_{d,t})-\cB^\tensor_{B,d,t}(\pi_{d,t})}<\eps.$$
	\end{claim}
	\begin{proof}
		Since $\cB^\tensor_{B,d,t}(\pi_{d,t}) \uparrow \cB^\tensor_{d,t}(\pi_{d,t})$ as $B \to \infty$, this follows from \Prop~\ref{prop_noela}.
	\end{proof}

	\begin{proof}[Proof of \Lem~\ref{lem_varproc}]
		\Lem~\ref{lem_bethe_fin} ensures that $\cB_{d,t}^\tensor(\pi_{d,t})<\infty$ for all $t$.
		Moreover, Claims~\ref{claim_varproc_a} and~\ref{claim_varproc_b} imply that we just need to show that for large $B>0$,
		\begin{align}\label{eq_lem_varproc_1}
			\ex\abs{\Sigma_B(L,L')-\cB^\tensor_{B,d,t}(\pi_{d,t})}<\eps.
		\end{align}
		Let $t=L/m$ and let $(\vS_M)_{L\leq M<L'}$ be independent copies of the random variable
		\begin{align*}
			\prod_{h=1}^2\Lambda_B&\bc{1-(\vecone\{\vr_{-1,1}=-1\}+\vr_{-1,1}\MU_{\pi_{d,t},-1,1,h})(\vecone\{\vr_{-1,2}=-1\}+\vr_{-1,2}\MU_{\pi_{d,t},-1,2,h})}.
		\end{align*}
		Furthermore, let
		\begin{align*}
			\vD_M&=\prod_{h=1}^2\Lambda_B\bcfr{Z(\hPHI_h(M,m-M))}{Z(\hPHI_h(M-1,m-M))}
		\end{align*}
		Then Claims~\ref{claim_varproc_c}--~\ref{claim_varproc_d} show that
		\begin{align}\label{eq_lem_varproc_2}
			\ex\brk{W_1\bc{\sum_{M=L}^{L-1}\vD_M,\sum_{M=L}^{L-1}\vS_M}}&=o(1).
		\end{align}
		Finally, since $\sum_{M=L}^{L-1}\vS_M$ is a sum of bounded independent random variables, \eqref{eq_lem_varproc_1} follows from \eqref{eq_lem_varproc_2} and the strong law of large numbers.
	\end{proof}

	\begin{proof}[Proof of \Prop~\ref{prop_varproc}]
		The second equality in~\eqref{eqmyhh} follows from \Cor~\ref{cor_comb_l2}, \Lem~\ref{lem_varproc} and the triangle inequality.
		Thus, we are left to verify the first condition.
		Since \Lem~\ref{lem_comb_proc} shows that
		\begin{align*}
			\vX_M^2=\frac1m\ex\brk{\vDelta(M)+\vDelta(M)'-2\vDelta''(M)\mid\fF_M},
		\end{align*}
		it suffices to prove that
		\begin{align}\label{eq_thm_clt_1}
			\ex\max_{1\leq M\leq m}\vDelta(M)^2+\ex\max_{1\leq M\leq m}\vDelta'(M)^2+\ex\max_{1\leq M\leq m}\vDelta''(M)^2&=o(m).
		\end{align}
		We will bound the first expectation; similar arguments apply to the others.

		Retracing the steps of the proof of \Lem~\ref{lem_comb_l2}, we write
		\begin{align}\nonumber
			\vec\Delta(M)^2&=\log^2\bcfr{Z(\hPHI_1(M,m-M))}{Z(\hPHI_1(M-1,m-M))} \cdot \log^2\bcfr{Z(\hPHI_2(M,m-M))}{Z(\hPHI_2(M-1,m-M))}.
		\end{align}
		Since $\PHI_1(M,m-M)$ is obtained from $\PHI_1(M-1,m-M)$ by adding the random clause $\va_M$, \Lem~\ref{lem_forced_add} implies that
		\begin{align}\label{eq_thm_clt_2}
			\log\bcfr{Z(\hPHI_h(M,m-M))}{Z(\hPHI_h(M-1,m-M))}&\leq|\cN(\PHI_h(M-1,m-M),\va_M)|\log2&&(h=1,2).
		\end{align}
		As $|\cN(\PHI_h(M-1,m-M),\va_M)|$ has the same distribution as $\vN'$ from \Lem~\ref{lem_tail}, we obtain $$\ex\brk{|\cN(\PHI_1(M-1,m-M),\va_M)|^4}=O(1)$$ uniformly for all $M$.
		Therefore, Markov's inequality implies
		\begin{align}\label{eq_thm_clt_3}
			\pr\brk{|\cN(\PHI_1(M-1,m-M),\va_M)|>m^{1/3}}&\leq O(m^{-4/3}).
		\end{align}
		Finally, \eqref{eq_thm_clt_1} follows from \eqref{eq_thm_clt_2}, \eqref{eq_thm_clt_3} and the union bound.
		%
		%%%%%%%%%%%%%%%% Pavel's proof of the third condition.
		%\subsection{Condition (3)}
		%
		%Define $S_{i}^{(n)}$ and $X_i^{(n)}$ the same way as in \ref{sec:intro}. Then note that
		%\[
		%\Erw\Big[ \max_{1 \leq i \leq m} (X_i^{(n)})^2 \Big] \leq  \Erw\Big[ \sum_{1 \leq i \leq m} (X_i^{(n)})^2 \Big] =\Erw\Big[ \left(\sum_{1 \leq i \leq m} S_{i+1}^{(n)} - S_i^{(n)} \right)^2 \Big] = \Erw{\left(S_m^{(n)}\right)^2} =
		%\]
		%\[
		%= \Var\left[ {\log Z_m^{(n)} \over \sqrt{n}} - {\Erw \log Z_m^{(n)} \over \sqrt{n}}\right] ={1 \over n} \Var \Big[\log Z_m^{(n)}\Big] = {1 \over n} \Var \log Z,
		%\]
		%which is bounded in $n$.
		%
	\end{proof}

	As a final preparation towards the proof of \Prop~\ref{cor_var} we need a lower bound on $\log Z(\hPHI)$.

	\begin{lemma}\label{lem_var_pos}
		We have $\Var(\log Z(\hPHI))=\Omega(n)$.
	\end{lemma}
	\begin{proof}
		Let $\fC$ be the set of isolated sub-formulas of $\hPHI$ with precisely three clauses and three variables that are acyclic and whose unique variable of degree two appears with the same sign in both its adjacent clauses.
		Moreover, let $\vC'$ be the set of isolated sub-formulas of $\hPHI$ with precisely three clauses and three variables that are acyclic such that the unique variable of degree two appears with two different signs in its adjacent clauses.
		Then $\ex|\fC|=\ex|\fC'|=\Omega(n)$ and \whp\ we have
		\begin{align}\label{eqlem_var_pos1}
			\Var(|\fC|\mid|\fC|+|\fC'|)=\Var(|\fC'|\mid|\fC|+|\fC'|)=\Omega(n).
		\end{align}
		Additionally, for each sub-formula $\cC\in\fC$ we have $Z(\cC)=5$, while for $\cC'\in\fC'$ we have $Z(\vC')=4$.
		Since with the sum ranging over the connected components $\cC$ of $\hPHI$ we have $\log Z(\hPHI)=\sum_{\cC}\log Z(\cC)$, the assertion follows from~\eqref{eqlem_var_pos1}.
	\end{proof}

	\begin{proof}[Proof of \Cor~\ref{cor_var}]
		The corollary is an immediate consequence of \Lem~\ref{lem_comb_rgb}, \Lem~\ref{lem_bethe_fin},\\ \Cor~\ref{cor_comb_l2},
		\Lem~\ref{lem_varproc} and \Lem~\ref{lem_var_pos}.
	\end{proof}

	\section{Proof of \Thm~\ref{thm_clt}}%{The conditions of Theorem \ref{thm_hh} imply those of \cite[Theorem 3.2]{HH}}}
\label{sec_hh}

%\aco{merge in the following!}
%Finally, \Thm~\ref{thm_clt} follows from \Cor~\ref{cor_var} and \Prop~\ref{prop_varproc}.

\noindent
We derive \Thm~\ref{thm_clt} from the following general martingale central limit theorem, which is a special case of \cite[Theorem 3.2]{HH} (see also the subsequent remark there).
%Recall that a martingale is a sequence $(\vZ_i,\fF_i)_{0\leq i\leq m}$ of (real-valued, integrable) random variables $\vZ_i$ and nested $\sigma$-algebras $\fF_i$ such that $\vZ_i$ is $\fF_i$-measurable and such that the conditional expectations satisfy $\vZ_i=\ex\brk{\vZ_{i+1}\mid\fF_i}$.
%Towards the proof of \Thm~\ref{thm_clt} we will consider martingales whose individual steps are indexed by the clauses of the random formula $\PHI$.
%But of course we can hope to prove convergence to a Gaussian only in the limit $n\to\infty$.
%Hence, strictly speaking we actually work with a {\em martingale array} $(\vZ_{n,i},\fF_{n,i})_{n\geq1,0\leq i\leq m_n}$ with $m_n\sim dn/2$.
%This means that $(\vZ_{n,i},\fF_{n,i})_{0\leq i\leq m_n}$ is a martingale for every $n$.

%\nm{The theorem is a corollary of something in \cite{Eagleson}, right? Looking at the article, I do not see something that looks like this, and in the Appendix we use Hall and Heyde right now. Maybe align this better?}

%The first condition in \eqref{eqhh2} asks that all martingale differences $\vX_{n,i}$ be `small'.
%Thanks to pruning, this condition is easily seen to be satisfied.
%The second condition \eqref{eqhh2} asks that the sum of squared martingale differences concentrate on a deterministic value $\eta^2$.
%%The techniques that we developed towards the variance computation will prove vital to check \eqref{eqhh2}.

%The following formulation is a combination of Theorem 3.2 from \cite{HH} and subsequent remarks on simplifications in the case of a deterministic variance.

\begin{theorem}[{\cite[Theorem 3.2]{HH}}]\label{thm_hh_original}
	Let $(\vZ_{n,i}, \fF_{n,i})_{0\leq i \leq m_n, n \geq 1}$ be a zero-mean, square-integrable martingale array with differences $\vX_{n,i} = \vZ_{n,i} - \vZ_{n,i-1}$ for $1 \leq i \leq m_n$. Assume that there exists a constant $\eta^2$ such that
	\begin{align}
		\lim_{n \to \infty}\max_{1 \leq i \leq m_n} |\vX_{n,i}| = 0 \qquad &\text{in probability}, \label{eq_hh1} \\
		\lim_{n \to \infty}\sum_{i=1}^{m_n}\vX_{n,i}^2 = \eta^2 \qquad &\text{in probability},\label{eq_hh2} \\
		\ex\brk{\max_{1 \leq i \leq m_n} \vX_{n,i}^2} \qquad &\text{is bounded in } n. \label{eq_hh3} %\\
		%\fF_{n,i} \subseteq \fF_{n+1,i}, \qquad &0 \leq i \leq m_n, n \geq 1. \label{eq_hh4}
	\end{align}
	%Moreover, suppose that the sub-$\sigma$-fields satisfy $\fF_{n,0}\subset \fF_{n,1}\cdots\subset\fF_{n,m_n}$ for all $n \geq 0$.
	Then $\vZ_{n,m_n}$ converges in distribution to a Gaussian distribution with mean zero and variance $\eta^2$.% whose characteristic function $\myphi$ is given by
	%	\begin{align}\label{eqmixture}
		%     \myphi(t)=\Erw\brk{\exp\bc{-\frac{1}{2}t^2\eta^2}}.
		%\end{align}
		%centered Gaussian random variable with variance $\eta^2$.
	\end{theorem}

	\begin{proof}[Proof of \Thm~\ref{thm_clt}]
		We apply \Thm~\ref{thm_hh_original} to the filtration $(\fF_{n,M})_{0\leq M\leq m_n}$ from \Sec~\ref{sec_hh_apply} and to the Doob martingale $(\vZ_{n,M} - \Erw\brk{\vZ_{n,M}})_M$ from \eqref{eqZstab1}. This is zero-mean by construction and square-integrable, as $\log Z(\hat\PHI)$ is non-negative and bounded above by $n$. %Also the sub-$\sigma$-fields satisfy $\fF_{n,0}\subset \fF_{n,1}\cdots\subset\fF_{n,m_n}$ for all $n \geq 0$.
		Let $\vX_{n,M}=\vZ_{n,M}-\vZ_{n,M-1}$ be the martingale differences.
		\Prop~\ref{prop_varproc} immediately implies conditions \eqref{eq_hh1}--\eqref{eq_hh2} of \Thm~\ref{thm_hh_original} since $L^1$-convergence implies convergence in probability. Condition \eqref{eq_hh3} also follows from \Prop~\ref{prop_varproc} by observing that
		\begin{align*} %\label{eq_hh3_n2}
			\ex\brk{\max_{1 \leq M \leq m_n} \vX_{n,M}^2} &\leq  \Erw\brk{\sum_{M=1}^{m_n}\vX_{n,M}^2} \leq \Erw\abs{\sum_{M=1}^{m_n}\vX_{n,M}^2 - \eta(d)^2} + \eta(d)^2.
			%	\sum_{i=0}^{m_n}\Erw\brk{\vX_{n,i}^2\vecone\cbc{\abs{\vX_{n,i}}>\varepsilon}}&\leq \Erw\brk{\sum_{i=0}^{m_n}\vX_{n,i}^2\vecone\cbc{\max_{1\leq j\leq m_n}\abs{\vX_{n,j}}>\varepsilon}}\\
			%	\nonumber&\leq \Erw\brk{\eta^2\vecone\cbc{\max_{1\leq j\leq m_n}\abs{\vX_{n,j}}>\varepsilon}}+\Erw\brk{\abs{\sum_{i=0}^{m_n}\vX_{n,i}^2-\eta^2}}\\
			%	&=\eta^2\mathbb{P}\bc{\max_{1\leq j\leq m_n}\abs{\vX_{n,j}}>\varepsilon}+\Erw\brk{\abs{\sum_{i=0}^{m_n}\vX_{n,i}^2-\eta^2}}\\
			%	 &\nonumber\leq \frac{\eta^2}{\varepsilon}\Erw\brk{\max_{1\leq j\leq m_n}\abs{\vX_{n,j}}}+\Erw\brk{\abs{\sum_{i=0}^{m_n}\vX_{n,i}^2-\eta^2}}.\nonumber
		\end{align*}
		Furthermore, \Lem~\ref{lem_bethe_fin} guarantees that $\eta(d)<\infty$, while \Cor~\ref{cor_var} shows that $\eta(d)>0$.
		Thus, the assertion follows from \Thm~\ref{thm_hh_original}.% with a non-random $\eta(d)>0$, so that the limiting distribution is Gaussian with variance $\eta(d)^2$.
	\end{proof}

	\subsection*{Acknowledgement}
	Amin Coja-Oghlan's research is supported by DFG CO 646/3, DFG CO 646/5 and DFG CO 646/6.  Pavel Zakharov's research is supported by DFG CO 646/6.  Haodong Zhu's research is supported by the European Union's Horizon 2020 research and innovation programme under the Marie Sk\l odowska-Curie grant agreement no.~945045, and by the NWO Gravitation project NETWORKS under grant no.~024.002.003.  Noela M\"uller's research is supported by the NWO Gravitation project NETWORKS under grant no.~024.002.003.
	We thank Nicloa Kistler for helpful discussions, and particularly for bringing~\cite{ChenDeyPanchenko} to our attention.

\end{document}